\documentclass{article}

\newcommand\ie{\textit{i.e.,}}
\newcommand\eg{\textit{e.g.,}}

\newcommand{\norm}[1]{\left\lVert#1\right\rVert}
\newcommand{\red}{\textcolor{red}}

\newcommand{\beq}{\begin{equation}}
\newcommand{\eeq}{\end{equation}}
\newcommand{\beqnn}{\begin{equation*}}
\newcommand{\eeqnn}{\end{equation*}}
\newcommand{\beqy}{\begin{eqnarray}}
\newcommand{\eeqy}{\end{eqnarray}}
\newcommand{\beqynn}{\begin{eqnarray*}}
\newcommand{\eeqynn}{\end{eqnarray*}}
\newcommand{\bit}{\begin{itemize}}
\newcommand{\eit}{\end{itemize}}
\newcommand{\ben}{\begin{enumerate}}
\newcommand{\een}{\end{enumerate}}
\newcommand{\bex}{\begin{example}}
\newcommand{\eex}{\end{example}}
\newcommand{\trace}{\mathrm{trace}}

\newcommand{\balg}[1]{\begin{algorithm} \caption{#1}}
\newcommand{\ealg}{\end{algorithm}}

\newcommand{\balgc}{\begin{algorithmic}[1]}
\newcommand{\ealgc}{\end{algorithmic}}

\newcommand{\bary}{\begin{array}}
\newcommand{\eary}{\end{array}}
\newcommand{\bmx}{\begin{bmatrix}}
\newcommand{\emx}{\end{bmatrix}}
\newcommand{\bsmx}{\left[\begin{smallmatrix}}
\newcommand{\esmx}{\end{smallmatrix}\right]}
\newcommand{\bmxc}[1]{\left[\begin{array}{@{}#1@{}}}
\newcommand{\emxc}{\end{array}\right]}
\newcommand{\bcn}{\begin{center}}
\newcommand{\ecn}{\end{center}}

\newcommand{\Rbb}{{\mathbb{R}}}

\providecommand{\norm}[1]{\lVert#1\rVert}

\newenvironment{theorem}[2][Theorem]{\begin{trivlist}
		\item[\hskip \labelsep {\bfseries #1}\hskip \labelsep {\bfseries #2.}]}{\end{trivlist}}
\newenvironment{lemma}[2][Lemma]{\begin{trivlist}
		\item[\hskip \labelsep {\bfseries #1}\hskip \labelsep {\bfseries #2.}]}{\end{trivlist}}

\newenvironment{proposition}[2][Proposition]{\begin{trivlist}
		\item[\hskip \labelsep {\bfseries #1}\hskip \labelsep {\bfseries #2.}]}{\end{trivlist}}

\usepackage{amssymb,bm,amsthm}

\usepackage{amsmath,amsfonts,bm}

\def\eqref#1{equation~\ref{#1}}

\def\1{\bm{1}}

\DeclareMathAlphabet{\mathsfit}{\encodingdefault}{\sfdefault}{m}{sl}
\SetMathAlphabet{\mathsfit}{bold}{\encodingdefault}{\sfdefault}{bx}{n}

\DeclareMathOperator*{\argmax}{arg\,max}

\usepackage{subfig}
\usepackage{floatrow}
\usepackage{multirow}
\usepackage{array}
\usepackage{color}
\usepackage{colortbl}
\usepackage{wrapfig}
\newtheorem{definition}{Definition}
\usepackage[nonatbib,final]{neurips_2023}
\usepackage{graphicx}
\usepackage{bbding}
\usepackage[utf8]{inputenc} 
\usepackage[T1]{fontenc}    
\usepackage{hyperref}       
\usepackage{url}            
\usepackage{booktabs}       
\usepackage{amsfonts}       
\usepackage{nicefrac}       
\usepackage{xcolor}         
\usepackage{microtype}

\usepackage{nccmath}

\usepackage{amsmath}
\usepackage{amssymb}
\usepackage{mathtools}
\usepackage{amsthm}

\usepackage[capitalize,noabbrev]{cleveref}

\usepackage[utf8]{inputenc} 
\usepackage[T1]{fontenc}    
\usepackage{hyperref}       
\usepackage{url}            
\usepackage{booktabs}       
\usepackage{amsfonts}       
\usepackage{nicefrac}       
\usepackage{microtype}      
\usepackage{xcolor}         
\usepackage{algorithm}
\usepackage{algpseudocode}

\title{When Do Graph Neural Networks Help with Node Classification? Investigating the Impact of Homophily Principle on Node Distinguishability}
\author{
Sitao Luan$^{1,2}$, Chenqing Hua$^{1,2}$, Minkai Xu$^4$, Qincheng Lu$^{1}$, Jiaqi Zhu$^{1}$,\\ \textbf{Xiao-Wen Chang$^{1,\dagger}$, Jie Fu$^{2,5,\dagger}$, Jure Leskovec$^{4,\dagger}$, Doina Precup$^{1,2,3,\dagger}$} \\
\{sitao.luan@mail, chenqing.hua@mail, qincheng.lu@mail, jiaqi.zhu@mail, chang@cs, \\dprecup@cs\}.mcgill.ca, \{minkai, jure\}@cs.stanford.edu, jiefu@ust.hk\\
$^1$McGill University; $^2$ Mila - Quebec Artificial Intelligence Institute; $^3$Google DeepMind; \\
$^4$Stanford University; $^5$HKUST; $^\dagger$ Corresponding Authors\\
}

\begin{document}

\maketitle
\vspace{-0.5cm}
\begin{abstract}
\vspace{-0.3cm}
Homophily principle, \ie{} nodes with the same labels are more likely to be connected, has been believed to be the main reason for the performance superiority of Graph Neural Networks (GNNs) over Neural Networks on node classification tasks. Recent research suggests that, even in the absence of homophily, the advantage of GNNs still exists as long as nodes from the same class share similar neighborhood patterns \cite{ma2021homophily}. However, this argument only considers intra-class Node Distinguishability (ND) but neglects inter-class ND, which provides incomplete understanding of homophily on GNNs. In this paper, we first demonstrate such deficiency with examples and argue that an ideal situation for ND is to have smaller intra-class ND than inter-class ND. To formulate this idea and study ND deeply, we propose Contextual Stochastic Block Model for Homophily (CSBM-H) and define two metrics, Probabilistic Bayes Error (PBE) and negative generalized Jeffreys divergence, to quantify ND. With the metrics, we visualize and analyze how graph filters, node degree distributions and class variances influence ND, and investigate the combined effect of intra- and inter-class ND. Besides, we discovered the mid-homophily pitfall, which occurs widely in graph datasets. Furthermore, we verified that, in real-work tasks, the superiority of GNNs is indeed closely related to both intra- and inter-class ND regardless of homophily levels. Grounded in this observation, we propose a new hypothesis-testing based performance metric beyond homophily, which is non-linear, feature-based and can provide statistical threshold value for GNNs' the superiority. Experiments indicate that it is significantly more effective than the existing homophily metrics on revealing the advantage and disadvantage of graph-aware modes on both synthetic and benchmark real-world datasets.
\end{abstract}
\vspace{-0.5cm}
\section{Introduction}
\vspace{-0.3cm}
\label{sec:introduction}
Graph Neural Networks (GNNs) have gained popularity in recent years as a powerful tool for graph-based machine learning tasks. By combining graph signal processing and convolutional neural networks, various GNN architectures have been proposed  \cite{kipf2016classification,hamilton2017inductive,velivckovic2017attention,luan2019break,hua2022high}, and have been shown to outperform traditional neural networks (NNs) in tasks such as node classification (\textbf{NC}), graph classification, link prediction and graph generation. 
The success of GNNs is believed to be rooted in the homophily principle (assumption) \cite{mcpherson2001birds}, which states that connected nodes tend to have similar attributes \cite{hamilton2020graph}, providing extra useful information to the aggregated features over the original node features. Such relational inductive bias is thought to be a major contributor to the superiority of GNNs over NNs on various tasks \cite{battaglia2018relational}. On the other hand, the lack of homophily, \ie{} heterophily, is considered as the main cause of the inferiority of GNNs on heterophilic graphs, because nodes from different classes are connected and mixed, which can lead to indistinguishable node embeddings, making the classification task more difficult for GNNs \cite{zhu2020beyond, zhu2020graph, luan2022complete}. Numerous models have been proposed to address the heterophily challenge lately \cite{pei2020geom,zhu2020beyond, zhu2020graph,luan2022complete,bo2021beyond,lim2021new,chien2021adaptive,yan2021two,he2021bernnet,luan2021heterophily,li2022finding,wang2022acmp,luan2022revisiting}.

Recently, both empirical and theoretical studies indicate that the relationship between homophily principle and GNN performance is much more complicated than "homophily wins, heterophily loses" and the existing homophily metrics cannot accurately indicate the superiority of GNNs  \cite{ma2021homophily,luan2022revisiting, luan2022complete}. For example, the authors in \cite{ma2021homophily} stated that, as long as nodes within the same class share similar neighborhood patterns, their embeddings will be similar after aggregation. They provided experimental evidence and theoretical analysis, and concluded that homophily may not be necessary for GNNs to distinguish nodes. The paper \cite{luan2022revisiting} studied homophily/heterophily from post-aggregation node similarity perspective and found that heterophily is not always harmful, which is consistent with \cite{ma2021homophily}. Besides, the authors in \cite{luan2022complete} have proposed to use high-pass filter to address some heterophily cases, which is adopted in \cite{chien2021adaptive,bo2021beyond} as well. They have also proposed aggregation homophily, which is a linear feature-independent performance metric and is verified to be better at revealing the performance advantages and disadvantages of GNNs than the existing homophily metrics  \cite{pei2020geom,zhu2020beyond,lim2021new}. Moreover, \cite{chen2023exploiting} has investigated heterophily from a neighbor identifiable perspective and stated that heterophily can be helpful for NC when the neighbor distributions of intra-class nodes are identifiable. 

Inspite that the current literature on studying homophily principle provide the profound insights, they are still deficient: 1. \cite{ma2021homophily,chen2023exploiting} only consider intra-class node distinguishability (\textbf{ND}), but ignore inter-class ND; 2. \cite{luan2022revisiting} does not show when and how high-pass filter can help with heterophily problem; 3. There is a lack of a non-linear, feature-based performance metric which can leverage richer information to provide an \textbf{accurate statistical threshold value} to indicate whether GNNs are really needed on certain task or not.

To address those issues, in this paper: 1. We show that, to comprehensively study the impact of homophily on ND, one needs to consider both intra- and inter-class ND and an ideal case is to have smaller intra-class ND than inter-class ND; 2. To formulate this idea, we propose Contextual Stochastic Block Model for Homophily (CSBM-H) as the graph generative model. It incorporates an explicit parameter to manage homophily levels, alongside class variance parameters to control intra-class ND, and node degree parameters which are important to study homophily \cite{ma2021homophily, yan2021two}; 3. To quantify ND of CSBM-H, we propose and compute two ND metrics, Probabilistic Bayes Error (\textbf{PBE}) and Negative Generalized Jeffreys Divergence ($D_\text{NGJ}$), for the optimal Bayes classifier of CSBM-H. Based on the metrics, we can analytically study how intra- and inter-class ND impact ND together. We visualize the relationship between PBE, $D_\text{NGJ}$ and homophily levels and discuss how different graph filters (full-, low- and high-pass filters), class variances and node degree distributions will influence ND in details; 4. In practice, we verify through hypothesis testing that the performance superiority of GNNs is indeed related to whether intra-class ND is smaller than inter-class ND, regardless of homophily levels. Based on this conclusion and the p-values of hypothesis testing, we propose Classifier-based Performance Metric (CPM), a new non-linear feature-based metric that can provide statistical threshold values. Experiments show that CPM is significantly more effective than the existing homophily metrics on predicting the superiority of graph-aware models over graph-agnostic.
\vspace{-0.3cm}
\section{Preliminaries}
\vspace{-0.3cm}
\label{sec:prelimiary_notation}
We use \textbf{bold} font for vectors (\eg{}$\bm{v}$) and define a connected graph $\mathcal{G}=(\mathcal{V},\mathcal{E})$, where $\mathcal{V}$ is the set of nodes with a total of $N$ elements, $\mathcal{E}$ is the set of edges without self-loops. $A$ is the symmetric adjacency matrix with $A_{i,j}=1$ if there is an edge between nodes $i$ and $j$, otherwise $A_{i,j}=0$. 
$D$ is the diagonal degree matrix of the graph, with $D_{i,i} = d_i = \sum_j A_{i,j}$. The neighborhood set $\mathcal{N}_i$ of node $i$ is defined as $\mathcal{N}_i=\{j: e_{ij} \in \mathcal{E}\}$. A graph signal is a vector in $\mathbb{R}^N$, whose $i$-th entry is a feature of node $i$.  Additionally, we use ${X} \in \mathbb{R}^{N\times F_h}$ to denote the feature matrix, whose columns are graph signals and $i$-th row ${X_{i,:}} = \bm{x}_i^\top$ is the feature vector of node $i$ (\ie{} the full-pass (FP) filtered signal). The label encoding matrix is $Z\in \mathbb{R}^{N\times C}$, where $C$ is the number of classes, and its $i$-th row $Z_{i,:}$ is the one-hot encoding of the label of node $i$. We denote $z_i = \argmax_{j} Z_{i,j} \in \{1,2,\dots C\}$. The indicator function $\bm{1}_B$ equals 1 when event $B$ happens and 0 otherwise. 

For nodes $i,j \in \mathcal{V}$, if $z_i=z_j$, then they are considered as \textit{intra-class nodes}; if $z_i \neq z_j$, then they are considered to be \textit{inter-class nodes}. Similarly, an edge $e_{i,j}\in \mathcal{E}$  is considered to be an \textit{intra-class edge} if $z_i = z_j$, and an \textit{inter-class edge} if $z_i \neq z_j$.
\vspace{-0.2cm}
\subsection{Graph-aware Models and Graph-agnostic Models} 
\vspace{-0.2cm}
A network that includes the feature aggregation step according to graph structure is called graph-aware (\textbf{G-aware}) model, \eg{} GCN \cite{kipf2016classification}, SGC \cite{wu2019simplifying}; A network that does not use graph structure information is called graph-agnostic (\textbf{G-agnostic}) model, such as Multi-Layer Perceptron with 2 layers (MLP-2) and MLP-1. A G-aware model is often coupled with a G-agnostic model because when we remove the aggregation step in G-aware model, it becomes exactly the same as its coupled G-agnostic model, \eg{} GCN is coupled with MLP-2 and SGC-1 is coupled with MLP-1 as below, 
\begin{equation}
\begin{aligned}
    \label{eq:gcn_original}
   &\textbf{GCN: } Y = \text{Softmax} (\hat{A}_\text{sym} \; \text{ReLU} (\hat{A}_\text{sym} {X} W_0 ) \; W_1 ),\;\; \textbf{MLP-2: } Y = \text{Softmax} (  \text{ReLU} (  {X} W_0 ) \; W_1 )\\
   &\textbf{SGC-1: } Y = \text{Softmax} (\hat{A}_\text{sym}  {X} W_0 ), 
   \;\; \textbf{MLP-1: } Y = \text{Softmax} ( {X} W_0 )
   \end{aligned}
\end{equation}
where $\hat{A}_\text{sym} = \tilde{D}^{-1/2} \tilde{A} \tilde{D}^{-1/2}$, $\tilde{A} \equiv A+I$ and $\tilde{D} \equiv D+I$; $W_0 \in \Rbb^{F_0 \times F_1}$ and $W_1 \in \Rbb^{F_1\times O}$ are learnable parameter matrices. For simplicity, we denote $y_i = \argmax_j Y_{i,j} \in \{1,2,\dots C\}$. 
The random walk renormalized matrix $\hat{A}_{\text{rw}} = \tilde{D}^{-1} \tilde{A}$ can also be applied to GCN, which is essentially a mean aggregator commonly used in some spatial-based GNNs \cite{hamilton2017inductive}. 
To bridge spectral and spatial methods, we use $\hat{A}_{\text{rw}}$ in the theoretical analysis, but 
\textbf{self-loops are not added to the adjacency matrix} to maintain consistency with previous literature \cite{ma2021homophily,luan2022revisiting}.

To address the heterophily challenge, high-pass (HP) filter \cite{ekambaram2014graph}, such as $I-\hat{A}_{\text{rw}}$, is often used to replace low-pass (LP) filter \cite{maehara2019revisiting} $\hat{A}_{\text{rw}}$ in GCN \cite{bo2021beyond, chien2021adaptive,luan2022revisiting}. In this paper, we use $\hat{A}_{\text{rw}}$ and $I-\hat{A}_{\text{rw}}$ as the LP and HP operators, respectively. The LP and HP filtered feature matrices are represented as $H=\hat{A}_{\text{rw}} X$ and $H^\text{HP} = (I-\hat{A}_{\text{rw}})X$. For simplicity, we denote $\bm{h}_i = (H_{i,:})^\top, \bm{h}_i^\text{HP} = (H_{i,:}^\text{HP})^\top$.

\textbf{To measure how likely the G-aware model can outperform its coupled G-agnostic model before training them} (\ie{} if the aggregation step according to graph structure is helpful for node classification or not), a lot of homophily metrics have been proposed and we will introduce the most commonly used ones in the following subsection.
\vspace{-0.2cm}
\subsection{Homophily Metrics}
\vspace{-0.2cm}
\label{sec:homophily_metrics}
The homophily metric is a way to describe the relationship between node labels and graph structure. We introduce five commonly used homophily metrics: edge homophily \cite{abu2019mixhop,zhu2020beyond}, node homophily \cite{pei2020geom}, class homophily \cite{lim2021new}, generalized edge homophily \cite{jin2022raw} and aggregation homophily \cite{luan2022revisiting}, adjusted homophily \cite{platonov2022characterizing} and label informativeness \cite{platonov2022characterizing} as follows:
 \vspace{-0.2cm}
\begin{equation}
\begin{aligned}
\label{eq:definition_homophily_metrics}
& \resizebox{0.94\hsize}{!}{$\text{H}_\text{edge}(\mathcal{G}) = \frac{\big|\{e_{uv} \mid e_{uv}\in \mathcal{E}, Z_{u,:}=Z_{v,:}\}\big|}{|\mathcal{E}|},\;  \text{H}_\text{node}(\mathcal{G}) = \frac{1}{|\mathcal{V}|} \sum_{v \in \mathcal{V}}  \text{H}_\text{node}^v= \frac{1}{|\mathcal{V}|} \sum_{v \in \mathcal{V}} 
    \frac{\big|\{u \mid u \in \mathcal{N}_v, Z_{u,:}=Z_{v,:}\}\big|}{d_v},$} \\
&\resizebox{0.94\hsize}{!}{$\text{H}_\text{class}(\mathcal{G}) \!=\! \frac{1}{C\!-\!1} \sum_{k=1}^{C}\bigg[h_{k}
    \!-\! \frac{\big|\{v \!\mid\! Z_{v,k} \!=\! 1 \}\big|}{N}\bigg]_{+},\; \text{where } h_{k}\! =\! \frac{\sum_{v \in \mathcal{V}, Z_{v,k} = 1 } \big|\{u \!\mid\!  u \in \mathcal{N}_v, Z_{u,:}\!=\!Z_{v,:}\}\big| }{\sum_{v \in \{v|Z_{v,k}=1\}} d_{v}},$}  \\
&\resizebox{0.94\hsize}{!}{$\text{H}_\text{GE} (\mathcal{G}) = \frac{\sum\limits_{(i,j) \in \mathcal{E}} \text{cos}(\bm{x}_{i}, \bm{x}_{j})}{|\mathcal{E}|}, \; \text{H}_{\text{agg}}(\mathcal{G}) =  \frac{1}{\left| \mathcal{V} \right|} \times \left| \left\{v   \,\big| \, \mathrm{Mean}_u  \big( \{S(\hat{A},Z)_{v,u}^{Z_{u,:}=Z_{v,:}} \}\big) \geq  \mathrm{Mean}_u\big(\{S(\hat{A},Z)_{v,u}^{Z_{u,:} \neq Z_{v,:}}  \} \big) \right\} \right| ,$}\\
& \text{H}_\text{adj}=\frac{\text{H}_\text{edge} - \sum_{c=1}^C \bar{p}_c^2}{1-\sum_{k=1}^C \bar{p}_c^2},\; \text{LI} = -\frac{\sum_{c_1, c_2} p_{c_1, c_2 } \log \frac{p_{c_1, c_2}}{\bar{p}_{c_1} \bar{p}_{c_2}}}{\sum_c \bar{p}_c \log \bar{p}_c} = 2-\frac{\sum_{c_1, c_2} p_{c_1, c_2} \log p_{c_1, c_2}}{\sum_c \bar{p}_c \log \bar{p}_c}
\end{aligned}
\end{equation}
where $\text{H}_\text{node}^v$ is the local homophily value for node $v$; $[a]_{+}=\max (0, a)$; $h_{k}$ is the class-wise homophily metric \cite{lim2021new}; $\mathrm{Mean}_u\left(\{\cdot\}\right)$ takes the average over $u$ of a given multiset of values or variables and $S(\hat{A},Z)=\hat{A}Z(\hat{A}Z)^\top$ is the post-aggregation node similarity matrix; $D_c=\sum_{v: z_v=c} d_v, \bar{p}_c = \frac{D_c}{2|\mathcal{E}|}$, $p_{c_1, c_2} = \sum_{(u, v) \in \mathcal{E}} \frac{\bm{1} \left\{z_u=c_1, z_v=c_2 \right\}}{2|\mathcal{E}|}, c, c_1, c_2 \in \{1,\dots,C\}$. 

These metrics all fall within the range of $[0,1]$, with a value closer to $1$ indicating strong homophily and implying that G-aware models are more likely to outperform its coupled G-agnostic model, and vice versa. However, the current homophily metrics are almost all linear, feature-independent metrics which cannot provide a threshold value \cite{luan2022revisiting} for the superiority of G-aware model and fail to give an accurate measurement of node distinguishability (ND) . In the following section, we focus on quantifying the ND of graph models with with homophily levels and analyzing their relations. 

\begin{wrapfigure}{R}{0.34\textwidth}
  \begin{center}
    \includegraphics[width=1.15\textwidth]{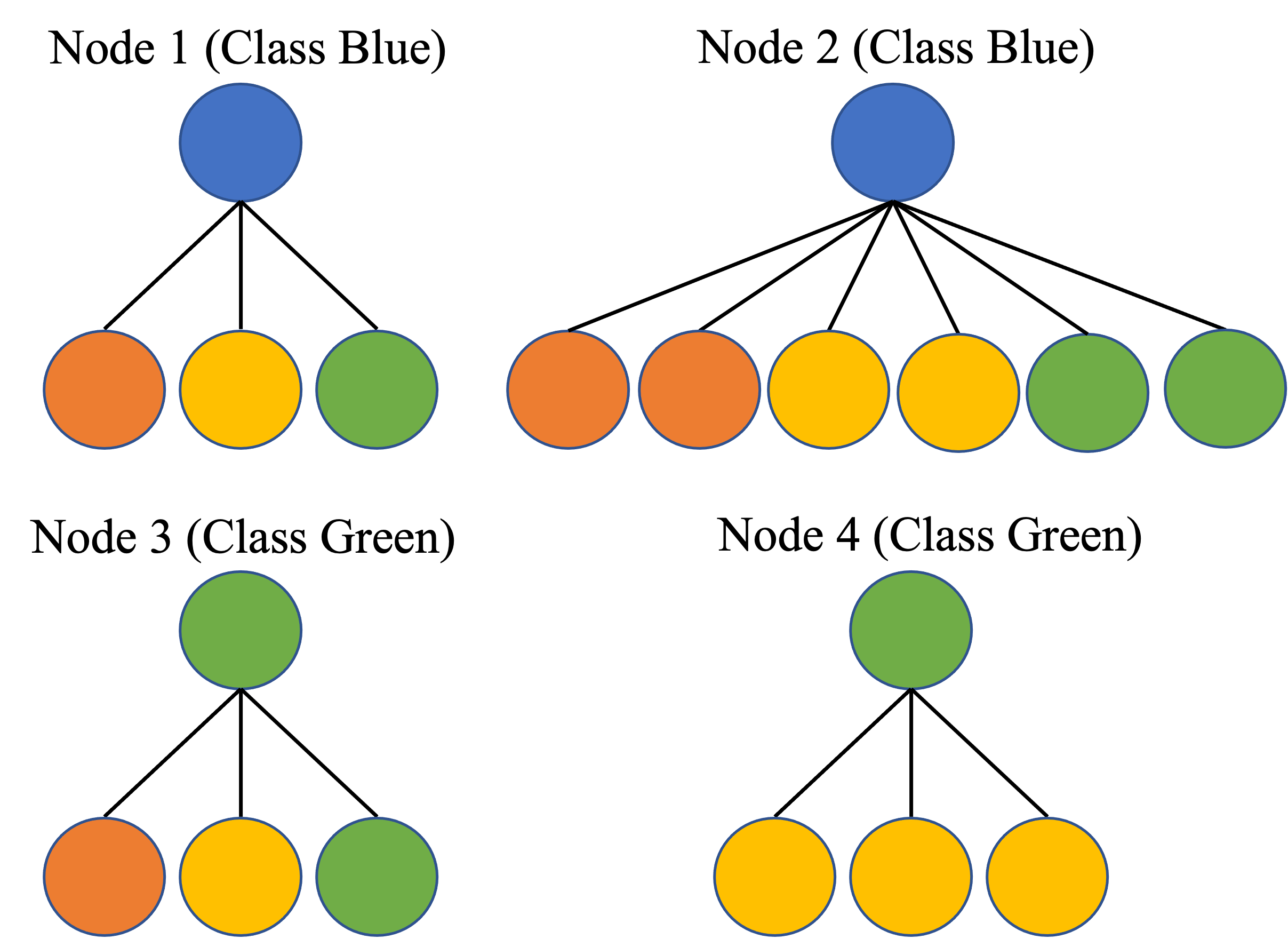}
  \end{center}
  \caption{Example of intra- and inter-class node distinguishability.
  }
  \label{fig:example_intra_inter_class_embeddings}
\end{wrapfigure}

\vspace{-0.3cm}
\section{Analysis of Homophily on Node Distinguishability (ND)}
\vspace{-0.2cm}
\label{sec:analysis_node_distinguishability}

\subsection{Motivation}
\vspace{-0.2cm}
\label{sec:motivation}

\paragraph{The Problem in Current Literature} Recent research has shown that heterophily does not always negatively impact the embeddings of intra-class nodes, as long as their neighborhood patterns "corrupt in the same way" \cite{ma2021homophily,chen2023exploiting}. For example, in Figure~\ref{fig:example_intra_inter_class_embeddings}, nodes \{1,2\} are from class blue and both have the same heterophilic neighborhood patterns. As a result, their aggregated features will still be similar and they can be classified into the same class.

However, this is only partially true for ND if we forget to discuss inter-class ND, \eg{}  node 3 in Figure~\ref{fig:example_intra_inter_class_embeddings} is from class green and has the same neighborhood pattern ($1/3$ orange, $1/3$ yellow and $1/3$ green) as nodes \{1,2\}, which means the inter-class ND will be lost after aggregation. This highlights the necessity for careful consideration of both intra- and inter-class ND when evaluating the impact of homophily on the performance of GNNs and an ideal case 
for NC would be node \{1,2,4\}, where we have smaller intra-class "distance" than inter-class "distance". We will formulate the above idea in this section and verify if it really relates to the performance of GNNs in section \ref{sec:empirical_study}. In the following subsection, we will propose a toy graph model, on which we can study the relationship between homophily and ND directly and intuitively, and analyze intra- and inter-class ND analytically.

\vspace{-0.3cm}
\subsection{CSBM-H and Optimal Bayes Classifier}
\vspace{-0.2cm}
\label{sec:csbmh_analysis}
In order to have more control over the assumptions on the node embeddings, we consider the Contextual Stochastic Block Model (CSBM) \cite{deshpande2018contextual}. It is a generative model that is commonly used to create graphs and node features, and it has been widely adopted to study the behavior of GNNs \cite{tsitsulin2022synthetic, baranwal2021graph, wei2022understanding}. To investigate the impact of homophily on ND, the authors in \cite{ma2021homophily} simplify CSBM to the two-normal setting, where the node features $X$ are assumed to be sampled from two normal distributions and intra- and inter-class edges are generated according to two separate parameters. This simplification does not lose much information about CSBM, but 1. it does not include an explicit homophily parameter to study homophily directly and intuitively; 2. it does not include class variance parameters to study intra-class ND; 3. the authors do not rigorously quantify ND.

In this section, we introduce the Contextual Stochastic Block Model for Homophily/Heterophily (CSBM-H), which is a variation of CSBM that incorporates an explicit homophily parameter $h$ for the two-normal setting and also has class variance parameters $\sigma_0^2,\sigma_1^2$ to describe the intra-class ND. We then derive the optimal Bayes classifier ($\text{CL}_\text{Bayes}$) and negative generalized Jeffreys divergence for CSBM-H, based on which we can quantify and investigate ND for CSBM-H.

\noindent \textbf{CSBM-H($\bm{\mu}_0,\bm{\mu}_1,\sigma_0^2 I,\sigma_1^2 I,{d}_0,{d}_1,h$)}\footnote{This implies that we generate undirected graphs. See Appendix \ref{appendix:directed_or_undirected} for the discussion of directed vs. undirected graphs. See \ref{appendix:extend_csbmh_general_settings} for the discussion on how to extend CSBM-H to more general settings.} \quad The generated graph consists of two disjoint sets of nodes, $i \in \mathcal{C}_0$ and $j \in \mathcal{C}_1$, corresponding to the two classes. The features of each node are generated independently, with $\bm{x}_i$ generated from $N(\bm{\mu}_0, \sigma_0^2 I)$ and $\bm{x}_j$ generated from $N(\bm{\mu}_1, \sigma_1^2 I)$, where $\bm{\mu}_0, \bm{\mu}_1 \in \mathbb{R}^{F_h}$ and ${F_h}$ is the dimension of the embeddings. The degree
of nodes in $\mathcal{C}_0$ and $\mathcal{C}_1$ are ${d}_0,{d}_1\in \mathbb{N}$ respectively. For $i \in \mathcal{C}_0$, its neighbors are generated by independently sampling from $h\cdot {d}_0$ intra-class nodes and $(1-{h})\cdot {d}_0$ inter-class nodes \footnote{To avoid unnecessary confusion: we relax $hd_0$ and $(1-h)d_0$ to be continuous values so that the visualization in the following sections are more readable and intuitive, especially to show the intersections of the curves.}. The neighbors of $j \in \mathcal{C}_1$ are generated in the same way. As a result, the FP (full-pass), LP and HP filtered features are generated as follows,
\vspace{-0.1cm}
\begin{equation}
\label{eq:two_normal_settings}
    \begin{split}
    i \in \mathcal{C}_0: \, & \bm{x}_i \sim N(\bm{\mu}_0,\sigma_0^2 I);\;
    \bm{h}_i \sim N(\tilde{\bm{\mu}}_0, \tilde{\sigma}_0^2  I), \ \bm{h}_i^{\text{HP}} \sim N\left(\tilde{\bm{\mu}}_0^{\text{HP}}, (\tilde{\sigma}_0^{\text{HP}})^2  I \right),    \\
     j \in \mathcal{C}_1: \,  & \bm{x}_j \sim N(\bm{\mu}_1,\sigma_1^2 I);\;
    \bm{h}_j \sim N(\tilde{\bm{\mu}}_1, \tilde{\sigma}_1^2  I),\ \bm{h}_j^{\text{HP}} \sim N\left(\tilde{\bm{\mu}}_1^{\text{HP}}, (\tilde{\sigma}_1^{\text{HP}})^2 I \right), 
    \end{split}
\end{equation}
where $\tilde{\bm{\mu}}_0 = {h}(\bm{\mu}_0 - \bm{\mu}_1)+\bm{\mu}_1, \; \tilde{\bm{\mu}}_1 = {h}(\bm{\mu}_1 - \bm{\mu}_0)+\bm{\mu}_0, \; \tilde{\bm{\mu}}_0^{\text{HP}} = (1-{h})(\bm{\mu}_0 - \bm{\mu}_1), \; \tilde{\bm{\mu}}_1^{\text{HP}} = (1-{h})(\bm{\mu}_1 - \bm{\mu}_0), \; \tilde{\sigma}_0^2 = \frac{({h}(\sigma_0^2-\sigma_1^2)+\sigma_1^2)}{{d}_0}, \; \tilde{\sigma}_1^2 = \frac{({h}(\sigma_1^2-\sigma_0^2)+\sigma_0^2)}{{d}_1}, \; (\tilde{\sigma}_0^{\text{HP}})^2 = \sigma_0^2+ \frac{({h}(\sigma_0^2-\sigma_1^2)+\sigma_1^2)}{{d}_0}, \; (\tilde{\sigma}_1^{\text{HP}})^2 = \sigma_1^2 + \frac{({h}(\sigma_1^2-\sigma_0^2)+\sigma_0^2)}{{d}_1}$. If $\sigma_0^2<\sigma_1^2$, we refer to $\mathcal{C}_0$ as the low variation class and $\mathcal{C}_1$ as the high variation class. The variance of each class can reflect the intra-class ND. We abuse the notation $\bm{x}_i\in {\cal C}_0$ for $i\in {\cal C}_0$  and $\bm{x}_j\in {\cal C}_1$ for $j\in {\cal C}_1$.

To quantify the ND of CSBM-H, we first compute the optimal Bayes classifier in the following theorem. The theorem is about the original features, but the results are applicable to the filtered features when the parameters are replaced according to \eqref{eq:two_normal_settings}.

\vspace{-0.3cm}
\begin{theorem} 1 Suppose $\sigma_0^2 \neq \sigma_1^2$ and $\sigma_0^2, \sigma_1^2 > 0$, the prior distribution for $\bm{x}$ is $\mathbb{P}(\bm{x}\in {\cal C}_0) = \mathbb{P}(\bm{x}\in {\cal C}_1) = 1/2$, then the optimal Bayes Classifier ($\text{CL}_{\text{Bayes}}$) for CSBM-H ($\bm{\mu}_0,\bm{\mu}_1,\sigma_0^2 I,\sigma_1^2 I,{d}_0,{d}_1,{h}$) is 
\vspace{-0.4cm}
$$\text{CL}_\text{Bayes}(\bm{x}) = \left\{
\begin{aligned}
1,\; \eta(\bm{x}) \geq 0.5\\
0,\; \eta(\bm{x}) < 0.5
\end{aligned}
\right. ,\; \eta(\bm{x}) = \mathbb{P}(z=1|\bm{x}) = \frac{1}{1+\exp{\left(Q(\bm{x})\right)}}, $$
where $Q(\bm{x}) = a \bm{x}^\top \bm{x} + \bm{b}^\top \bm{x} +c,\ 
a = \frac{1}{2}\left(\frac{1}{\sigma_1^2} - \frac{1}{\sigma_0^2}\right), 
\bm{b}=\frac{\bm{\mu}_0}{\sigma_0^2}-\frac{\bm{\mu}_1}{\sigma_1^2}, 
c=\frac{\bm{\mu}_1^\top\bm{\mu}_1}{2\sigma_1^2} - \frac{\bm{\mu}_0^\top\bm{\mu}_0}{2\sigma_0^2} + \ln{\left(\frac{\sigma_1^{F_h} }{\sigma_0^{F_h} } \right)}$ \footnote{The Bayes classifier for multiple categories ($>2$) can be computed by stacking multiple expectation terms using similar methods as in \cite{devroye2013probabilistic, farago1993strong}. We do not discuss the more complicated settings in this paper.}. See the proof in Appendix \ref{appendix:proof_of_theorem1}.
\end{theorem}
\vspace{-0.3cm}
\noindent \textbf{Advantages of $\text{CL}_\text{Bayes}$ Over the Fixed Linear Classifier in \cite{ma2021homophily}} \quad The decision boundary in \cite{ma2021homophily} is defined as $P=\{ \bm{x}| \bm{w}^\top \bm{x} - \bm{w}^\top (\bm{\mu}_0 + \bm{\mu}_1)/2\}$ where $\bm{w}= (\bm{\mu}_0 - \bm{\mu}_1)/||\bm{\mu}_0 - \bm{\mu}_1||_2$ is a fixed parameter. This classifier only depends on $\bm{\mu}_0, \bm{\mu}_1$ and is independent of $h$. However, as $h$ changes, the "separability" of the two normal distributions should be different. The fixed classifier cannot capture this difference, and thus is not qualified to measure ND for different $h$. Besides, we cannot investigate how variances $\sigma_0^2, \sigma_1^2$ and node degrees $d_0, d_1$ affect ND with the fixed classifier in \cite{ma2021homophily}. 

In the following subsection, we will define two methods to quantify ND of CSBM-H, one is based on $\text{CL}_{\text{Bayes}}$, which is a precise measure but hard to be explainable; another is based on KL-divergence, which can give us more intuitive understanding of how intra- and inter-class ND will impact ND  at different homophily levels. These two measurements can be used together to analyze ND.
\vspace{-0.2cm}
\subsection{Measure Node Distinguishability of CSBM-H}
\vspace{-0.2cm}
The Bayes error rate (BE) of the data distribution is the probability of a node being mis-classified when the true class probabilities are known given the predictors \cite{hastie2009elements}. It can be used to measure the distinguishability of node embeddings and the BE for $\text{CL}_{\text{Bayes}}$ is defined as follows,
\begin{definition}[Bayes Error Rate]
The Bayes error rate \cite{hastie2009elements} for $\textup{CL}_{\textup{Bayes}}$ is defined as
\begin{align*}
    \textup{BE} & = \mathbb{E}_{\bm{x}}[\mathbb{P}(z|\textup{CL}_{\textup{Bayes}}(\bm{x})\neq z)] =  \mathbb{E}_{\bm{x}}[1-\mathbb{P}(\textup{CL}_{\textup{Bayes}}(\bm{x}) = z | \bm{x})] 
\end{align*}
\end{definition}
Specifically, the BE for CSBM-H can be written as 
\begin{equation}
\label{eq:bayes_error_csbmh}
\begin{aligned}
 \resizebox{0.94\hsize}{!}{$\text{BE} = \mathbb{P}\left(\bm{x}\in \mathcal{C}_0\right) \left(1- \mathbb{P}(\text{CL}_\text{Bayes}(\bm{x})=0|\bm{x}\in \mathcal{C}_0) \right) + \mathbb{P}(\bm{x}\in \mathcal{C}_1) \left(1- \mathbb{P} (\text{CL}_\text{Bayes}(\bm{x})=1|\bm{x}\in \mathcal{C}_1)\right).$}
\end{aligned}
\end{equation}
To estimate the above value, we compute Probabilistic Bayes Error (PBE) for CSBM-H as follows. 

\noindent\textbf{Probabilistic Bayes Error (PBE)} \quad
The random variable in each dimension of $\bm{x}$ is independently normally distributed. As a result, 
$Q(\bm{x})$ defined in Theorem 1 follows a generalized $\chi^2$ distribution \cite{davies1973numerical,davies1980algorithm}
(See the calculation in Appendix \ref{appendix:noncentral_chisquare}). 
Specifically,
\begin{align*}
\text{For } \bm{x}_i \in \mathcal{C}_0,\;  Q(\bm{x}_i) \sim \tilde{\chi}^2(w_0,F_h,\lambda_0) +\xi;\;  \bm{x}_j \in \mathcal{C}_1,\;  Q(\bm{x}_j) \sim \tilde{\chi}^2(w_1,F_h,\lambda_1) +\xi
\end{align*}
\vspace{-0.1cm}
where $w_0 = a\sigma_0^2, w_1 = a\sigma_1^2$ , the degree of freedom is $F_h$, $ \lambda_0 = (\frac{\bm{\mu}_0}{\sigma_0} + \frac{\bm{b}}{2a\sigma_0})^\top (\frac{\bm{\mu}_0}{\sigma_0} + \frac{ \bm{b}}{2a\sigma_0}), \; \lambda_1 = (\frac{\bm{\mu}_1}{\sigma_1} + \frac{\bm{b}}{2a\sigma_1})^\top (\frac{\bm{\mu}_1}{\sigma_1}+\frac{\bm{b}}{2a\sigma_1}) \; \text{ and }\xi = c-\frac{\bm{b}^\top\bm{b}}{4a}\;$. Then, by using the Cumulative Distribution Function (CDF) of $\tilde{\chi}^2$, we can calculate the predicted probabilities directly as,
\vspace{-0.1cm}
\begin{align*}
\resizebox{1\hsize}{!}{$\mathbb{P}(\text{CL}_\text{Bayes}(\bm{x})=0|\bm{x}\in \mathcal{C}_0)  = 1-\text{CDF}_{\tilde{\chi}^2(w_0,F_h,\lambda_0)}(-\xi),\ 
\mathbb{P}(\text{CL}_\text{Bayes}(\bm{x})=1|\bm{x}\in \mathcal{C}_1) = \text{CDF}_{\tilde{\chi}^2(w_1,F_h,\lambda_1)}(-\xi) $} . 
\end{align*}
\vspace{-0.1cm}
Suppose we have a balanced prior distribution $\mathbb{P}(\bm{x} \in \mathcal{C}_0)=\mathbb{P}(\bm{x} \in \mathcal{C}_1)=1/2$. Then, PBE is,
\vspace{-0.1cm}
$$
    \frac{ \text{CDF}_{\tilde{\chi}^2(w_0,F_h,\lambda_0)}(-\xi) + \left(1- \text{CDF}_{\tilde{\chi}^2(w_1,F_h,\lambda_1)}(-\xi) \right)}{2}
$$
To investigate the impact of homophily on the ND for LP and HP filtered embeddings, we just need to replace $\left( \bm{\mu}_0, \sigma_0^2,\bm{\mu}_1, \sigma_1^2 \right)$ by $\left( \tilde{\bm{\mu}}_0, \tilde{\sigma}_0^2 ,\tilde{\bm{\mu}}_1, \tilde{\sigma}_1^2 \right)$ and $\left(\tilde{\bm{\mu}}_0^{\text{HP}}, (\tilde{\sigma}_0^{\text{HP}})^2,\tilde{\bm{\mu}}_1^{\text{HP}}, (\tilde{\sigma}_1^{\text{HP}})^2 \right)$ as \eqref{eq:two_normal_settings}. 

PBE can be numerically calculated and visualized to show the relationship between $h$ and ND precisely. However, we do not have an analytic expression for PBE, which makes it less explainable and intuitive. To address this issue, we define another metric for ND in the following paragraphs.

\noindent \textbf{Generalized Jeffreys Divergence} \quad The KL-divergence is a statistical measure of how a probability distribution $P$ is different from another distribution $Q$ \cite{csiszar1975divergence}. It offers us a tool to define an explainable ND measure, generalized Jeffreys divergence.

\begin{definition}[Generalized Jeffreys Divergence] For a random variable $\bm{x}$
which has either the distribution $P(\bm{x})$ or the distribution $Q(\bm{x})$, the generalized Jeffreys divergence \footnote{Jeffreys divergence \cite{jeffreys1998theory} is originally defined as $D_{\text{KL}}(P||Q)+D_{\text{KL}}(Q||P)$} is defined as
\vspace{-0.2cm}
\begin{align*}
    D_\text{GJ}(P,Q) & =  \mathbb{P}(\bm{x} \sim P) \mathbb{E}_{\bm{x}\sim P}\left[\ln{\frac{P(\bm{x})}{Q(\bm{x})}} \right] + \mathbb{P}(\bm{x} \sim Q) \mathbb{E}_{\bm{x}\sim Q}\left[\ln{\frac{Q(\bm{x})}{P(\bm{x})}} \right] 
\end{align*}

\end{definition}
With $\mathbb{P}(\bm{x} \sim P) = \mathbb{P}(\bm{x} \sim Q) = 1/2$ \footnote{We provide an open-ended discussion of imbalanced prior distributions in Appendix \ref{appendix:imbalanced_prior_distribution}.}, the negative generalized Jeffreys divergence for the two-normal setting in CSBM-H can be computed by (See Appendix \ref{appendix:NGJD} for the calculation)

\begin{equation}
\label{eq:ngjd}
\hspace*{-0.2cm}
    D_\text{NGJ}(\text{CSBM-H}) = \underbrace{- d_X^2(\frac{1}{4\sigma_1^2}+\frac{1}{4\sigma_0^2})}_{\text{Negative Normalized Distance}}  \underbrace{-\frac{F_h}{4}(\rho^2 + \frac{1}{\rho^2}-2)}_{\text{Negative Variance Ratio}} 
\end{equation}
where $d_X^2 = (\bm{\mu}_0- \bm{\mu}_1)^\top (\bm{\mu}_0 - \bm{\mu}_1)$ is the squared Euclidean distance between centers; $\rho=\frac{\sigma_0}{\sigma_1}$ and since we assume $\sigma_0^2 < \sigma_1^2$,  we have $0<\rho<1$. For $\bm{h}$ and $\bm{h}^\text{HP}$, we have $d_H^2 = (2h-1)^2d_X^2, d_\text{HP}^2 = 4(1-h)^2 d_X^2$. The smaller $D_\text{NGJ}$ the CSBM-H has, the more distinguishable the embeddings are.

From \eqref{eq:ngjd}, we can see that $D_\text{NGJ}$ implies that ND relies on two terms, Expected Negative Normalized Distance (ENND) and the Negative Variance Ratio (NVR): 1. ENND depends on how large is the inter-class ND $d_X^2$ compared with the normalization term $\frac{1}{4\sigma_1^2}+\frac{1}{4\sigma_0^2}$, which is determined by intra-class ND (variances $\sigma_0,\sigma_1$); NVR depends on how different the two intra-class NDs are, \ie{} when the intra-class ND of high-variation class is significantly larger than that of low-variation class ($\rho$ is close to 0), NVR is small which means the nodes are more distinguishable and vice versa. 

Now, we can investigate the impact of homophily on ND through the lens of PBE and $D_\text{NGJ}$ \footnote{See two more metrics, negative squared Wasserstein distance and Hellinger distance, in Appendix \ref{appendix:two_more_nd_metrics}.}. Specifically, we set the standard CSBM-H as $\bm{\mu}_0 = [-1,0],\bm{\mu}_1 =[0,1], \sigma_0^2 = 1,\sigma_1^2 = 2, d_0 = 5,d_1 = 5$. And as shown in Figure \ref{fig:standard_csbmh_example}, its PBE and $D_\text{NGJ}$ curves for LP filtered feature $\bm{h}$ are bell-shaped \footnote{This is consistent with the empirical results found in \cite{luan2022revisiting} that the relationship between the prediction accuracy of GNN and homophily value is a U-shaped curve.}. This indicates that, contrary to the prevalent belief that heterophily has the most negative impact on ND, a medium level of homophily actually has a more detrimental effect on ND than extremely low levels of homophily. We refer to this phenomenon as the \textbf{mid-homophily pitfall}. 

The PBE and $D_\text{NGJ}$ curves for $\bm{h}^\text{HP}$ are monotonically increasing, which means that the high-pass filter works better in heterophily areas than in homophily areas. Moreover, it is observed that $\bm{x}$, $\bm{h}$, and $\bm{h}^\text{HP}$ will get the lowest PBE and $D_\text{NGJ}$ in different homophily intervals, which we refer to as the "FP regime \textit{(black)}", "LP regime \textit{(green)}", and "HP regime \textit{(red)}" respectively. This indicates that LP filter works better at very low and very high homophily intervals (two ends), HP filter works better at low to medium homophily interval \footnote{This verifies the conjecture made in \cite{luan2022revisiting} saying that high-pass filter cannot address all kinds of heterophily and only works well for certain heterophily cases.}, the original (\ie{} full-pass or FP filtered) features works betters at medium to high homophily area.

Researchers have always been interested in exploring how node degree relates to the effect of homophily \cite{ma2021homophily,yan2021two}. In the upcoming subsection, besides node degree, we will also take a deeper look at the impact of class variances via the homophily-ND curves and the FP, LP and HP regimes.

\begin{figure}[htbp!]
    \centering
     {  \vspace{-0.3cm}
     \resizebox{1\hsize}{!}{
     \subfloat[PBE]{
     \captionsetup{justification = centering}
     \includegraphics[width=0.22\textwidth]{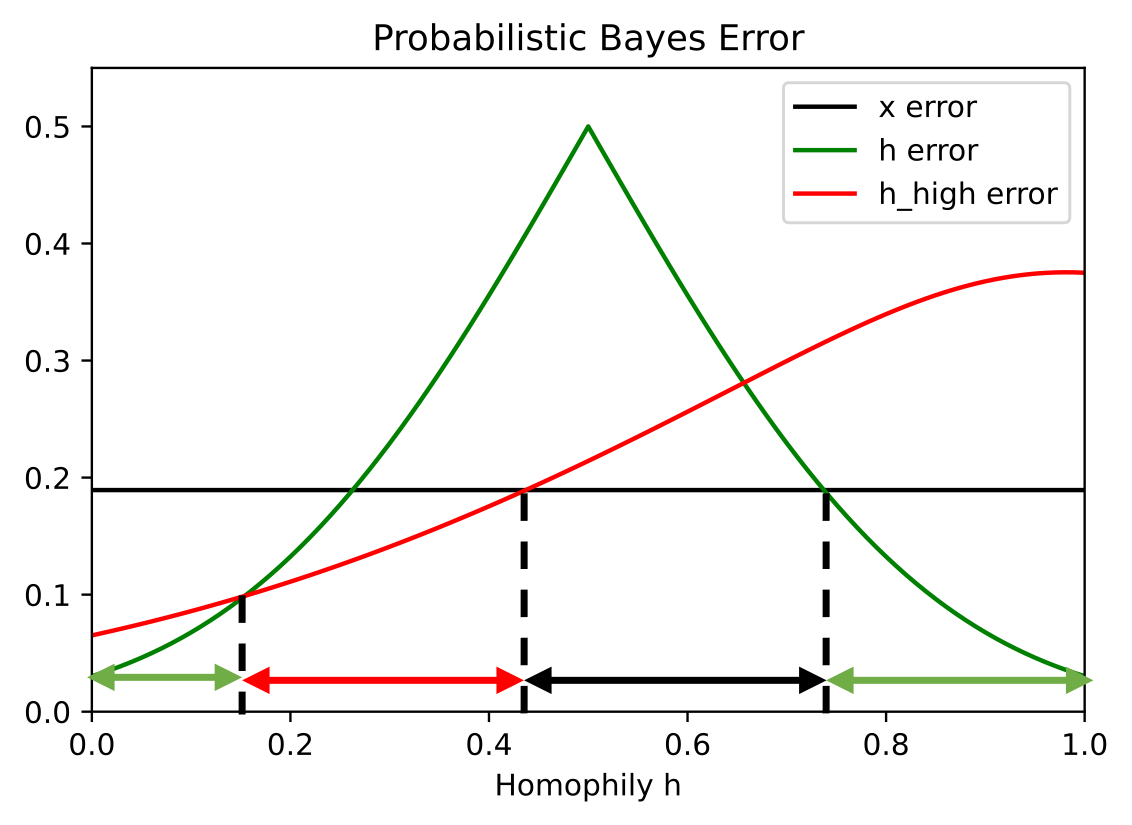}
     }
     \subfloat[$D_\text{NGJ}$]{
     \captionsetup{justification = centering}
     \includegraphics[width=0.25\textwidth]{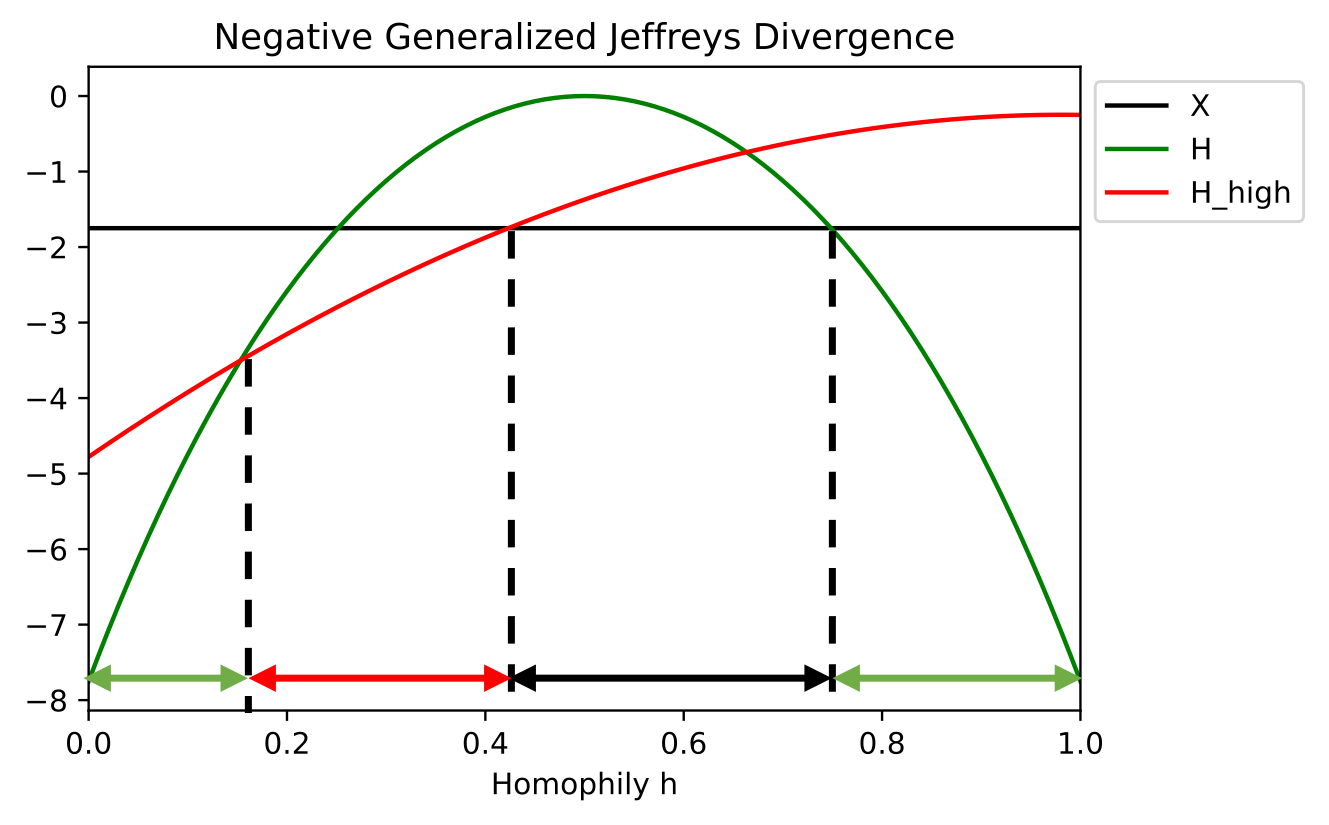}
     }\hspace*{-0.3cm}
     \subfloat[ENND]{
     \captionsetup{justification = centering}
     \includegraphics[width=0.25\textwidth]{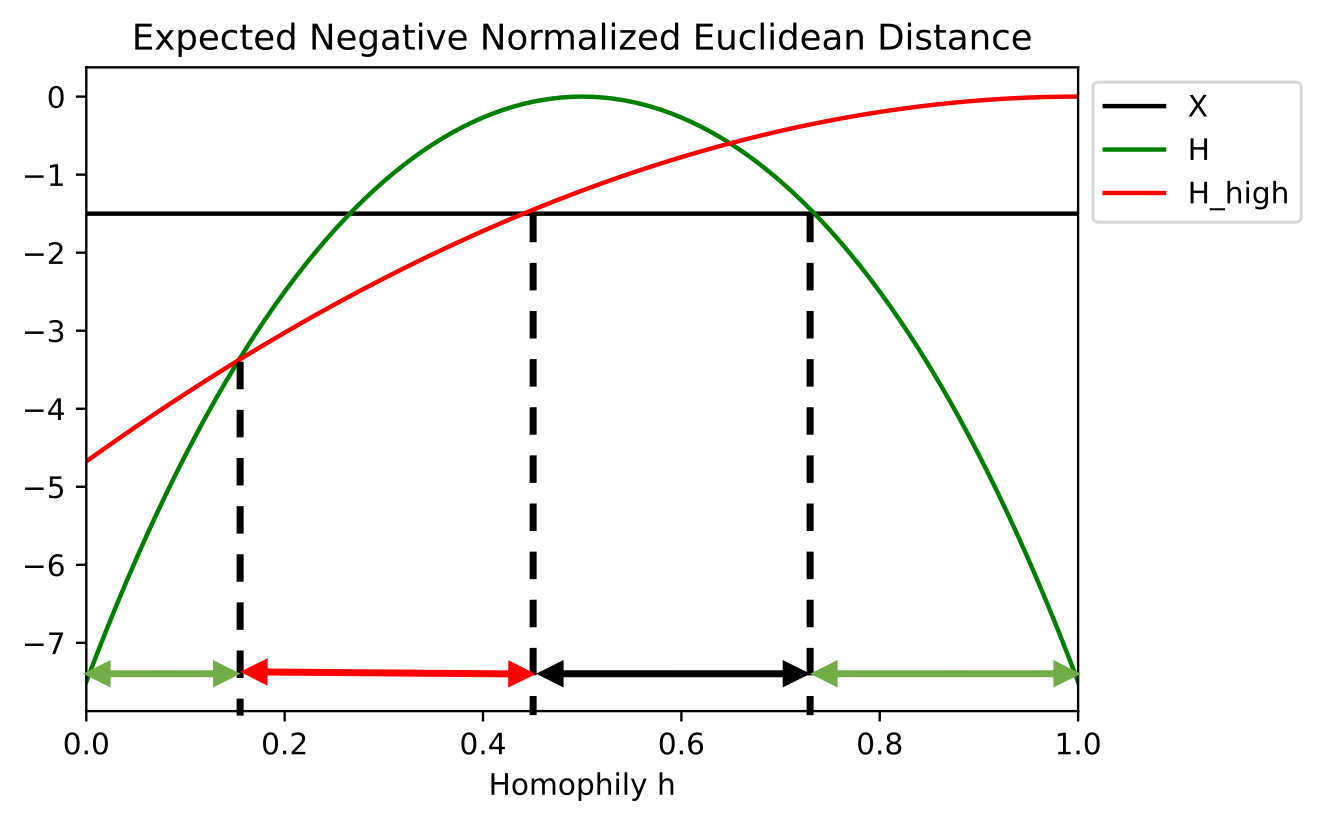}
     } \hspace*{-0.3cm}
     \subfloat[NVR]{
     \captionsetup{justification = centering}
     \includegraphics[width=0.22\textwidth]{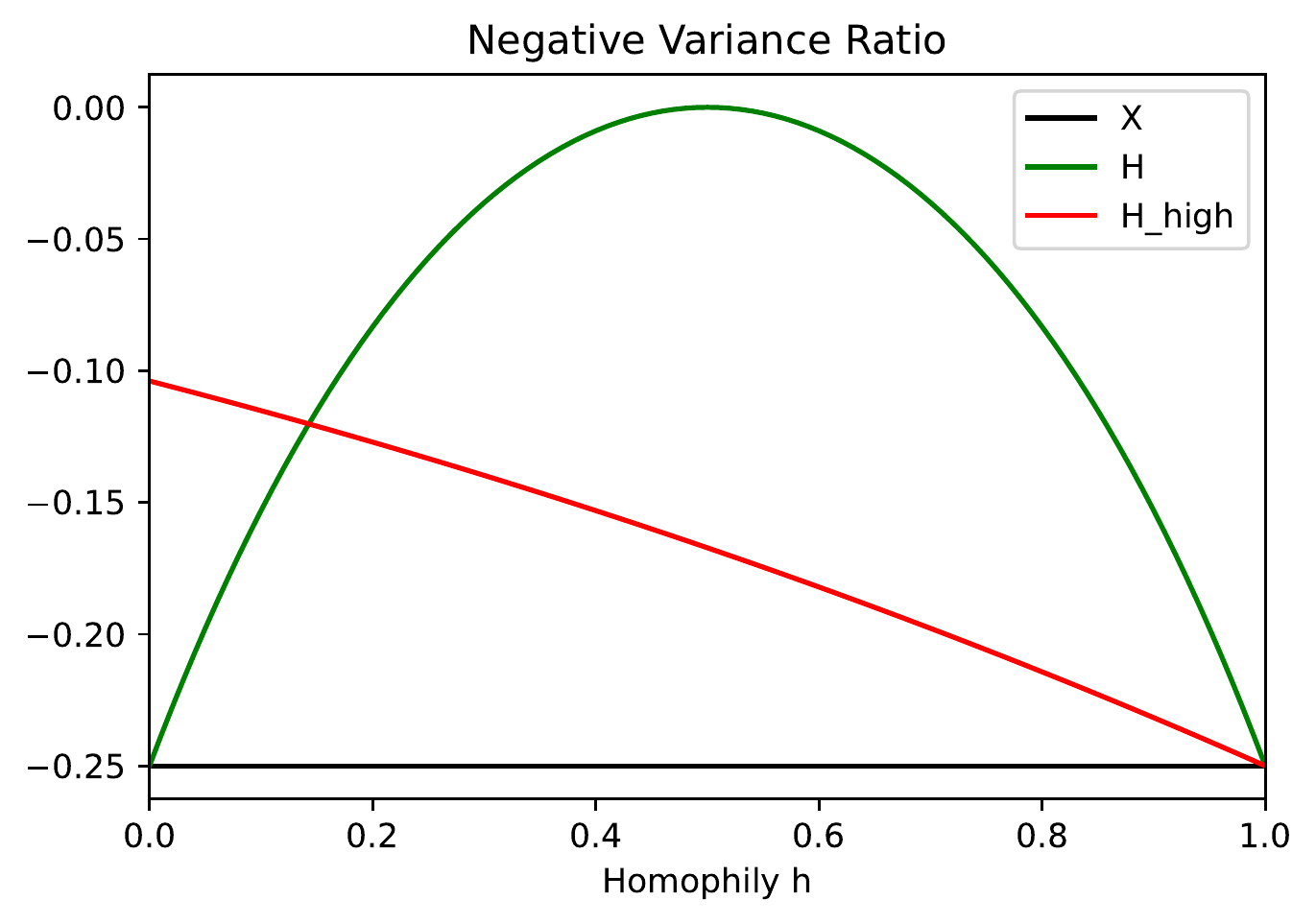}
     } 
     }
     }
     \vspace{-0.2cm}
       \caption{
  Visualization of CSBM-H $\left(\bm{\mu}_0 = [-1,0],\bm{\mu}_1 =[0,1], \sigma_0^2 = 1,\sigma_1^2 = 2, \right.$ $\left. d_0 = 5,d_1 = 5 \right)$
  }
     \label{fig:standard_csbmh_example}
\end{figure}

\vspace{-0.3cm}
\subsection{Ablation Study on CSBM-H}
\vspace{-0.2cm}
\label{sec:obervations_on_csbmh}
\begin{figure}[htbp!]
    \centering
     {\hspace*{-0.3cm}
     \resizebox{1\hsize}{!}{
     \subfloat[PBE]{
     \captionsetup{justification = centering}
     \includegraphics[width=0.22\textwidth]{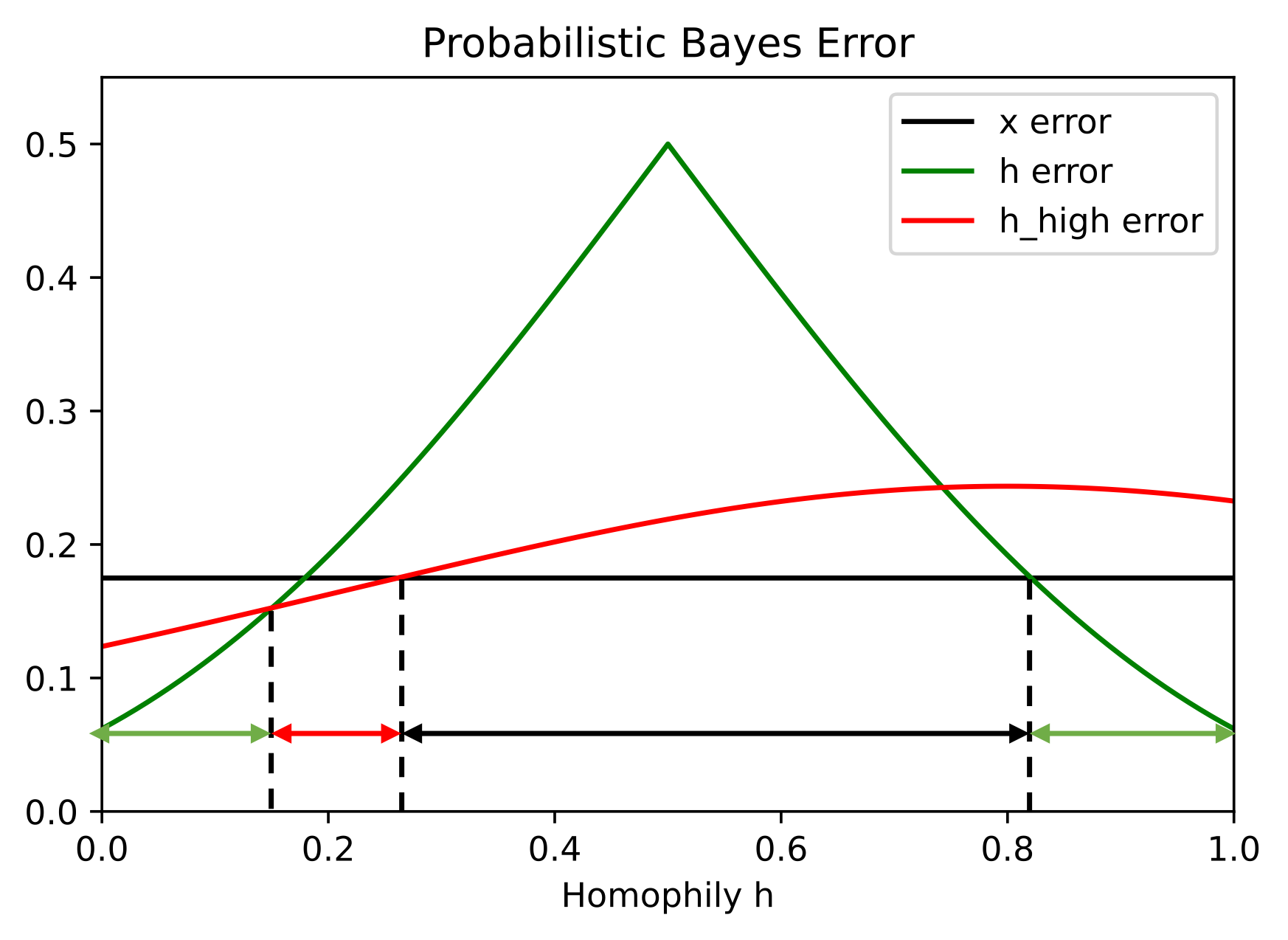}
     } \hspace*{-0.2cm}
     \subfloat[$D_\text{NGJ}$]{
     \captionsetup{justification = centering}
     \includegraphics[width=0.25\textwidth]{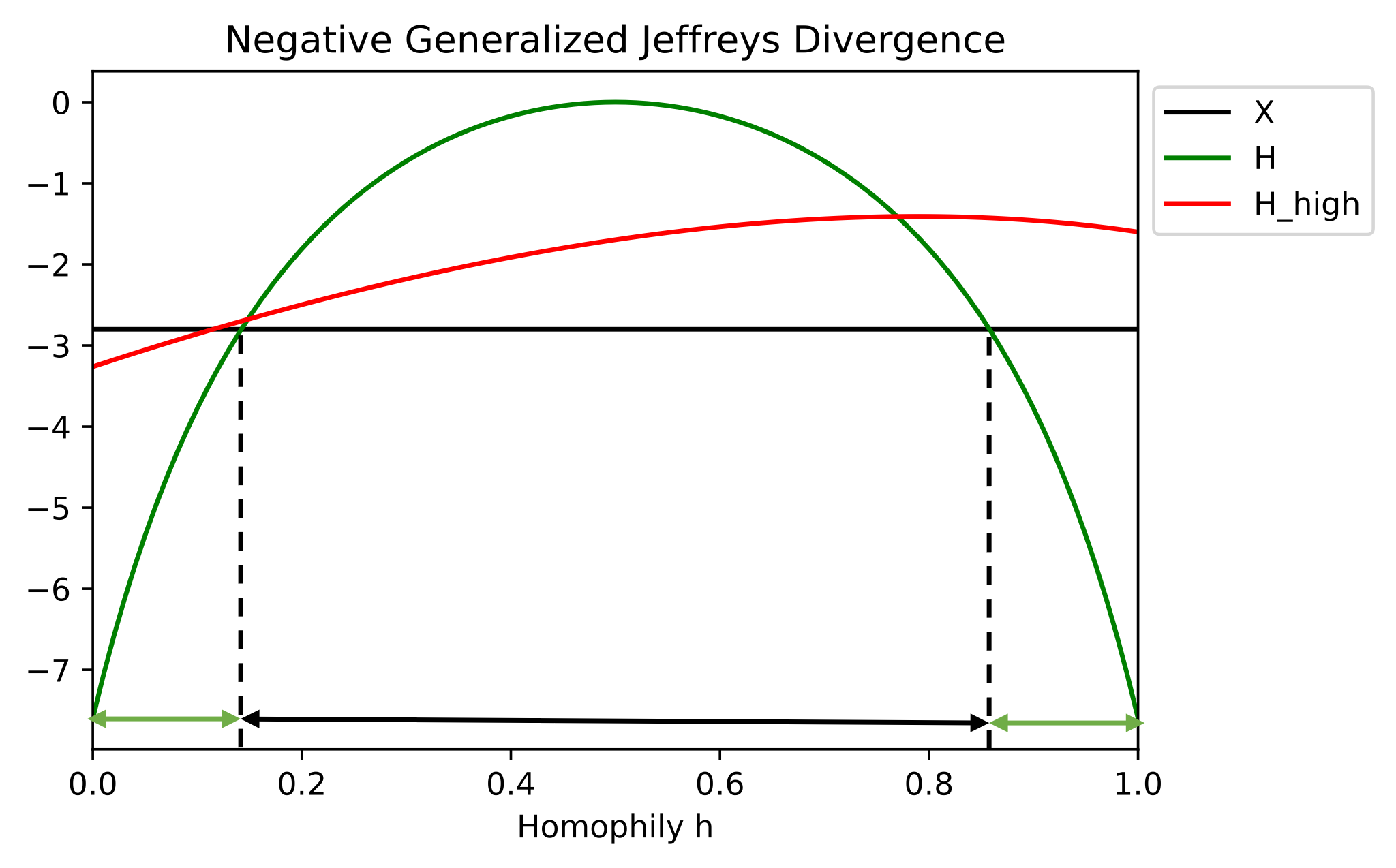}
     }  \hspace*{-0.2cm}
     \subfloat[ENND]{
     \captionsetup{justification = centering}
     \includegraphics[width=0.25\textwidth]{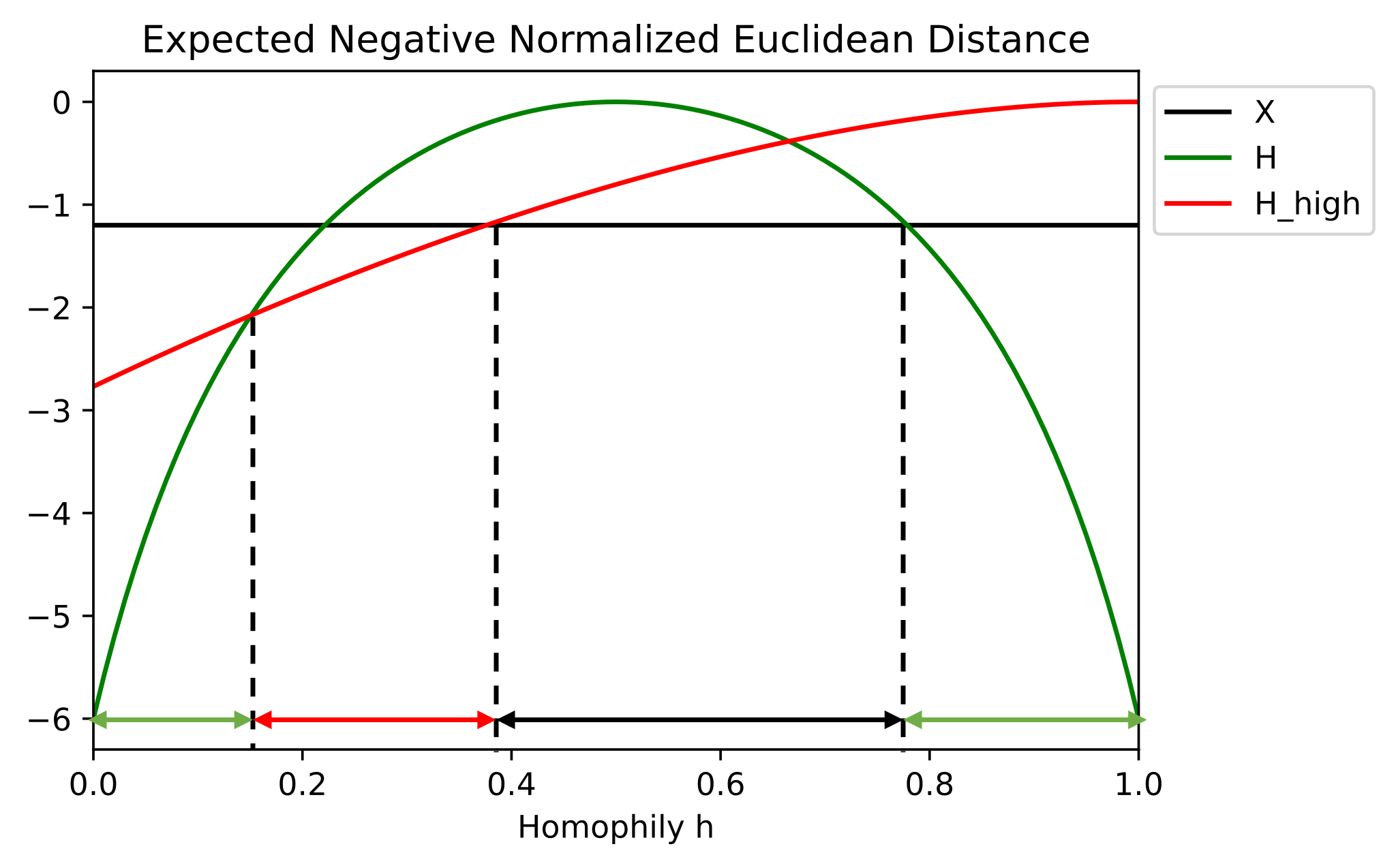}
     } \hspace*{-0.2cm}
     \subfloat[NVR]{
     \captionsetup{justification = centering}
     \includegraphics[width=0.25\textwidth]{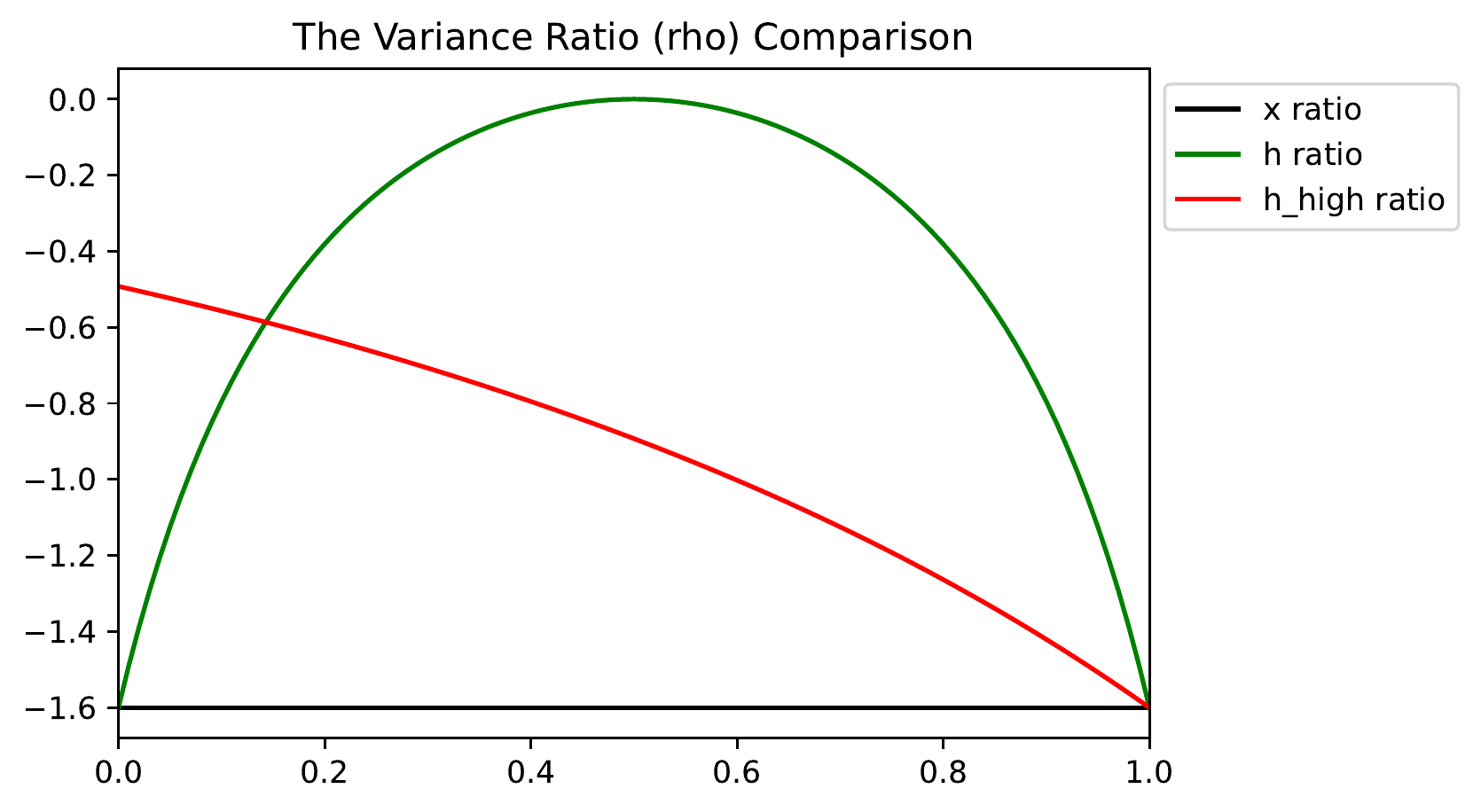}
     } 
     }
     }
     \vspace{-0.2cm}
     \caption{Comparison of CSBM-H with $\sigma^2_0=1,\sigma^2_1=5$.}
     \label{fig:csbmh_sigma0=1_sigma1=5}
\end{figure}
\begin{figure}[htbp!]
    \centering
    \hspace*{-0.3cm}
     \resizebox{1\hsize}{!}{
     \subfloat[PBE]{
     \captionsetup{justification = centering}
     \includegraphics[width=0.27\textwidth]{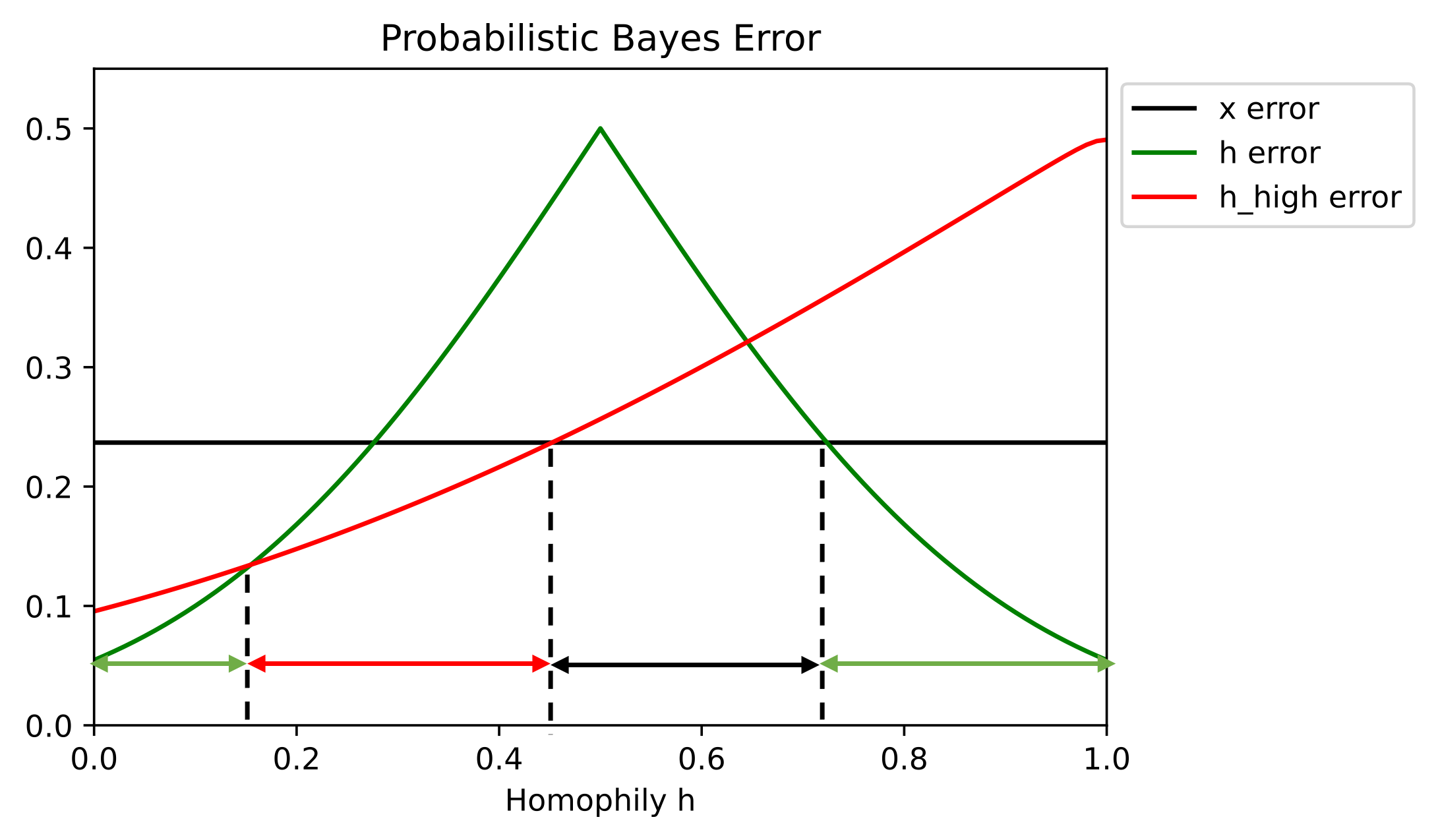}
     } \hspace*{-0.3cm}
     \subfloat[$D_\text{NGJ}$]{
     \captionsetup{justification = centering}
     \includegraphics[width=0.25\textwidth]{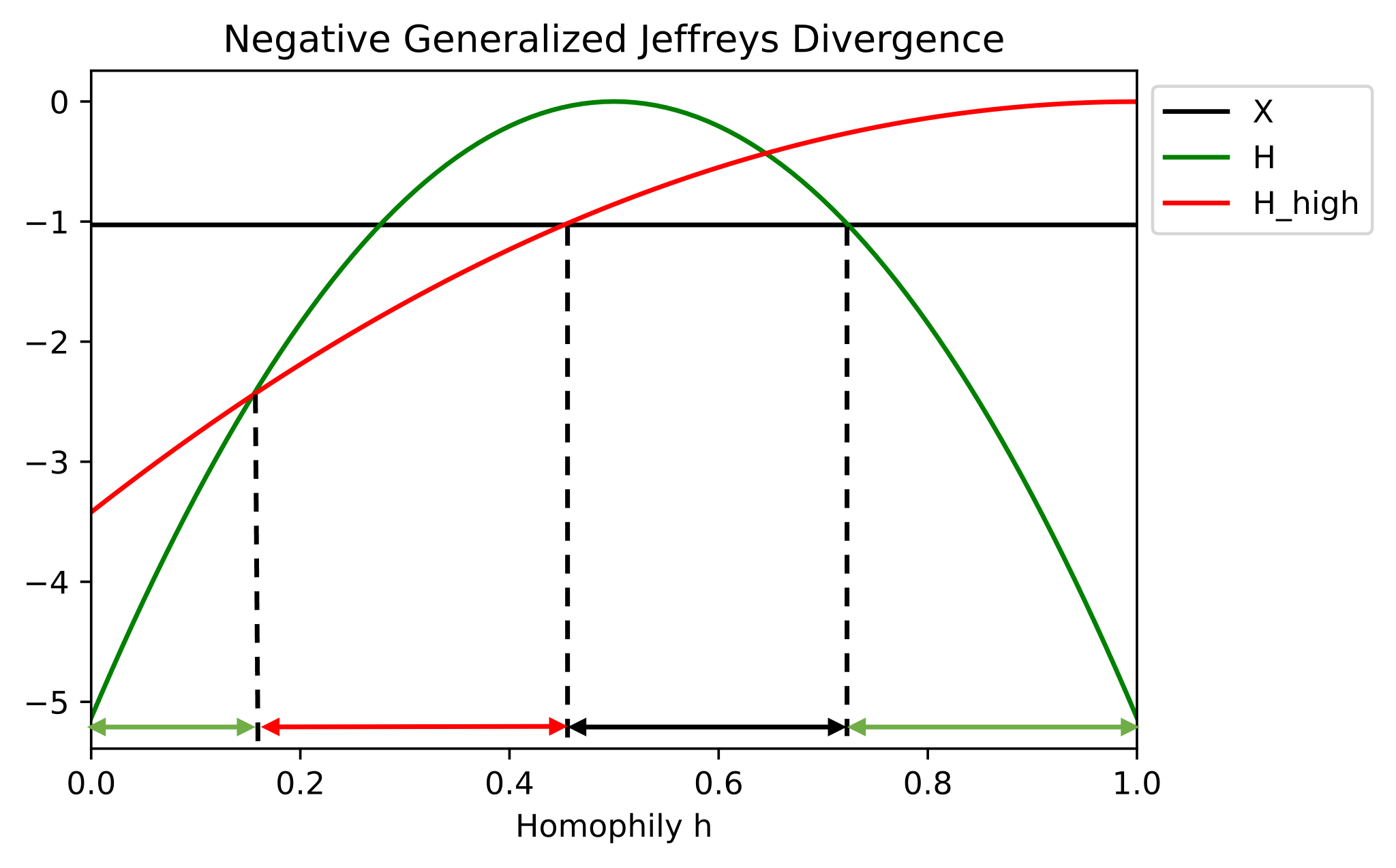}
     }\hspace*{-0.3cm}
     \subfloat[ENND]{
     \captionsetup{justification = centering}
     \includegraphics[width=0.25\textwidth]{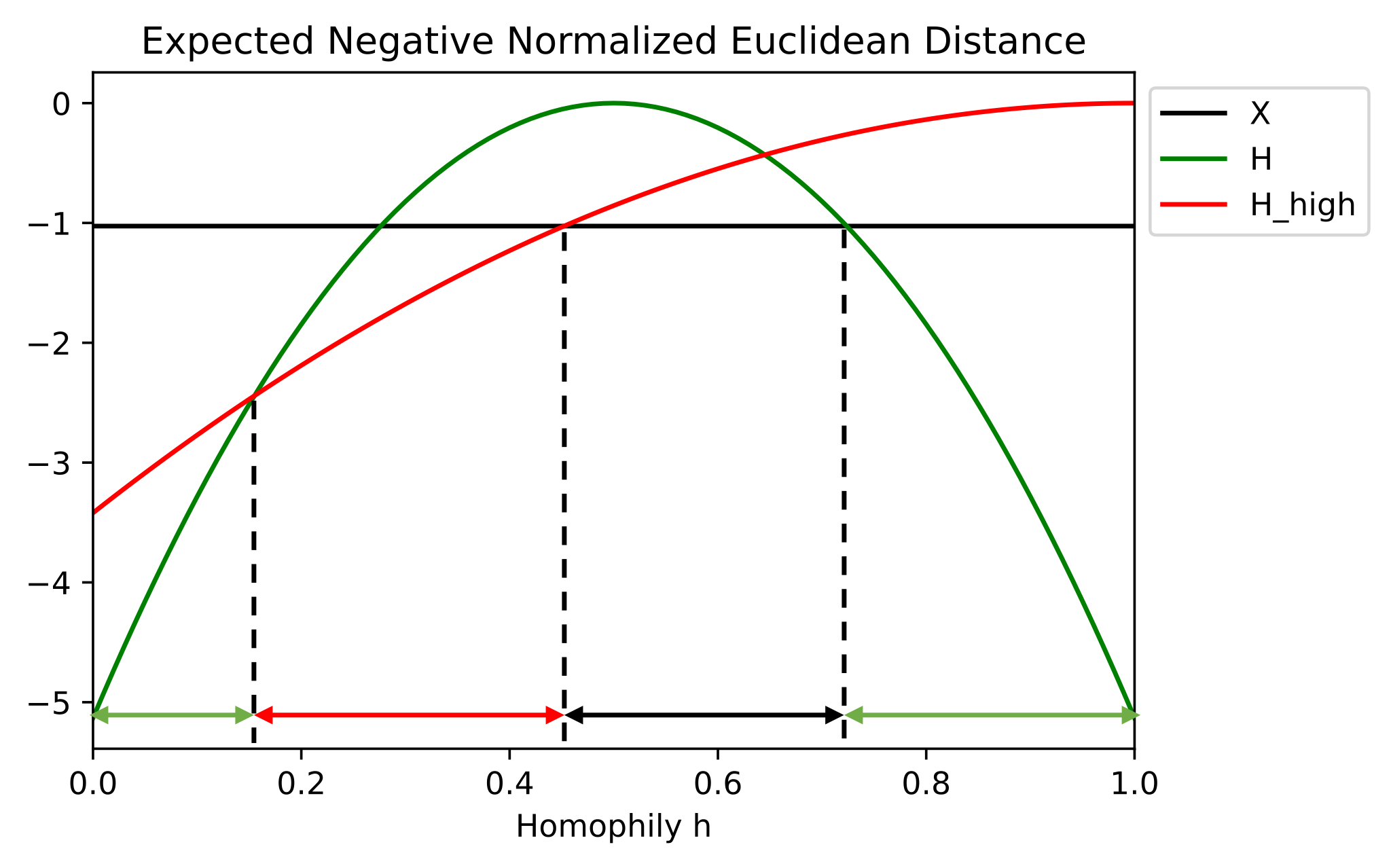}
     } \hspace*{-0.3cm}
     \subfloat[Negative Variance Ratio]{
     \captionsetup{justification = centering}
     \includegraphics[width=0.23\textwidth]{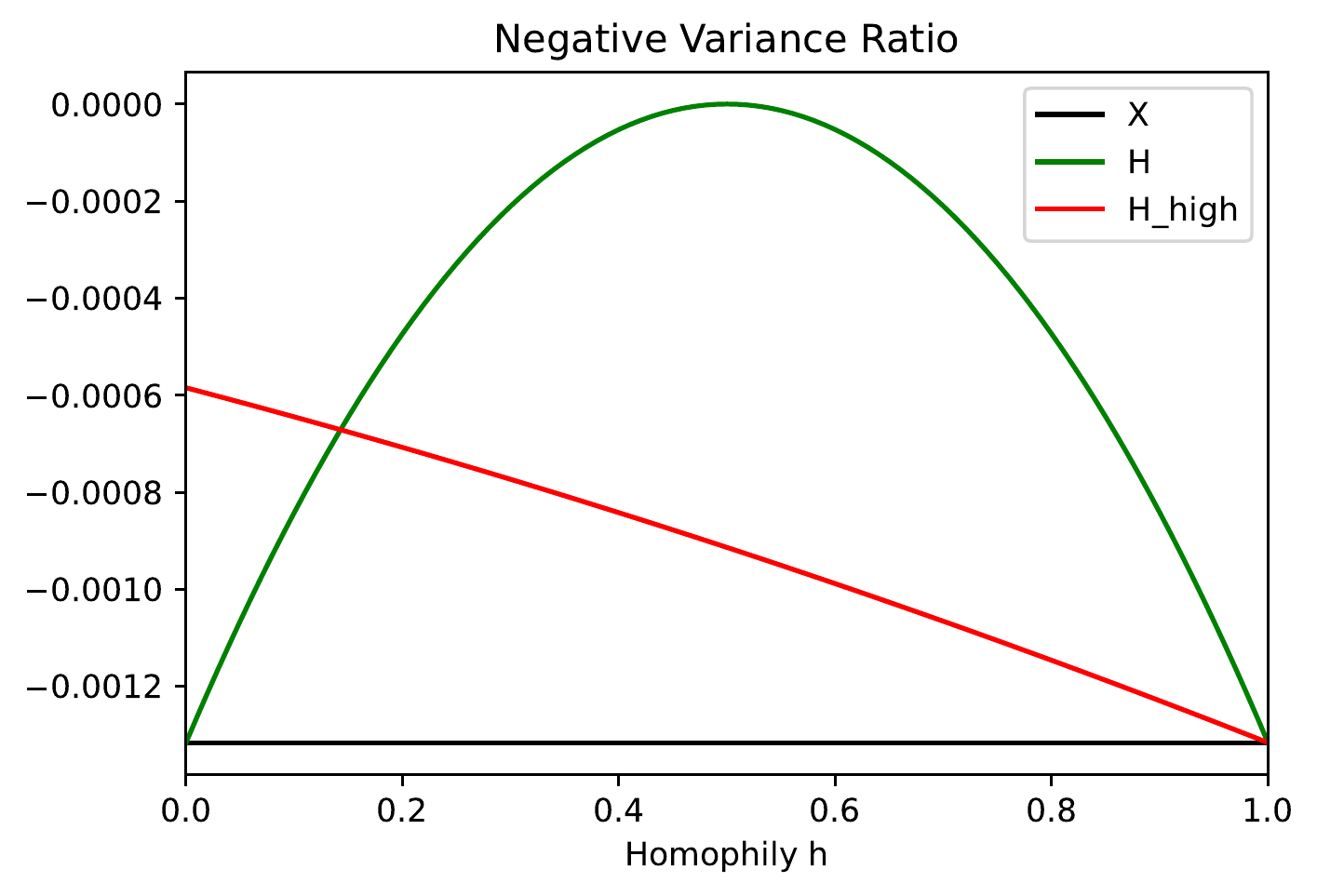}
     } 
     }
     \vspace{-0.2cm}
     \caption{Comparison of CSBM-H with $\sigma^2_0=1.9,\sigma^2_1=2$.}
     \label{fig:csbmh_sigma0=1.9_sigma1=2}
\end{figure}

\noindent \textbf{Increase the Variance of High-variation Class ($\sigma_0^2=1, \sigma_1^2=5$)} \quad  
From Figure~\ref{fig:csbmh_sigma0=1_sigma1=5}, it is observed that as the variance in $\mathcal{C}_1$ increases and the variance between $\mathcal{C}_0$ and $\mathcal{C}_1$ becomes more imbalanced, the PBE and $D_\text{NGJ}$ of the three curves all go up which means the node embeddings become less distinguishable under HP, LP and FP filters. The significant shrinkage of the HP regimes and the expansion of the FP regime indicates that the original features are more robust to imbalanced variances especially in the low homophily area. From Figure \ref{fig:csbmh_sigma0=1_sigma1=5} (d), we can see that the main cause is that the NVR of the 3 curves all move down but the HP curve moves less in low homophily area than other 2 curves. This implies that the HP curve exhibits less sensitivity to $\rho$ within the area of low homophily. 

\noindent \textbf{Increase the Variance of Low-variation Class ($\sigma_0^2=1.9, \sigma_1^2=2$)} \quad 
As shown in Figure \ref{fig:csbmh_sigma0=1.9_sigma1=2}, when the variance in $\mathcal{C}_0$ increases and the variance between $\mathcal{C}_0$ and $\mathcal{C}_1$ becomes more balanced, PBE and $D_\text{NGJ}$ curves go up, which means the node embeddings become less distinguishable. The LP, HP and the FP regimes almost stays the same because the magnitude of NVR becomes too small that it almost has no effect to ND as shown in Figure~\ref{fig:csbmh_sigma0=1.9_sigma1=2} (d).

Interestingly, we found the change of variances cause less differences of the 3 regimes in ENND than that in NVR \footnote{To verify this, we increase $\sigma_0^2$ and $\sigma_1^2$  proportionally. From Figure \ref{fig:csbmh_sigma0=2.5_sigma1=5} in Appendix \ref{appendix:more_figures_csbmh}, relative sizes of the FP, LP, and HP areas remain similar.} and HP filter is less sensitive to $\rho$ changes in low homophily area than LP and FP filters. This insensitivity will have significant impact to the 3 regimes when $\rho$ is close to $0$ and have trivial effect when $\rho$ is close to $1$ because the magnitude of NVR is too small.

\noindent \textbf{Increase the Node Degree of High-variation Class ($d_0=5, d_1=25$)} \quad
From Figure~\ref{fig:csbmh_d0=5_d1=25}, it can be observed that as the node degree of the high-variation class increases, the PBE and $D_\text{NGJ}$ curves of FP and HP filters almost stay the same while the curves of LP filters go down with a large margin. This leads to a substantial expansion of LP regime and shrinkage of FP and HP regime. This is mainly due to the decrease of ENND of LP filters and the decrease of its NVR in low homophily area also plays an important role.

\noindent \textbf{Increase the Node Degree of Low-variation Class ($d_0=25, d_1=5$)} \quad
From Figure~\ref{fig:csbmh_d0=25_d1=5}, we have the similar observation as when we increase the node degree of high-variation class. The difference is that the expansion of LP regime and shrinkage of FP and HP regimes are not as significant as before.

From $\tilde{\sigma}_0^2, \; \tilde{\sigma}_1^2$ we can see that increasing node degree can help LP filter reduce variances of the aggregated features so that the ENND will decrease, especially for high-variation class while HP filter is less sensitive to the change of variances and node degree. 
\begin{figure}[htbp!]
    \centering
     {\hspace*{-0.3cm}
     \resizebox{1\hsize}{!}{
     \subfloat[PBE]{
     \captionsetup{justification = centering}
     \includegraphics[width=0.22\textwidth]{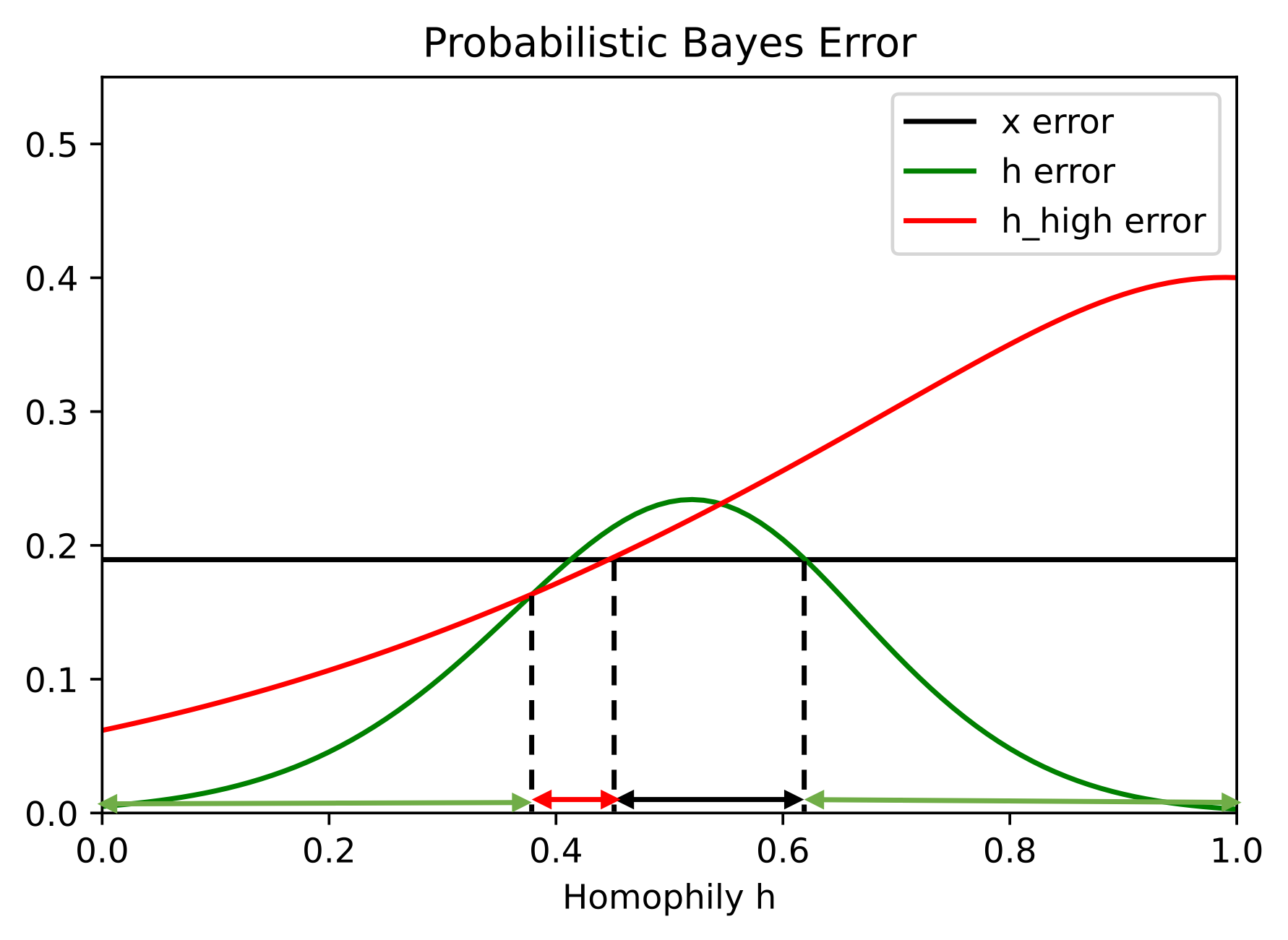}
     }\hspace*{-0.3cm}
     \subfloat[$D_\text{NGJ}$]{
     \captionsetup{justification = centering}
     \includegraphics[width=0.25\textwidth]{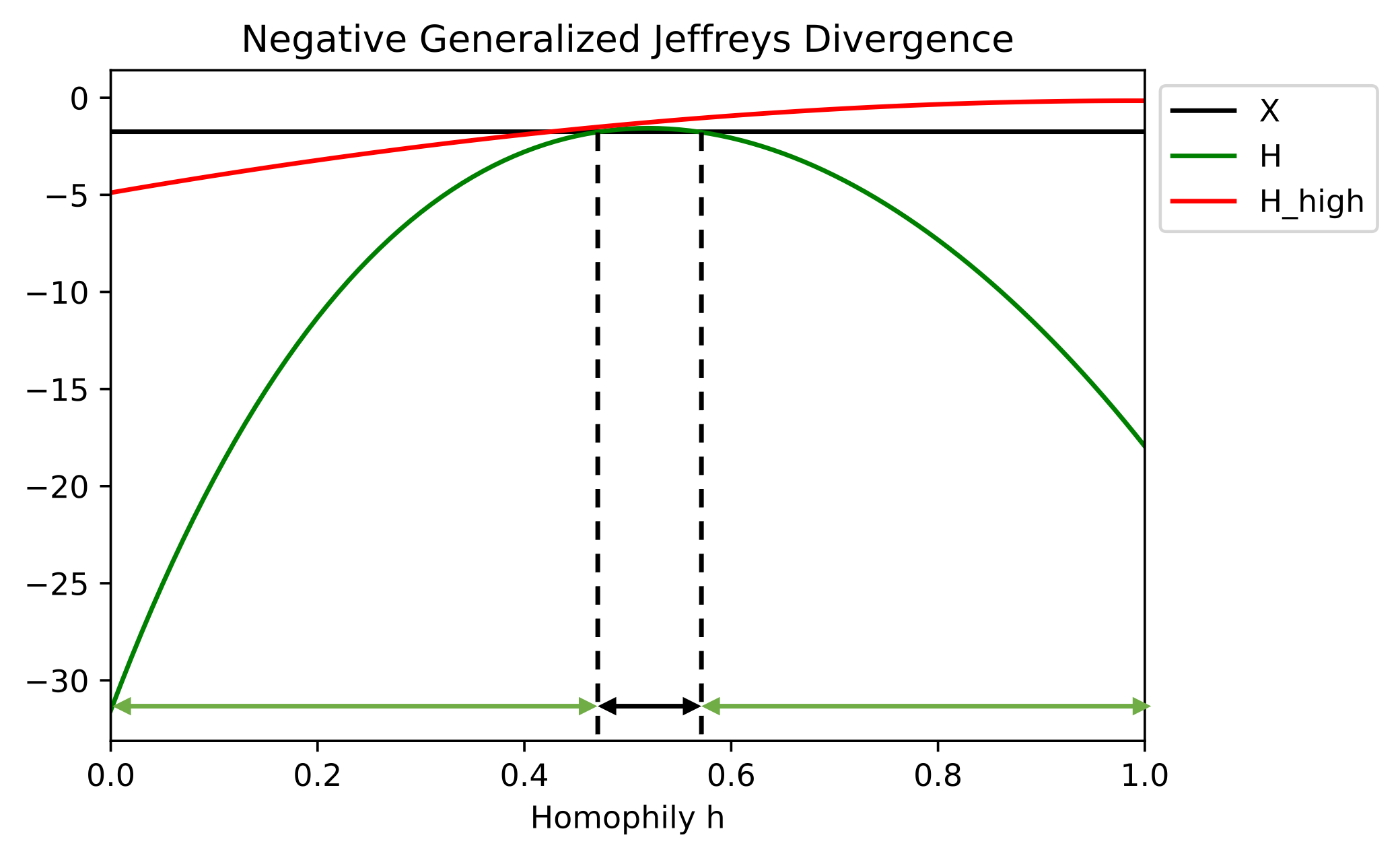}
     } \hspace*{-0.3cm}
     \subfloat[ENND]{
     \captionsetup{justification = centering}
     \includegraphics[width=0.25\textwidth]{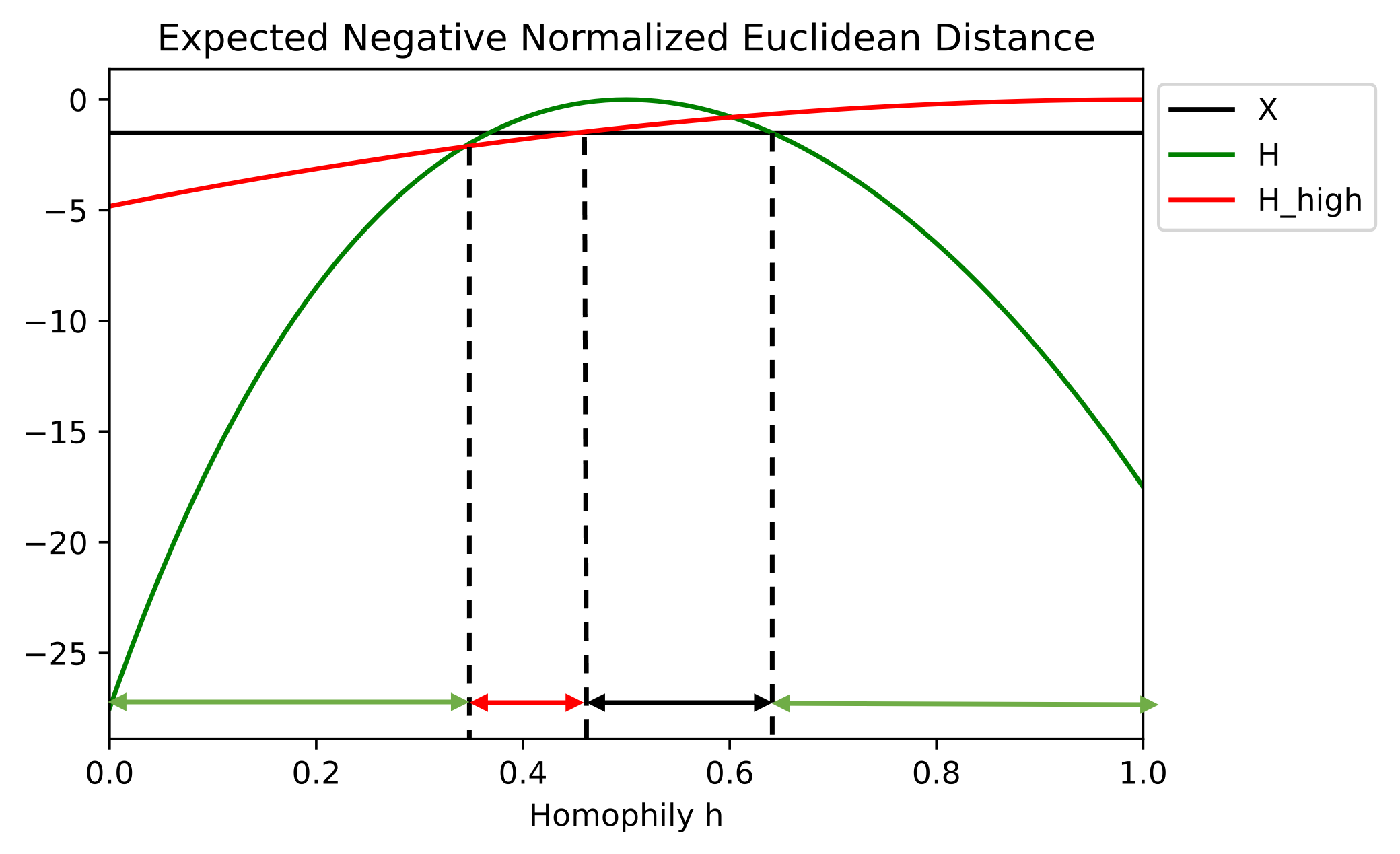}
     } \hspace*{-0.3cm}
     \subfloat[NVR]{
     \captionsetup{justification = centering}
     \includegraphics[width=0.25\textwidth]{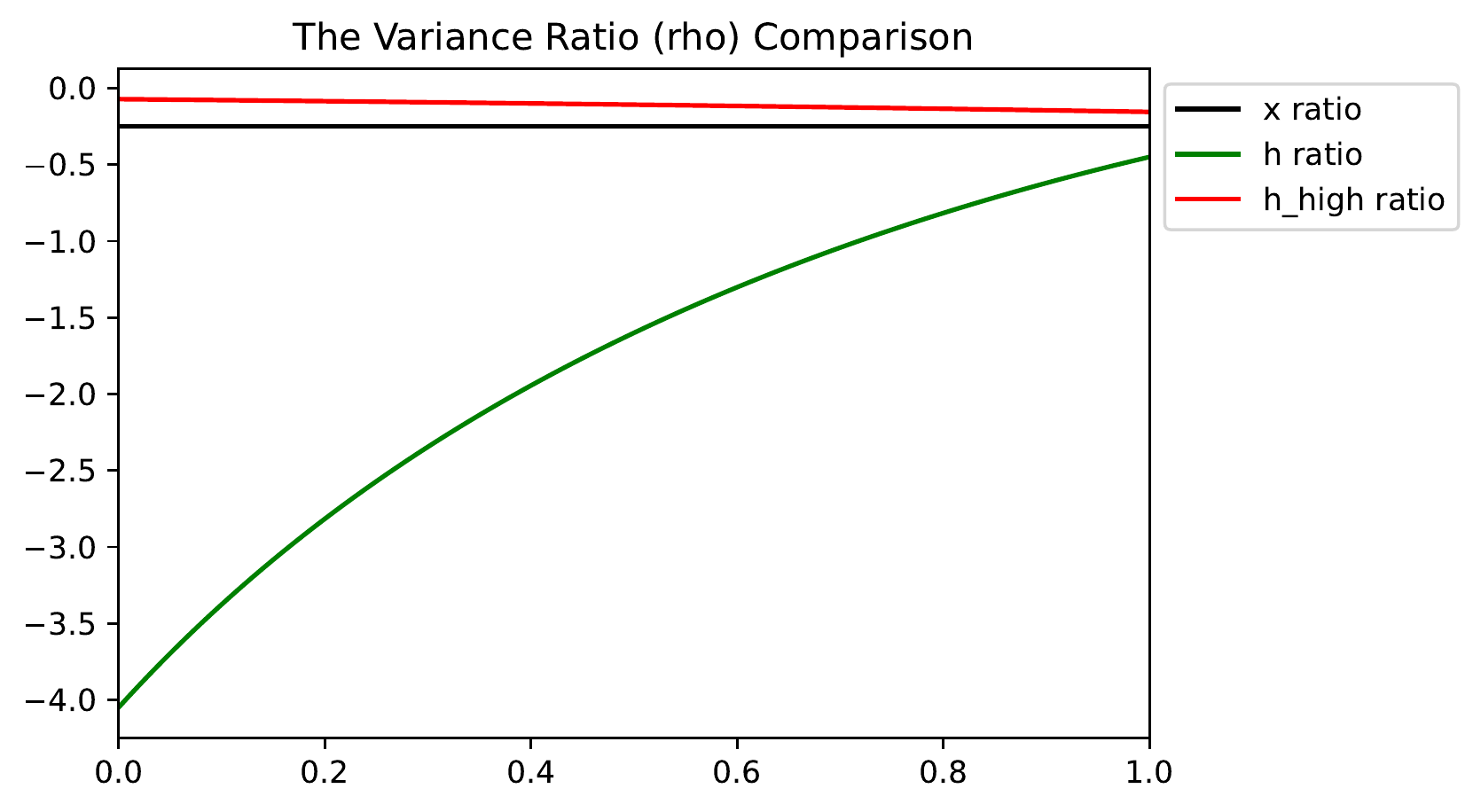}
     }
     }
     }
     \hspace*{-0.2cm}
     \caption{Comparison of CSBM with different $d_0=5,d_1=25$ setups.}
     \label{fig:csbmh_d0=5_d1=25}
\end{figure}

\begin{figure}[htbp!]
    \centering
     {\hspace*{-0.3cm}
     \resizebox{1\hsize}{!}{
     \subfloat[PBE]{
     \captionsetup{justification = centering}
     \includegraphics[width=0.22\textwidth]{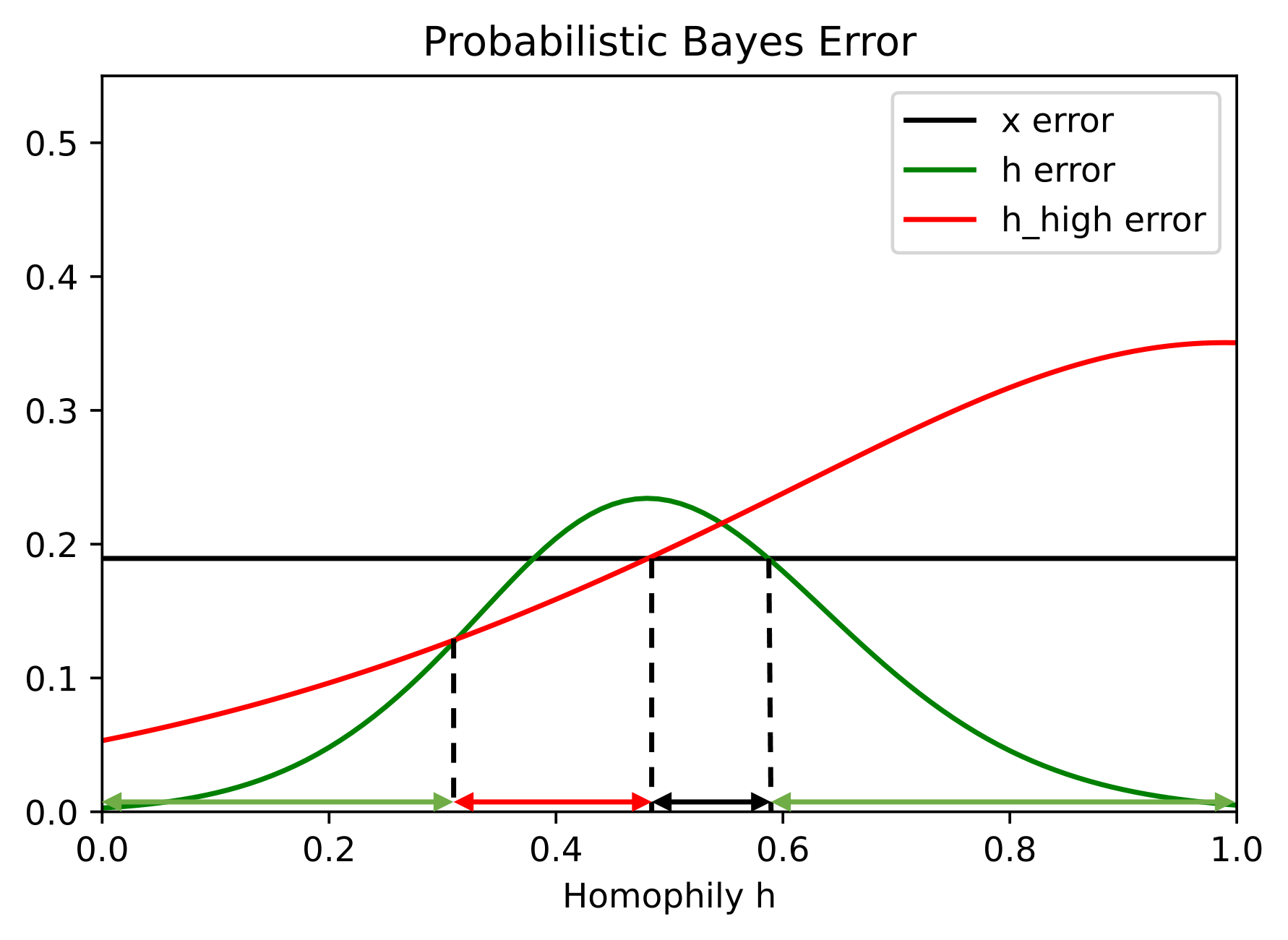}
     }\hspace*{-0.3cm}
     \subfloat[$D_\text{NGJ}$]{
     \captionsetup{justification = centering}
     \includegraphics[width=0.25\textwidth]{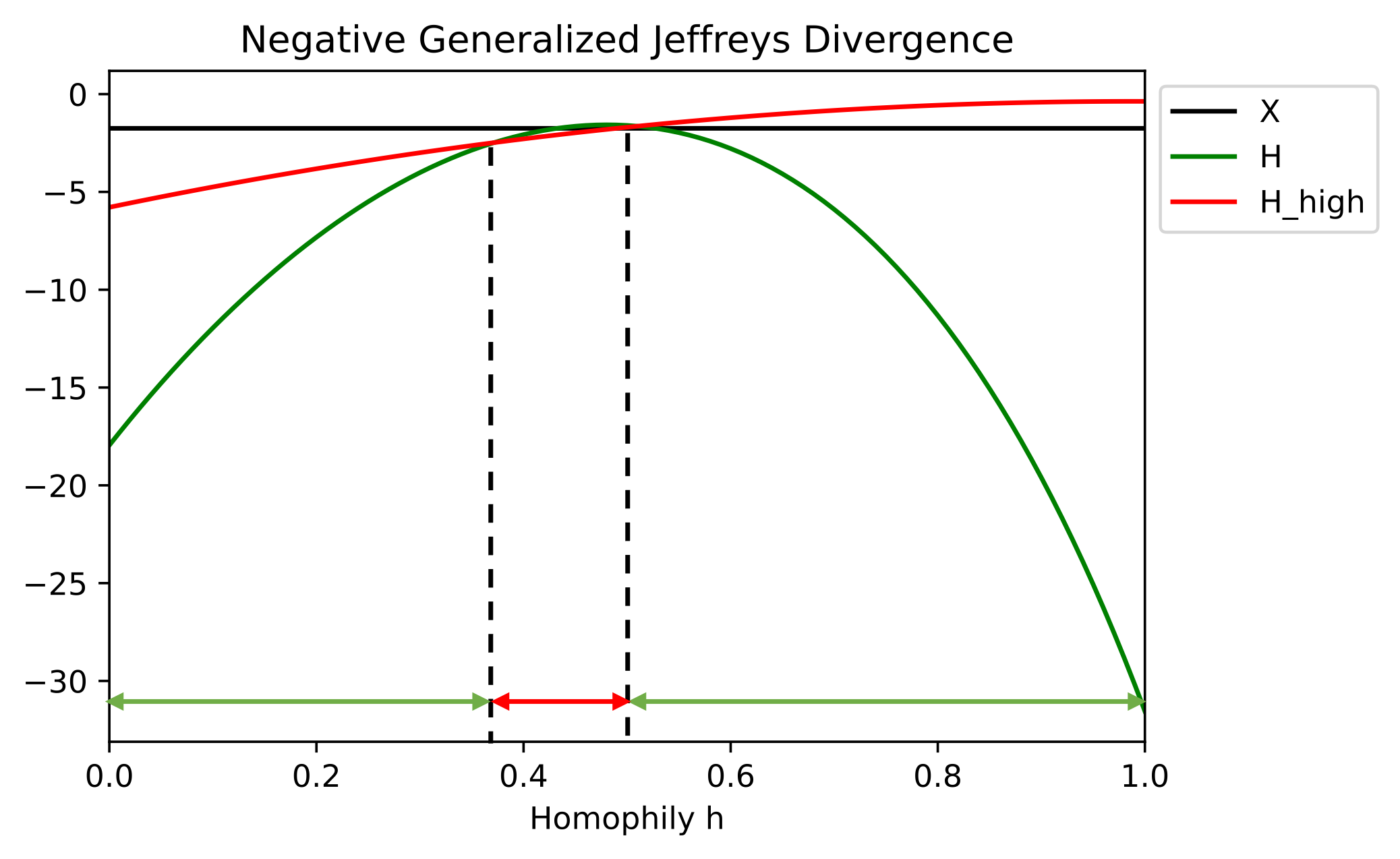}
     } \hspace*{-0.3cm}
     \subfloat[ENND]{
     \captionsetup{justification = centering}
     \includegraphics[width=0.25\textwidth]{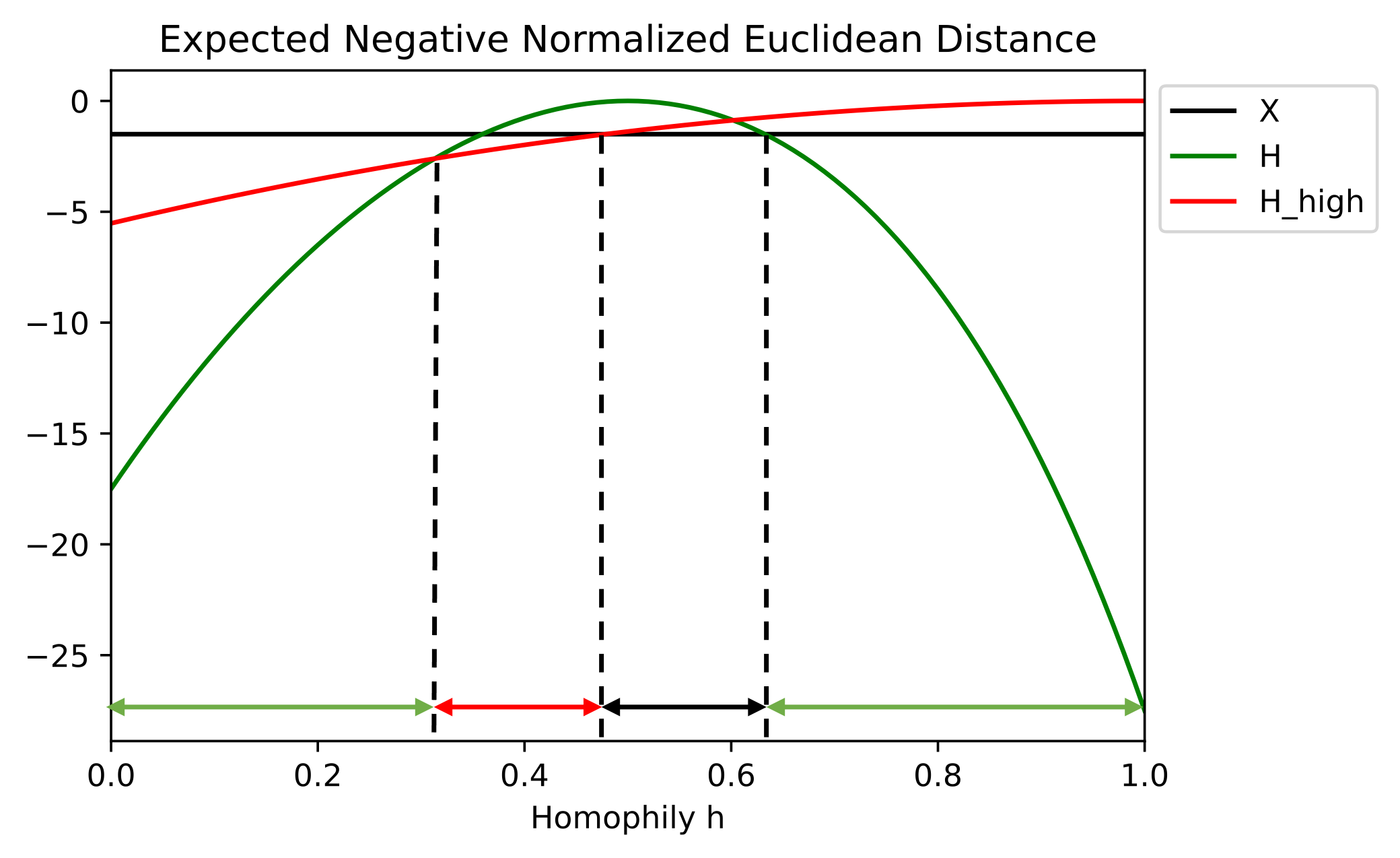}
     } \hspace*{-0.3cm}
     \subfloat[NVR]{
     \captionsetup{justification = centering}
     \includegraphics[width=0.25\textwidth]{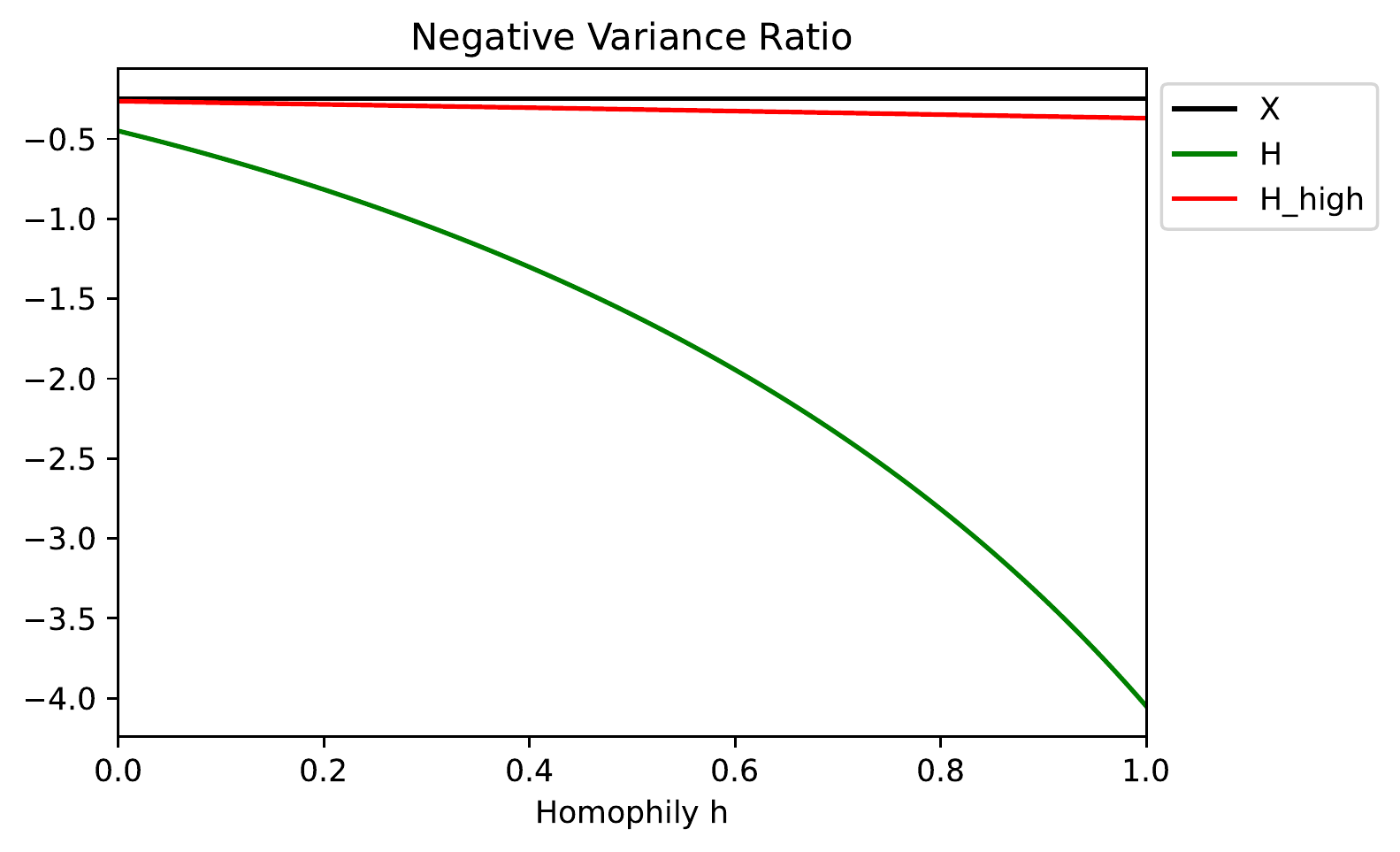}
     }
     }
     }
     \hspace*{-0.2cm}
     \caption{Comparison of CSBM with different $d_0=25,d_1=5$ setups.}
     \label{fig:csbmh_d0=25_d1=5}
\end{figure}

\vspace{-0.1cm}
\subsection{More General Theoretical Analysis}
\vspace{-0.2cm}
\label{sec:general_theoretical_analysis}
Besides the toy example, in this subsection, we aim to gain a deeper understanding of how LP and HP filters affect ND in a broader context beyond the two-normal settings. To be consistent with previous literature, we follow the assumptions outlined in \cite{ma2021homophily}, which are: 1. The features of node $i$ are sampled from distribution $\mathcal{F}_{z_{i}}$, \ie{}, $\bm{x}_{i} \sim \mathcal{F}_{z_{i}}$, with mean $\bm{\mu}_{z_{i}} \in \mathbb{R}^{F_h}$; 2. Dimensions of $\bm{x}_{i}$ are independent to each other; 3. Each dimension in feature $\bm{x}_{i}$ is bounded, \ie{} $a \leq \bm{x}_{i,k} \leq b$; 4. For node $i$, the labels of its neighbors are independently sampled from neighborhood distribution $\mathcal{D}_{z_{i}}$ and repeated for $d_i$ times. We refer to a graph that follows the above assumptions as $\mathcal{G}=\left\{\mathcal{V}, \mathcal{E},\left\{\mathcal{F}_{c}, c \in \mathcal{C}\right\},\left\{\mathcal{D}_{c}, c \in \mathcal{C}\right\}\right\}, \mathcal{C}=\{1,\dots,C\}$ and $(b-a)^2$ reflects how variation the features are. The authors in \cite{ma2021homophily} analyze the distance between the aggregated node embedding and its expectation, \ie{} $\norm{\bm{h}_i - \mathbb{E}(\bm{h}_i)}_2$, which only considers the intra-class ND and has been shown to be inadequate for a comprehensive understanding of ND. Instead, we investigate \textbf{how significant the intra-class embedding distance is smaller than the inter-class embedding distance} in the following theorem, which is a better way to understand ND.

\begin{theorem} 2
Suppose a graph $\mathcal{G}=\left\{\mathcal{V}, \mathcal{E},\left\{\mathcal{F}_{c}, c \in \mathcal{C}\right\},\left\{\mathcal{D}_{c}, c \in \mathcal{C}\right\}\right\}$ meets all the above assumptions (1-4). For nodes $i,j,v \in \mathcal{V}$, suppose $z_i \neq z_j$ and $z_i = z_v$, then for constants $t_x,t_h,t_\text{HP}$ that satisfy $t_x \geq \sqrt{F_h} D_x(i,j), \; t_h \geq \sqrt{F_h} D_h(i,j), \; t_\text{HP} \geq \sqrt{F_h} D_\text{HP}(i,j)$ we have
\vspace{-0.1cm}
\begin{equation}
\begin{aligned}
 & \resizebox{0.7\hsize}{!}{$\mathbb{P}\left(\norm{\bm{x}_i - \bm{x}_j}_2  \geq \norm{\bm{x}_i - \bm{x}_v}_2 + t_x \right)  \leq 2 F_h \exp{\left( -\frac{(D_x(v,j) - \frac{t_x}{\sqrt{F_h}} )^2}{V_x(v,j)} \right)}, $} \\
 & \resizebox{0.7\hsize}{!}{$\mathbb{P}(\norm{\bm{h}_i-\bm{h}_j}_2  \geq \norm{\bm{h}_i-\bm{h}_v}_2 + t_h)  \leq 2 F_h \exp{\left( -\frac{(D_h(v,j) - \frac{t_h}{\sqrt{F_h}} )^2}{V_h(v,j)} \right)},$} \\
& \resizebox{0.7\hsize}{!}{$ \mathbb{P}(\norm{\bm{h}_i^\text{HP} - \bm{h}_j^\text{HP}}_2  \geq \norm{\bm{h}_i^\text{HP} - \bm{h}_v^\text{HP}}_2  + t_\text{HP}) \leq 2 F_h \exp{\left( -\frac{\left(D_\text{HP}(v,j) - \frac{t_\text{HP}}{\sqrt{F_h}} \right)^2}{V_\text{HP}(v,j)} \right)}, $}
\end{aligned}
\end{equation}
\vspace{-0.3cm}
\begin{align*}
 &\resizebox{1\hsize}{!}{$\text{where }D_x(v,j) = \norm{\bm{\mu}_{z_v} - \bm{\mu}_{z_j}}_2, \; V_x(v,j) = (b-a)^2, \ D_h(v,j) = \norm{\tilde{\bm{\mu}}_{z_v} - \tilde{\bm{\mu}}_{z_j}}_2, V_h(v,j) = \left(\frac{1}{2d_v} + \frac{1}{2d_j}\right)(b-a)^2,$} \\
 &\resizebox{1\hsize}{!}{$ D_\text{HP}(v,j) = \norm{\bm{{\mu}}_{z_v} - \bm{\tilde{\mu}}_{z_v} - \left( \bm{{\mu}}_{z_j} - \bm{\tilde{\mu}}_{z_j} \right) }_2, \ V_\text{HP}(v,j) = \left(1+\frac{1}{2d_v} + \frac{1}{2d_j}\right)(b-a)^2, \ \bm{\tilde{\mu}}_{z_v} =  \sum\limits_{\substack{u \in \mathcal{N}(v)} } \mathbb{E}_{\substack{ z_{u} \sim \mathcal{D}_{z_v}, \\ \mathbf{x}_{u} \sim \mathcal{F}_{z_{u}}} } \left[ \frac{1}{d_v}  \bm{x}_{u} \right].$}
\end{align*}
See the proof in Appendix \ref{appendix:proof_theorem2}
\end{theorem}
We can see that, the probability upper bound mainly depends on a distance term (inter-class ND) and normalized variance term (intra-class ND). The normalized variance term of HP filter is less sensitive to the changes of node degree than that of LP filter because there is an additional 1 in the constant term. Moreover, we show that the distance term of HP filter actually depends on the \textbf{relative center distance}, which is a novel discovery. As shown in Figure \ref{fig:relative_center_distance_hp_filter}, when homophily decreases, the aggregated centers will move away from the original centers, and the relative center distance (purple) will get larger which means the embedding distance of nodes from different classes will have larger probability to be big. This explains how HP filter work for some heterophily cases. 
\begin{wrapfigure}{R}{0.45\textwidth}
  \begin{center}
    \includegraphics[width=1.\textwidth]{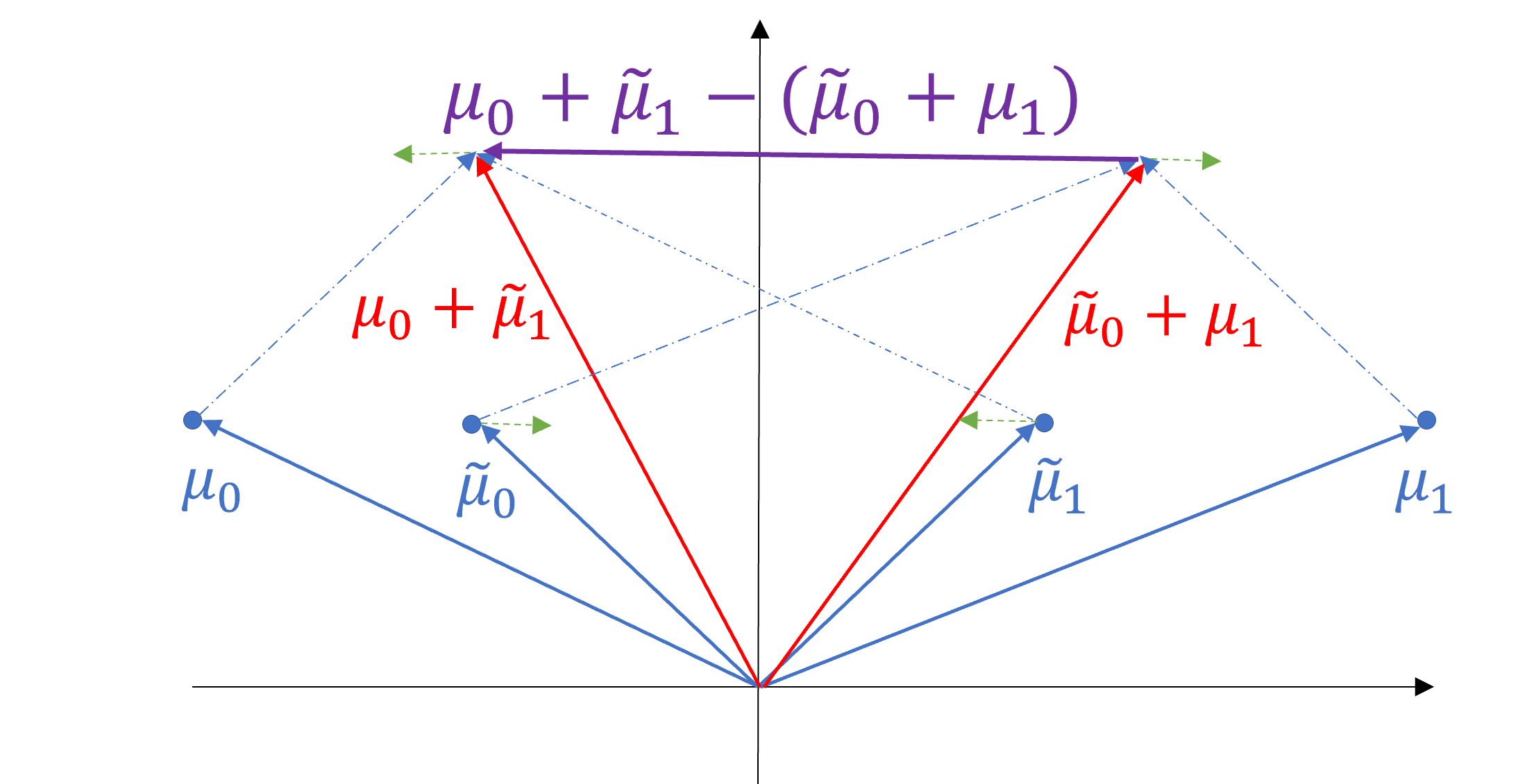}
  \end{center}
  \caption{Demonstration of how HP filter captures the relative center distance.}
  \label{fig:relative_center_distance_hp_filter}
\end{wrapfigure}
Overall, in a more general setting with weaker assumptions,  we can see that ND is also described by the intra- and inter-class ND terms together rather than intra-class ND only, which is consistent with CSBM-H.
\begin{table*}[htbp]
  \centering
  \caption{P-values, homophily values and classifier-based performance metrics on 9 real-world benchmark datasets. Cells marked by grey are incorrect results for both SGC v.s. MLP-1 and GCN v.s. MLP-2 and cells marked by blue are incorrect for 1 of the 2 tests. We use 0.5 as the threshold value of the homophily metrics.}
   \resizebox{1\hsize}{!}{
   \begin{tabular}{c|c|ccccccccc}
    \toprule
    \toprule
          &       & Cornell & Wisconsin & Texas & Film  & Chameleon & Squirrel & Cora  & CiteSeer & PubMed \\
    \midrule
          & $\text{H}_\text{edge}$ & \cellcolor[rgb]{ .647,  .647,  .647}0.5669 & 0.4480 & 0.4106 & 0.3750 & \cellcolor[rgb]{ .647,  .647,  .647}0.2795 & \cellcolor[rgb]{ .647,  .647,  .647}0.2416 & 0.8100 & 0.7362 & \cellcolor[rgb]{ .357,  .608,  .835}0.8024 \\
     & $\text{H}_\text{node}$ & 0.3855 & 0.1498 & 0.0968 & 0.2210 & \cellcolor[rgb]{ .647,  .647,  .647}0.2470 & \cellcolor[rgb]{ .647,  .647,  .647}0.2156 & 0.8252 & 0.7175 & \cellcolor[rgb]{ .357,  .608,  .835}0.7924 \\
    {Baseline} & $\text{H}_\text{class}$ & 0.0468 & 0.0941 & 0.0013 & 0.0110 & \cellcolor[rgb]{ .647,  .647,  .647}0.0620 & \cellcolor[rgb]{ .647,  .647,  .647}0.0254 & 0.7657 & 0.6270 & \cellcolor[rgb]{ .357,  .608,  .835}0.6641 \\
    { Homophily} & $\text{H}_\text{agg}$ & \cellcolor[rgb]{ .647,  .647,  .647}0.8032 & \cellcolor[rgb]{ .647,  .647,  .647}0.7768 & \cellcolor[rgb]{ .647,  .647,  .647}0.694 & \cellcolor[rgb]{ .647,  .647,  .647}0.6822 & 0.61  & \cellcolor[rgb]{ .647,  .647,  .647}0.3566 & 0.9904 & 0.9826 & \cellcolor[rgb]{ .357,  .608,  .835}0.9432 \\
          { Metrics}& $\text{H}_\text{GE}$ & 0.31  & 0.34  & 0.35  & 0.16  & \cellcolor[rgb]{ .647,  .647,  .647}0.0152 & \cellcolor[rgb]{ .647,  .647,  .647}0.0157 & \cellcolor[rgb]{ .647,  .647,  .647}0.17 & \cellcolor[rgb]{ .647,  .647,  .647}0.19 & \cellcolor[rgb]{ .357,  .608,  .835}0.27 \\
    & $\text{H}_\text{adj}$ & 0.1889 & 0.0826 & 0.0258 & 0.1272 & \cellcolor[rgb]{ .647,  .647,  .647}0.0663 & \cellcolor[rgb]{ .647,  .647,  .647}0.0196 & 0.8178 & 0.7588 & \cellcolor[rgb]{ .357,  .608,  .835}0.7431 \\
    & $\text{LI}$    & 0.0169 & 0.1311 & 0.1923 & 0.0002 & \cellcolor[rgb]{ .647,  .647,  .647}0.048 & \cellcolor[rgb]{ .647,  .647,  .647}0.0015 & 0.5904 & \cellcolor[rgb]{ .647,  .647,  .647}0.4508 & \cellcolor[rgb]{ .357,  .608,  .835}0.4093 \\
    \midrule
    {Classifier-based } & $\text{KR}_\text{NNGP}$ & 0.00  & 0.00  & 0.00  & 0.00  & 1.00  & 1.00  & 1.00  & 1.00  & \cellcolor[rgb]{ .357,  .608,  .835}1.00 \\
    {Performance Metrics} & GNB & 0.00  & 0.00  & 0.00  & 0.00  & 1.00  & 1.00  & 1.00  & 1.00  & \cellcolor[rgb]{ .357,  .608,  .835}1.00 \\
    \midrule
    \multicolumn{1}{c|}{\multirow{4}[2]{*}{SGC v.s. MLP-1}} & p-value & 0.00  & 0.00  & 0.00  & 0.00  & 1.00  & 1.00  & 1.00  & 1.00  & 0.00 \\
          & ACC SGC & 70.98 $\pm$ 8.39 & 70.38 $\pm$ 2.85 & 83.28 $\pm$ 5.43 & 25.26 $\pm$ 1.18 & 64.86 $\pm$ 1.81 & 47.62 $\pm$ 1.27 & 85.12 $\pm$ 1.64 & 79.66 $\pm$ 0.75 & 85.5 $\pm$ 0.76 \\
          & ACC MLP-1 & 93.77 $\pm$ 3.34 & 93.87 $\pm$ 3.33 & 93.77 $\pm$ 3.34 & 34.53 $\pm$ 1.48 & 45.01 $\pm$ 1.58 & 29.17 $\pm$ 1.46 & 74.3 $\pm$ 1.27 & 75.51 $\pm$ 1.35 & 86.23 $\pm$ 0.54 \\
          & \textbf{Diff Acc} & -22.79 & -23.49 & -10.49 & -9.27 & 19.85 & 18.45 & 10.82 & 4.15  & -0.73 \\
    \midrule
    \multicolumn{1}{c|}{\multirow{4}[2]{*}{GCN v.s. MLP-2}} &  p-value & 0.00  & 0.00  & 0.00  & 0.00  & 1.00  & 1.00  & 1.00  & 1.00  & \cellcolor[rgb]{ .647,  .647,  .647}0.00 \\
          & ACC GCN & 82.46 $\pm$ 3.11 & 75.5 $\pm$ 2.92 & 83.11 $\pm$ 3.2 & 35.51 $\pm$ 0.99 & 64.18 $\pm$ 2.62 & 44.76 $\pm$ 1.39 & 87.78 $\pm$ 0.96 & 81.39 $\pm$ 1.23 & 88.9 $\pm$ 0.32 \\
          & ACC MLP-2 & 91.30 $\pm$ 0.70 & 93.87 $\pm$ 3.33 & 92.26 $\pm$ 0.71 & 38.58 $\pm$ 0.25 & 46.72 $\pm$ 0.46 & 31.28 $\pm$ 0.27 & 76.44 $\pm$ 0.30 & 76.25 $\pm$ 0.28 & 86.43 $\pm$ 0.13 \\
          & \textbf{Diff Acc} & -8.84 & -18.37 & -9.15 & -3.07 & 17.46 & 13.48 & 11.34 & 5.14  & 2.47 \\
    \bottomrule
    \bottomrule
    \end{tabular}%
    }
  \vspace{-0.3cm}\label{tab:hypothesis_testing_kernel_homo_performance_comparison}%
\end{table*}%
\vspace{-0.4cm}
\section{Empirical Study of Node Distinguishability}
\vspace{-0.2cm}
\label{sec:empirical_study}
Besides theoretical analysis, in this section, we will conduct experiments to verify whether the effect of homophily on the performance of GNNs really relates to its effect on ND. If a strong relation can be verified, then it indicates that we can design new training-free ND-based performance metrics beyond homophily metrics, to evaluate the superiority and inferiority of G-aware models against its coupled G-agnostic models.
\vspace{-0.0cm}
\subsection{Hypothesis Testing on Real-world Datasets}
\vspace{-0.0cm}
\label{sec:hypothesis_testing}

To test whether "intra-class embedding distance is smaller than the inter-class embedding distance" strongly relates to the superiority of G-aware models to their coupled G-agnostic models in practice, we conduct the following hypothesis testing \footnote{Authors in \cite{luan2022we} also conduct hypothesis testing to find out when to use GNNs for node classification, but they test the differences between connected nodes and unconnected nodes instead of intra- and inter-class nodes.}.

\noindent \textbf{Experimental Setup} \quad
We first train two G-aware models GCN, SGC-1 and their coupled G-agnostic models MLP-2 and MLP-1 with fine-tuned hyperparameters provided by \cite{luan2022revisiting}. For each trained model, we calculate the pairwise Euclidean distance of the node embeddings in output layers. Next, we compute the proportion of nodes whose intra-class node distance is significantly smaller than inter-class node distance \footnote{A node is considered as "significantly smaller" when the p-value for its intra-class node distance being smaller than inter-class node distance is smaller than 0.05. In other words, this node is considered as significantly distinguishable. This second statistical test is necessary to avoid noisy nodes. In practice, we noticed that the ratio of intra-class node distance to inter-class node distance is roughly 1 for lots of nodes. This is particularly evident when the labels are sparse and when we use sampling method. It will not only cause instability of the outputs, but also result in false results sometimes. Thus, we don't want to take account these "marginal nodes" into the comparison of Prop values and we found that using another hypothesis test would be helpful.} \eg{} we obtain $\text{Prop}(\text{GCN})$ for GCN. We use Prop to quantify ND and we train the models multiple times for samples to conduct the following hypothesis tests: 
\begin{align*}
    &\resizebox{1\hsize}{!}{$\text{H}_0: \text{Prop}(\text{G-aware model}) \geq \text{Prop}(\text{G-agnostic model}); \ \text{H}_1: \text{Prop}(\text{G-aware model})<\text{Prop}(\text{G-agnostic model})$}
\end{align*}
Specifically, we compare GCN vs. MLP-2 and SGC-1 vs. MLP-1 on $9$ widely used benchmark datasets with different homophily values for 100 times. In each time, we randomly split the data into training/validation/test sets with a ratio of 60\%/20\%/20\%. For the 100 samples, we conduct \textit{T-test for the means of two independent samples of scores}, and obtain the corresponding p-values. The test results and model performance comparisons are shown in Table~\ref{tab:hypothesis_testing_kernel_homo_performance_comparison} (See more experimental tests on state-of-the-art model in Appendix \ref{appendix:detailed_discussion_performance_metrics_more_experimental_results}).

It is observed that, in most cases (except for GCN  vs. MLP-2 on \textit{PubMed} \footnote{We discuss this special case in Appendix \ref{appendix:statistics_comparisons}, together with some similar inconsistency instances found on large-scale datasets.}), when $\text{H}_1$ significantly holds, G-aware models will underperform the coupled G-agnostic models and vice versa. This supports our claim that the performance of G-aware models is closely related to "intra-class vs. inter-class node embedding distances", no matter the homophily levels. It reminds us that the p-value can be a better performance metric for GNNs beyond homophily. Moreover, the p-value can provide a statistical threshold, such as $p \leq 0.05$. This property is not present in existing homophily metrics.

However, it is required to train and fine-tune the models to obtain the p-values, which make it less practical because of computational costs. To overcome this issue, in the next subsection, we propose a classifier-based performance metric that can provide p-values without training.

\vspace{-0.0cm}
\subsection{Towards A Better Metric Beyond Homophily: Classifier-based Performance Metric}
\vspace{-0.0cm}
A qualified classifier should not require iterative training. In this paper, we choose Gaussian Naïve Bayes (GNB)\cite{hastie2009elements} and Kernel Regression (KR) with Neural Network Gaussian Process (NNGP) \cite{lee2017deep,arora2019exact, garriga2018deep,matthews2018gaussian} to capture the \textbf{feature-based linear or non-linear} information.

To get the p-value efficiently, we first randomly sample 500 labeled nodes from $\mathcal{V}$ and splits them into 60\%/40\% as "training" and "test" data. The original features $X$ and aggregated features $H$ of the sampled training and test nodes can be calculated and are then fed into a given classifier. The predicted results and prediction accuracy of the test nodes will be computed directly with feedforward method. We repeat this process for 100 times to get 100 samples of prediction accuracy for $X$ and $H$. Then, for the given classifier, we compute the p-value of the following hypothesis testing, 
\begin{align*}
    &\text{H}_0: \text{Acc}(\text{Classifier}(H)) \geq \text{Acc}(\text{Classifier}(X)); \ \text{H}_1: \text{Acc}(\text{Classifier}(H)) < \text{Acc}(\text{Classifier}(X)) 
\end{align*}

The p-value can provide a statistical threshold value, such as 0.05, to indicate whether $H$ is significantly better than $X$ for node classification. As seen in Table \ref{tab:hypothesis_testing_kernel_homo_performance_comparison}, KR and GNB based metrics significantly outperform the existing homophily metrics, reducing the errors from at least $5$ down to just $1$ out of 18 cases. Besides, we only need a small set of the labels to calculate the p-value, which makes it better for sparse label scenario. Table \ref{tab:homophily_metrics_comparison} summarizes its advantages over the existing metrics. (See Appendix \ref{appendix:detailed_discussion_performance_metrics_more_experimental_results} for more details on classifier-based performance metrics, experiments on synthetic datasets, more detailed comparisons on small-scale and large-scale datasets, discrepancy between linear and non-linear models, results for symmetric renormalized affinity matrix and running time.)

\begin{wraptable}{r}{0.55\textwidth}
  \centering
  \caption{Property comparisons of performance metrics}
    \scalebox{.7}{
        \begin{tabular}{c|cccc}
    \toprule
    \toprule
    \multirow{2}[2]{*}{}Performance &{Linear or } & Feature  & Sparse  & Statistical \\
       Metrics   &  Non-linear    &   Dependency    &   Labels    & Threshold \\
    \midrule
    $\text{H}_\text{node}$ & linear &  \XSolidBrush  &  \XSolidBrush  & \XSolidBrush \\
    $\text{H}_\text{edge}$ & linear &   \XSolidBrush  &  \XSolidBrush  & \XSolidBrush \\
    $\text{H}_\text{class} $ & linear & \XSolidBrush   &  \XSolidBrush  & \XSolidBrush \\
    $\text{H}_\text{agg}$ & linear &  \XSolidBrush   & \Checkmark & \XSolidBrush \\
    $\text{H}_\text{GE}$ & linear & \Checkmark & \Checkmark & \XSolidBrush \\
    $\text{H}_\text{adj}$ & linear &  \XSolidBrush &  \XSolidBrush & \XSolidBrush \\
    $\text{LI}$ & linear &  \XSolidBrush &  \XSolidBrush & \XSolidBrush \\
    Classifier & both & \Checkmark & \Checkmark & \Checkmark \\
    \bottomrule
    \bottomrule
    \end{tabular}%
  \label{tab:homophily_metrics_comparison}%
  }
\end{wraptable} 
\section{Conclusions}
In this paper, we provide a complete understanding of homophily by studying intra- and inter-class ND together. To theoretically investigate ND, we study the PBE and $D_\text{NGJ}$ of the proposed CSBM-H and analyze how graph filters, class variances and node degree distributions will influence the PBE and $D_\text{NGJ}$ curves and the FP, LP, HP regimes. We extend the investigation to broader settings with weaker assumptions and theoretically prove that ND is indeed affected by both intra- and inter-class ND. We also discover that the effect of HP filter depends on the relative center distance. Empirically, through hypothesis testing, we corroborate that the performance of GNNs versus NNs is closely related to whether intra-class node embedding "distance" is smaller than inter-class node embedding "distance". We find that the p-value is a much more effective performance metric beyond homophily metrics on revealing the advantage and disadvantage of GNNs. Based on this observation, we propose classifier-based performance metric, which is a non-linear feature-based metric and can provide statistical threshold value.

\newpage
\section{Reproducibility and Blogs}
\begin{itemize}
    \item Code: \url{https://github.com/SitaoLuan/When-Do-GNNs-Help}. 
    \item Blog in English (on Medium): \url{https://medium.com/SitaoLuan/when-should-we-use-graph-neural-networks-for-node-classification-8ce77a772085}.
    \item Blog in Chinese (on Zhihu): \url{https://zhuanlan.zhihu.com/p/653631858}. 
\end{itemize}

\section{Acknowledgements}
This work was partially supported by the Natural Sciences and Engineering Research Council of Canada (NSERC) Grant RGPIN-2023-04125, RGPIN 2389 and Canadian Institute for Advanced Research (CIFAR) Grant CIFAR FS20-126, CIFAR 10450. Minkai Xu thanks the generous support of Sequoia Capital Stanford Graduate Fellowship.

\bibliography{references.bib}
\bibliographystyle{abbrv}

\newpage
\appendix
\onecolumn
\section{Proof of Theorem 1}
\label{appendix:proof_of_theorem1}

\begin{theorem} 1 Suppose $\sigma_0^2 \neq \sigma_1^2$ and $\sigma_0^2, \sigma_1^2 > 0$, the prior distribution for $\bm{x}_i$ is $\mathbb{P}(\bm{x}_i\in {\cal C}_0) = \mathbb{P}(\bm{x}_i\in {\cal C}_1) = 1/2$, then the optimal Bayes Classifier ($\text{CL}_{\text{Bayes}}$) for CSBM-H ($\bm{\mu}_0,\bm{\mu}_1,\sigma_0^2 I,\sigma_1^2 I,{d}_0,{d}_1,{h}$) is 
$$\text{CL}_\text{Bayes}(\bm{x}_i) = \left\{
\begin{aligned}
1,\; \eta(\bm{x}_i) \geq 0.5\\
0,\; \eta(\bm{x}_i) < 0.5
\end{aligned}
\right. ,\; \eta(\bm{x}_i) = \mathbb{P}(z_i=1|\bm{x}_i) = \frac{1}{1+\exp{\left(Q(\bm{x}_i)\right)}}, $$
where $Q(\bm{x}_i) = a \bm{x}_i^\top \bm{x}_i + \bm{b}^\top \bm{x}_i +c,\ 
a = \frac{1}{2}\left(\frac{1}{\sigma_1^2} - \frac{1}{\sigma_0^2}\right), 
\bm{b}=\frac{\bm{\mu}_0}{\sigma_0^2}-\frac{\bm{\mu}_1}{\sigma_1^2}, 
c=\frac{\bm{\mu}_1^\top\bm{\mu}_1}{2\sigma_1^2} - \frac{\bm{\mu}_0^\top\bm{\mu}_0}{2\sigma_0^2} + \ln{\left(\frac{\sigma_1^{F_h} }{\sigma_0^{F_h} } \right)}$. 
\end{theorem}

\begin{proof}
Since the prior distribution for $\text{CL}_{\text{Bayes}}$ is 
\begin{align*}
   \mathbb{P}(\bm{x}_i\in {\cal C}_0) = \mathbb{P}(\bm{x}_i\in {\cal C}_1) = \frac{1}{2}
\end{align*}
\begin{align*}
   & \text{CL}_\text{Bayes}(z_i=1|\bm{x}_i) = \frac{\mathbb{P}(z_i=1,\bm{x}_i)}{\mathbb{P}(\bm{x}_i)} = \frac{\mathbb{P}(z_i=1)\mathbb{P}(\bm{x}_i|z_i=1)}{\mathbb{P}(z_i=0)\mathbb{P}(\bm{x}_i|z_i=0)+\mathbb{P}(z_i=1)\mathbb{P}(\bm{x}_i|z_i=1)}\\
   & = \frac{1}{1 + \frac{\mathbb{P}(z_i=0)\mathbb{P}(\bm{x}_i|z_i=0)}{\mathbb{P}(z_i=1)\mathbb{P}(\bm{x}_i|z_i=1)}} = \frac{1}{1 + \frac{ (2\pi)^{-F_h/2}\det{\left(\sigma_0^2I\right)}^{-1/2} \exp{\left(-\frac{1}{2\sigma_0^2}(\bm{x}_i - \bm{\mu}_0)^\top(\bm{x}_i - \bm{\mu}_0)  \right)}}{ (2\pi)^{-F_h/2}\det{\left(\sigma_1^2I\right)}^{-1/2} \exp{\left(-\frac{1}{2\sigma_1^2}(\bm{x}_i - \bm{\mu}_1)^\top(\bm{x}_i - \bm{\mu}_1)  \right)}}} \\
   &= \frac{1}{1 + \frac{ \sigma_0^{-F_h} }{ \sigma_1^{-F_h} } \exp{\left(-\frac{1}{2\sigma_0^2}(\bm{x}_i - \bm{\mu}_0)^\top(\bm{x}_i - \bm{\mu}_0) + \frac{1}{2\sigma_1^2}(\bm{x}_i - \bm{\mu}_1)^\top(\bm{x}_i - \bm{\mu}_1) \right)} }\\
   &= \frac{1}{1 + \frac{ \sigma_0^{-F_h} }{ \sigma_1^{-F_h} } \exp{\left(-\frac{1}{2\sigma_0^2}(\bm{x}_i^\top \bm{x}_i - 2\bm{\mu}_0^\top \bm{x}_i +\bm{\mu}_0^\top \bm{\mu}_0) + \frac{1}{2\sigma_1^2}(\bm{x}_i^\top \bm{x}_i - 2\bm{\mu}_1^\top \bm{x}_i + \bm{\mu}_1^\top\bm{\mu}_1) \right)} }\\
   &= \frac{1}{1 + \exp{\left((\frac{1}{2\sigma_1^2}-\frac{1}{2\sigma_0^2})\bm{x}_i^\top \bm{x}_i + (\frac{\bm{\mu}_0}{\sigma_0^2} - \frac{\bm{\mu}_1}{\sigma_1^2})^\top \bm{x}_i + \frac{\bm{\mu}_1^\top \bm{\mu}_1}{2\sigma_1^2} - \frac{\bm{\mu}_0^\top \bm{\mu}_0}{2\sigma_0^2} +\ln{\left(\frac{ \sigma_1^{F_h} }{ \sigma_0^{F_h} } \right)} \right)} }
\end{align*}
For the more general case where $\mathbb{P}(\bm{x}_i\in {\cal C}_0) = \frac{n_0}{n_0+n_1},\ \mathbb{P}(\bm{x}_i\in {\cal C}_1) = \frac{n_1}{n_0 + n_1}$, the results for $a, \bm{b}$ are the same and $c=\frac{\bm{\mu}_1^\top\bm{\mu}_1}{2\sigma_1^2} - \frac{\bm{\mu}_0^\top\bm{\mu}_0}{2\sigma_0^2} + \ln{\left(\frac{n_0 \sigma_1^{F_h} }{n_1 \sigma_0^{F_h} } \right)}$.
\end{proof}

\section{Generalized Jeffreys Divergence}
\label{appendix:NGJD}
Suppose we have
\begin{align*}
    P(\bm{x}) = N(\bm{\mu}_0,\sigma_0^2 I),\ Q(\bm{x}) = N(\bm{\mu}_1,\sigma_1^2 I)
\end{align*}
Then, the KL-divergence between $P(\bm{x})$ and $Q(\bm{x})$ is
\begin{align*}
    & D_\text{KL}(P||Q) = \int P(\bm{x}) \ln{\frac{P(\bm{x})}{Q(\bm{x})}} d\bm{x} = \mathbb{E}_{\bm{x}\sim P(\bm{x})} \ln{\frac{P(\bm{x})}{Q(\bm{x})}} \\
    &= \mathbb{E}_{\bm{x}\sim P(\bm{x})} \ln{\left( \frac{\sigma_1^{F_h}}{\sigma_0^{F_h}} \exp \left(-\frac{1}{2}{\sigma}_0^{-2}(\bm{x}-\bm{\mu}_0)^{\top} (\bm{x}-\bm{\mu}_0) + \frac{1}{2}{\sigma}_1^{-2}(\bm{x}-\bm{\mu}_1)^{\top} (\bm{x}-\bm{\mu}_1)  \right) \right)} \\
    &= F_h \ln{\frac{\sigma_1}{\sigma_0}} + \mathbb{E}_{\bm{x}\sim P(\bm{x})} \left(-\frac{1}{2}{\sigma}_0^{-2}(\bm{x}-\bm{\mu}_0)^{\top} (\bm{x}-\bm{\mu}_0) + \frac{1}{2}{\sigma}_1^{-2}(\bm{x}-\bm{\mu}_1)^{\top} (\bm{x}-\bm{\mu}_1)  \right) \\
    &= F_h \ln{\frac{\sigma_1}{\sigma_0}}- \frac{F_h}{2} + \mathbb{E}_{\bm{x}\sim P(\bm{x})} \left(\frac{1}{2}{\sigma}_1^{-2}(\bm{x}-\bm{\mu}_1)^{\top} (\bm{x}-\bm{\mu}_1)  \right) \\
    &= F_h \ln{\frac{\sigma_1}{\sigma_0}}- \frac{F_h}{2} + F_h\frac{\sigma_0^{2}}{2\sigma_1^{2}} + \frac{(\bm{\mu}_0- \bm{\mu}_1)^\top (\bm{\mu}_0- \bm{\mu}_1)}{2\sigma_1^{2}} \\ 
    & = F_h \ln{\frac{\sigma_1}{\sigma_0}}- \frac{F_h}{2} + F_h\frac{\sigma_0^{2}}{2\sigma_1^{2}} + \frac{d_X^2}{2\sigma_1^{2}}
\end{align*}
where $d_X^2$ is the squared Euclidean distance. In the same way, we have
\begin{align*}
    D_\text{KL}(Q||P) = F_h \ln{\frac{\sigma_0}{\sigma_1}}- \frac{F_h}{2} + F_h\frac{\sigma_1^{2}}{2\sigma_0^{2}} + \frac{d_X^2}{2\sigma_0^{2}}\\
\end{align*}
Suppose $\mathbb{P}(\bm{x} \sim P)=\mathbb{P}(\bm{x} \sim Q)=\frac{1}{2}$, then we have
\begin{align*}
    D_\text{NGJ}(\text{CSBM-H}) & =  -\mathbb{P}(\bm{x} \sim P) \mathbb{E}_{\bm{x}\sim P}\left[\ln{\frac{P(\bm{x})}{Q(\bm{x})}} \right] - \mathbb{P}(\bm{x} \sim Q) \mathbb{E}_{\bm{x}\sim Q}\left[\ln{\frac{Q(\bm{x})}{P(\bm{x})}} \right] \\ 
    & = -\frac{F_h}{4}(\rho^2 + \frac{1}{\rho^2}-2) - d_X^2(\frac{1}{4\sigma_1^2}+\frac{1}{4\sigma_0^2})
\end{align*}

\section{Calculation of Probabilistic Bayes Error (PBE)}
\label{appendix:noncentral_chisquare}
\subsection{An Introduction}
\paragraph{Noncentral $\chi^2$ distribution}
Let $\left(X_1, X_2, \dots, X_i, \ldots, X_k\right)$ be $k$ independent, normally distributed random variables with means $\mu_i$ and unit variances. Then the random variable
$$
\sum_{i=1}^k X_i^2 \sim \chi'^2(k,\lambda)
$$
is distributed according to the noncentral $\chi^2$ distribution. It has two parameters $(k,\lambda)$: $k$ which specifies the number of degrees of freedom (\ie{} the number of $X_i$ ), and $\lambda$ which is the sum of the squared mean of the random variables $X_i$:
$$
\lambda=\sum_{i=1}^k \mu_i^2 .
$$
$\lambda$ is sometimes called the noncentrality parameter.
\paragraph{Generalized $\chi^2$ distribution}
The generalized $\chi^2$ variable can be written as a linear sum of independent noncentral $\chi^2$ variables and a normal variable: 
$$
\xi=\sum_i w_i Y_i+X, \quad Y_i \sim \chi^{\prime 2}\left(k_i, \lambda_i\right), \quad X \sim N\left(m, s^2\right)
$$
Here the parameters are the weights $w_i$, the degrees of freedom $k_i$ and non-centralities $\lambda_i$ of the constituent $\chi^2$, and the normal parameters $m$ and $s$.

\subsection{Quadratic Function}

For $i \in \mathcal{C}_0$, we rewrite $\bm{x}_i = \sigma_0 \bm{y}_i + \bm{\mu}_0$
where $\bm{y}_i$  is the standard normal variable. 
The quadratic function of $Q(\bm{x}_i)$ satisfies
\begin{align*}
    Q(\bm{x}_i) 
    & = a \bm{x}_i^\top \bm{x}_i + \bm{b}^\top \bm{x}_i +c \\
    & = a\Big(\bm{x}_i+\frac{\bm{b}}{2a}\Big)^\top\Big(\bm{x}_i+\frac{\bm{b}}{2a}\Big) 
          + c-\frac{\bm{b}^\top\bm{b}}{4a} \\ 
    & = a\Big(\sigma_0 \bm{y}_i + \bm{\mu}_0+\frac{\bm{b}}{2a}\Big)^\top\Big(\sigma_0 \bm{y}_i + \bm{\mu}_0+\frac{\bm{b}}{2a}\Big) + c-\frac{\bm{b}^\top\bm{b}}{4a} \\
    &= a\sigma_0^2\Big(\bm{y}_i + \frac{\bm{\mu}_0}{\sigma_0}+\frac{ \bm{b}}{2a\sigma_0}\Big)^\top\Big(\bm{y}_i + \frac{\bm{\mu}_0}{\sigma_0} +\frac{\bm{b}}{2a\sigma_0}\Big) + c-\frac{\bm{b}^\top\bm{b}}{4a} \\ 
    & =  w_0 {y}_{i}' + c-\frac{\bm{b}^\top\bm{b}}{4a} \sim \tilde{\chi}^2(w_0,F_h,\lambda_0) + c-\frac{\bm{b}^\top\bm{b}}{4a}
\end{align*}
where ${y}_{i}' = (\bm{y}_i + \frac{\bm{\mu}_0}{\sigma_0}+\frac{ \bm{b}}{2a\sigma_0})^\top (\bm{y}_i + \frac{\bm{\mu}_0}{\sigma_0}+\frac{ \bm{b}}{2a\sigma_0}) \sim \chi'^2(F_h,\lambda_0)$ is distributed as non-central $\chi^2$ distribution, the degree of freedom is $F_h$, $\lambda_0 = (\frac{\bm{\mu}_0}{\sigma_0} + \frac{ \bm{b}}{2a\sigma_0})^\top(\frac{\bm{\mu}_0}{\sigma_0}+\frac{ \bm{b}}{2a\sigma_0})$, the weight $w_0 = a\sigma_0^2$. Then, the cumulative distribution function (CDF) of $Q(\bm{x}_i)$ can be calculated as follows,
\begin{align*}
    \text{CDF}(x) &= \mathbb{P}( Q(\bm{x}_i) \leq x) = \mathbb{P}( \tilde{\chi}^2(w_0,F_h,\lambda_0) \leq  x- c+\frac{\bm{b}^\top\bm{b}}{4a}) = \text{CDF}_{\tilde{\chi}^2(w_0,F_h,\lambda_0)}(x - \xi)  
\end{align*}
where $\xi = c - \frac{\bm{b}^\top\bm{b}}{4a}$. For $j \in \mathcal{C}_1$ and $\bm{h}_{i}, \bm{h}_{j}$, we can apply the same computation. And since
\begin{align*}
& \mathbb{P}(\text{CL}_\text{Bayes}(\bm{x})=0|\bm{x}\in \mathcal{C}_0) = \mathbb{P}(Q(\bm{x})>0|\bm{x}\in \mathcal{C}_0)  = 1-\text{CDF}_{\tilde{\chi}^2(w_0,F_h,\lambda_0)}(-\xi), \\ %
& \mathbb{P}(\text{CL}_\text{Bayes}(\bm{x})=1|\bm{x}\in \mathcal{C}_1) = \mathbb{P}(Q(\bm{x})\leq 0|\bm{x}\in \mathcal{C}_1) = \text{CDF}_{\tilde{\chi}^2(w_1,F_h,\lambda_1)}(-\xi). %
\end{align*}
where $w_1 = a\sigma_1^2, \lambda_1 = (\frac{\bm{\mu}_1}{\sigma_1} + \frac{ \bm{b}}{2a\sigma_1})^\top(\frac{\bm{\mu}_1}{\sigma_1}+\frac{ \bm{b}}{2a\sigma_1})$. Then from  \eqref{eq:bayes_error_csbmh} and with $\mathbb{P}(\bm{x} \sim P) = \mathbb{P}(\bm{x} \sim Q) = 1/2$, the PBE for the two normal settings can be calculated as,
$$
    \frac{\text{CDF}_{\tilde{\chi}^2(w_0,F_h,\lambda_0)}(-\xi) + \left(1- \text{CDF}_{\tilde{\chi}^2(w_1,F_h,\lambda_1)}(-\xi) \right)}{2}
$$

For a special case where $\sigma_0^2=\sigma_1^2 \neq 0$, we have $a=0$ and $Q(\bm{x}_i)=\bm{b}^\top \bm{x}_i +c$ follows a normal distribution:
\begin{align*}
    Q(\bm{x}_i)\sim N(\bm{b}^\top \bm{\mu}_0+c, \sigma_0^2 \bm{b}^\top\bm{b} )
\end{align*}
Then
\begin{align*}
    \text{CDF}(x) 
    & = \mathbb{P}( Q(\bm{x}_i) \leq x) = \mathbb{P}( \bm{b}^\top \bm{\mu}_0+c + \sqrt{\sigma_0^2 \bm{b}^\top\bm{b}} x' \leq x) \\ 
    & = \mathbb{P}( x' \leq \frac{x-(\bm{b}^\top \bm{\mu}_0+c )}{\sqrt{\sigma_0^2 \bm{b}^\top\bm{b}}} )
\end{align*}
where $x' \sim N(0,1)$ and the PBE becomes
$$
    \frac{\text{CDF}_{N(0,1)} \left(\frac{-(\bm{b}^\top \bm{\mu}_0+c )}{\sqrt{\sigma_0^2 \bm{b}^\top\bm{b}}} \right) + \left(1- \text{CDF}_{N(0,1)} \left(\frac{-(\bm{b}^\top \bm{\mu}_1+c )}{\sqrt{\sigma_1^2 \bm{b}^\top\bm{b}}} \right) \right)}{2}
$$

\section{A Discussion of (Imbalanced) Prior Distribution}
\label{appendix:imbalanced_prior_distribution}
In this section, we provide an open-ended discussion on how the prior distribution (\ie{} imbalanced datasets) will influence the ND of CSBM-H, which can possibly lead to some interesting future works.

Let ${\mathbb{P}}(\bm{x} \sim \mathcal{C}_0) = \frac{n_0}{n_1+n_0}=p_0,\; {\mathbb{P}}(\bm{x} \sim \mathcal{C}_1) = \frac{n_1}{n_1+n_0}=p_1, p_0+p_1=1$ and $\rho = \frac{\sigma_0}{\sigma_1}$, which is the ratio of standard deviation and $0 \leq \rho \leq 1$, then $D_\text{NGJ}$ is
\begin{align*}
    D_\text{NGJ}(\text{CSBM-H})
    & = - p_0 D_\text{KL}(P||Q) - p_1 D_\text{KL}(Q||P) \\  
    &= F_h\ln{\rho}(p_1-p_0 ) + \frac{F_h}{2}(p_0\rho^2 + \frac{p_1}{\rho^2}-1) + d_X^2(\frac{p_0}{2\sigma_1^2}+\frac{p_1}{2\sigma_0^2})
\end{align*}
where $d_X^2 = (\bm{\mu}_0- \bm{\mu}_1)^\top (\bm{\mu}_0- \bm{\mu}_1)$, which is the square Euclidean distance between the means of the two distributions.

And the PBE of CSBM-H becomes
$$
    \frac{n_0 \text{CDF}_{\tilde{\chi}^2(w_0,F_h,\lambda_0)}(-\xi) + n_1 \left(1- \text{CDF}_{\tilde{\chi}^2(w_1,F_h,\lambda_1)}(-\xi) \right)}{n_0 + n_1}
$$

\begin{figure}[htbp!]
    \centering
     {
     \subfloat[PBE]{
     \captionsetup{justification = centering}
     \includegraphics[width=0.43\textwidth]{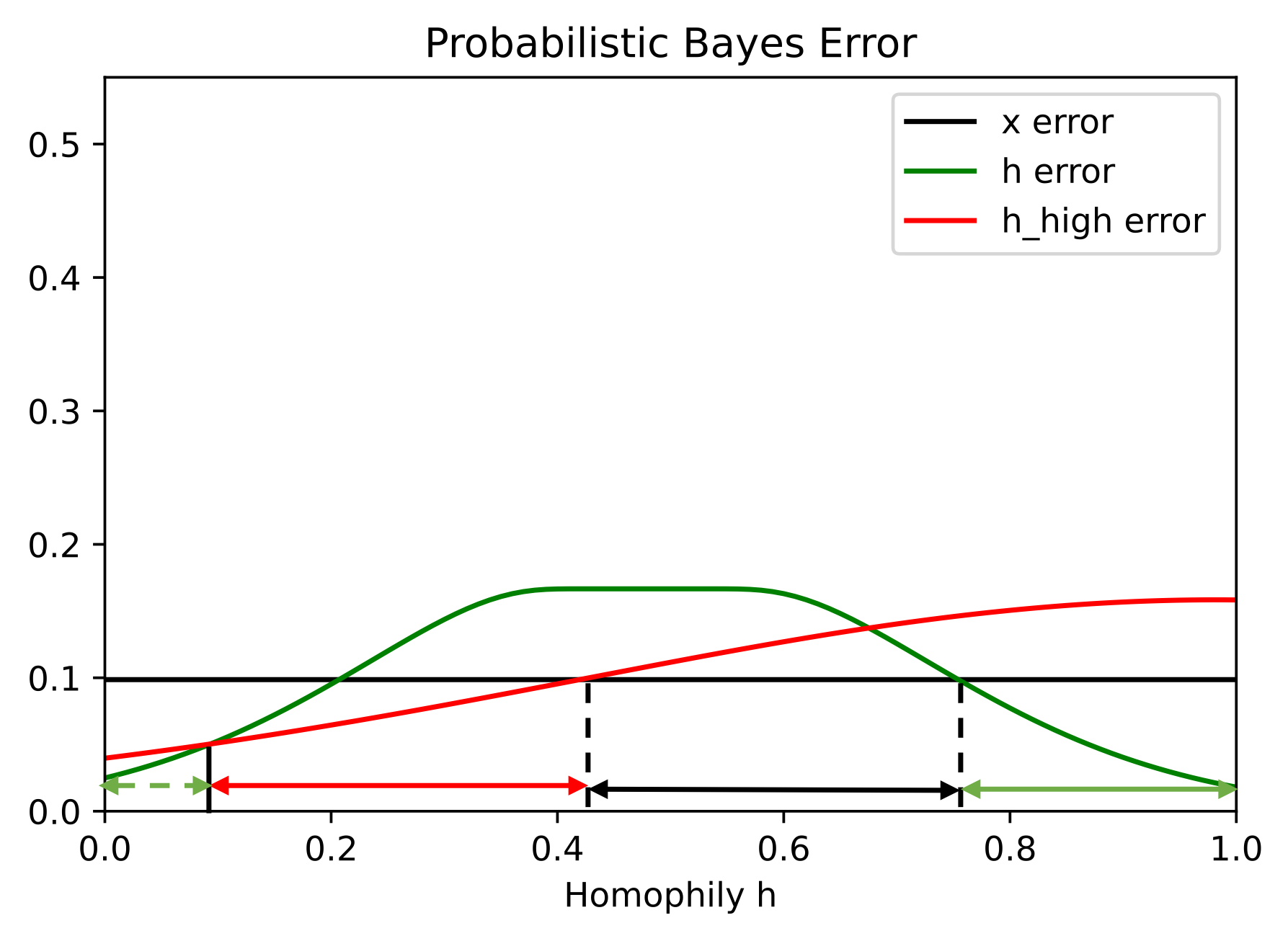}
     }
     \subfloat[$D_\text{NGJ}$]{
     \captionsetup{justification = centering}
     \includegraphics[width=0.5\textwidth]{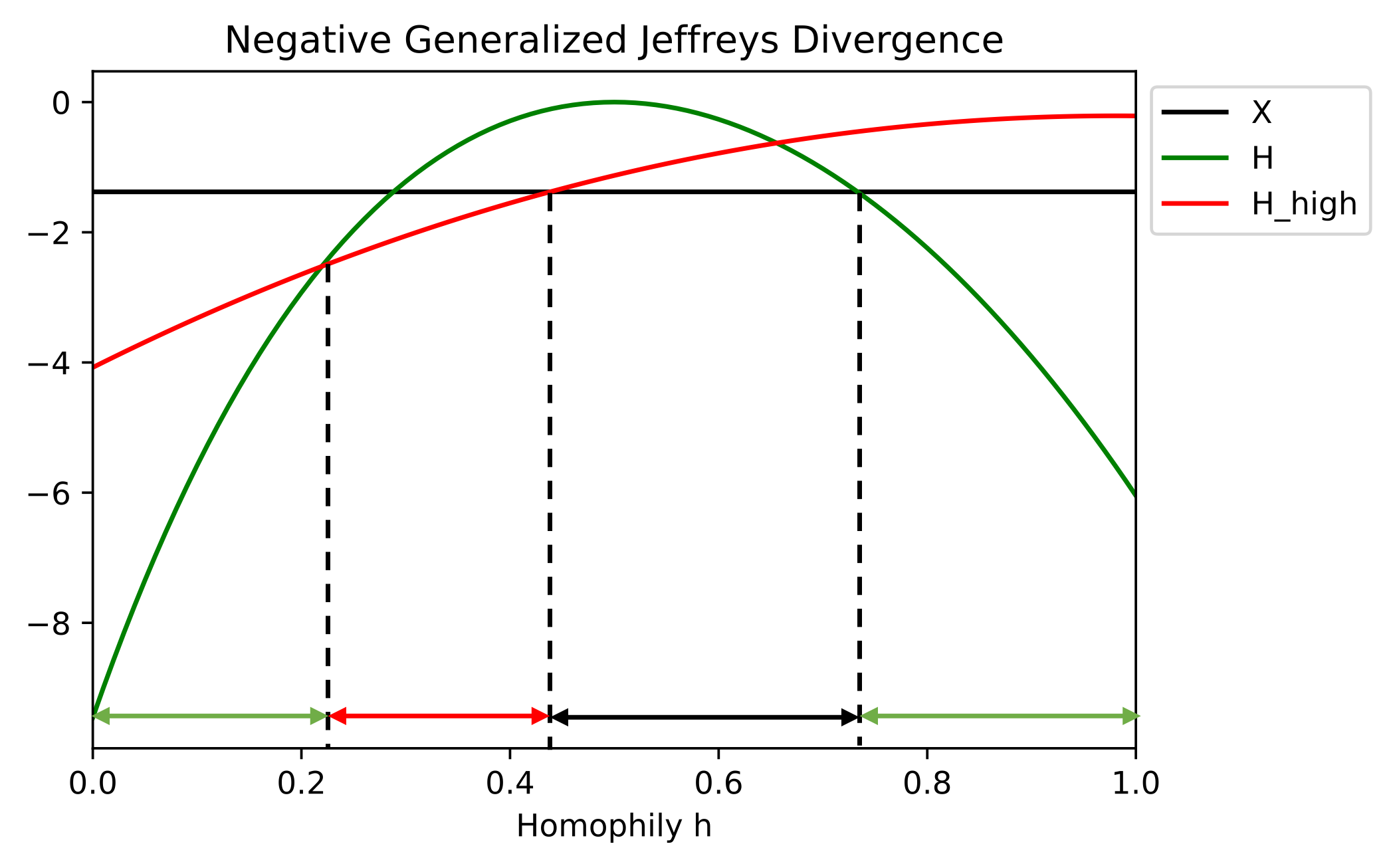}
     }\\
     \subfloat[ENND]{
     \captionsetup{justification = centering}
     \includegraphics[width=0.52\textwidth]{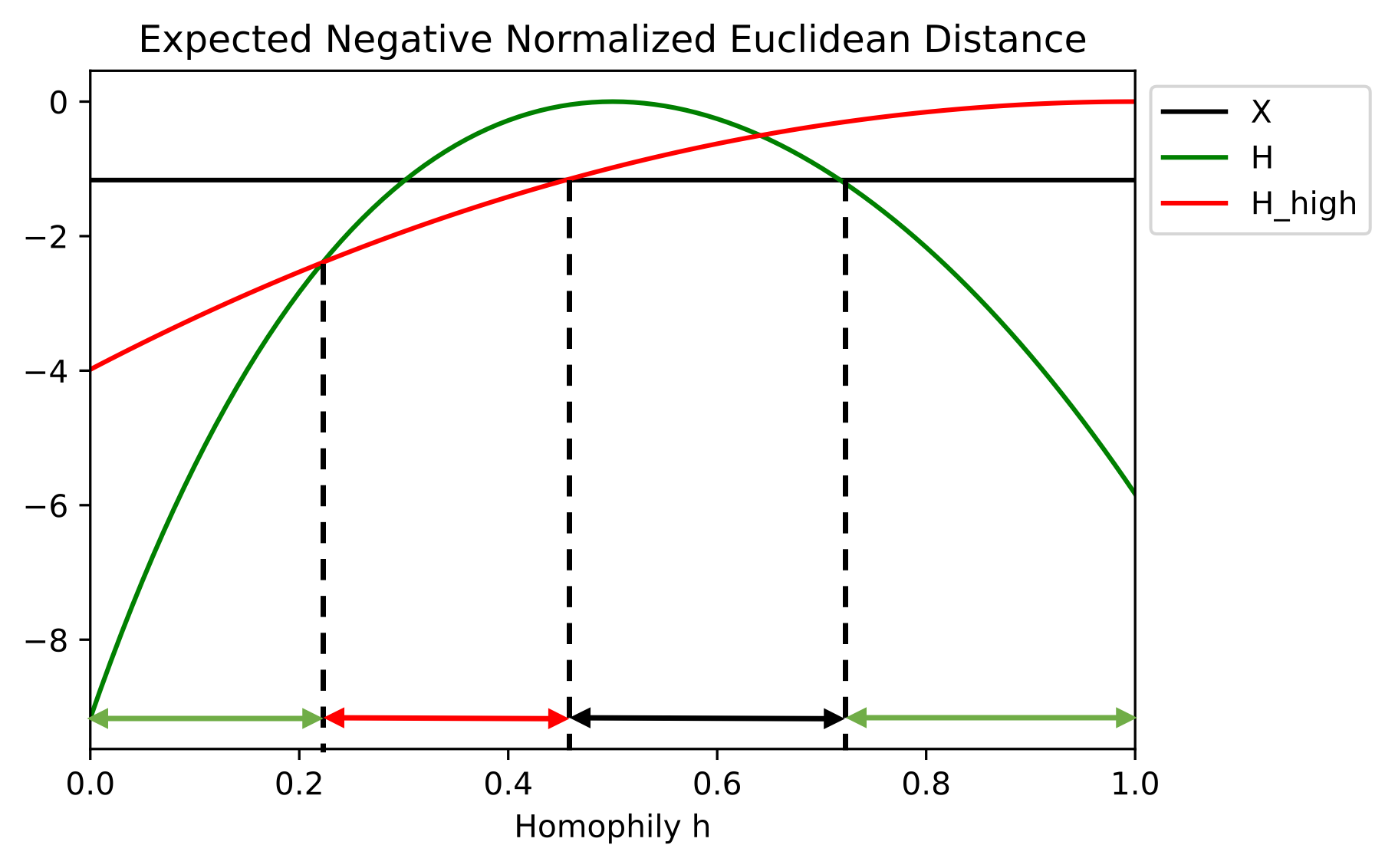}
     } \hspace*{-0.2cm}
     \subfloat[Negative Variance Ratio]{
     \captionsetup{justification = centering}
     \includegraphics[width=0.45\textwidth]{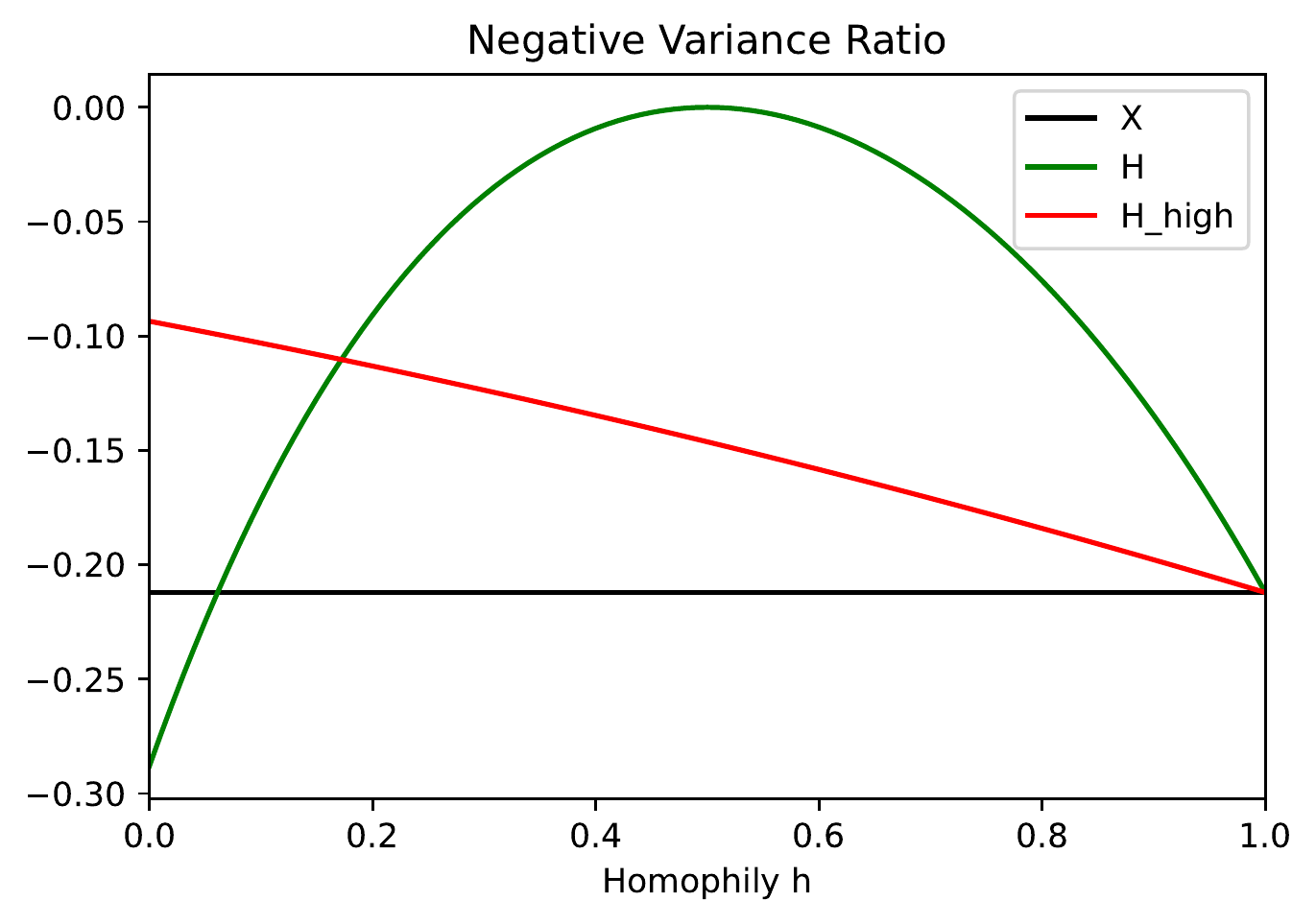}
     } 
     }
     \caption{Comparison of CSBM-H with $n_0=500,n_1=100$.}
     \label{fig:csbmh_n0=500_n1=100}
\end{figure}

\begin{figure}[htbp!]
    \centering
     {
     \subfloat[PBE]{
     \captionsetup{justification = centering}
     \includegraphics[width=0.43\textwidth]{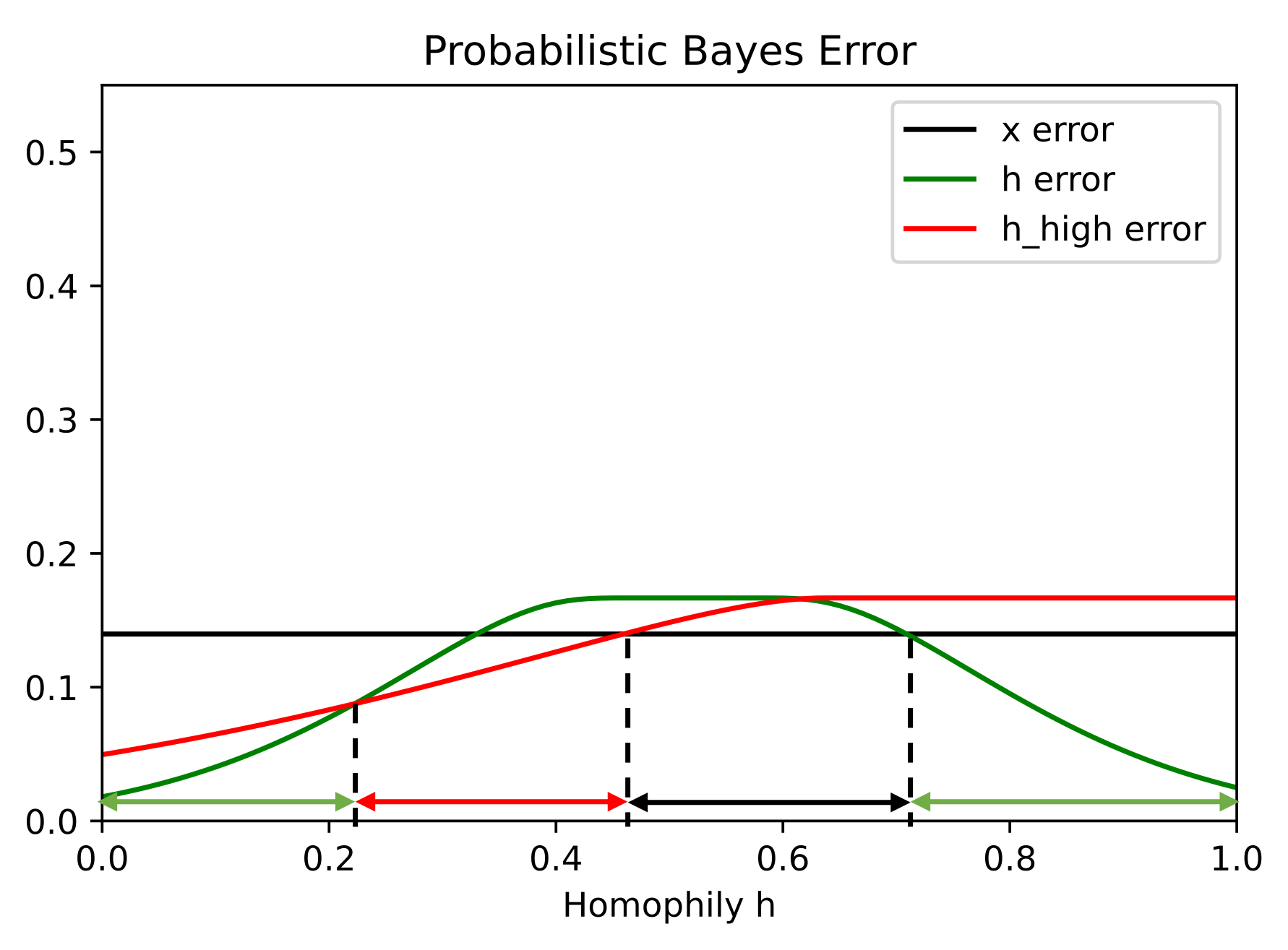}
     }
     \subfloat[$D_\text{NGJ}$]{
     \captionsetup{justification = centering}
     \includegraphics[width=0.5\textwidth]{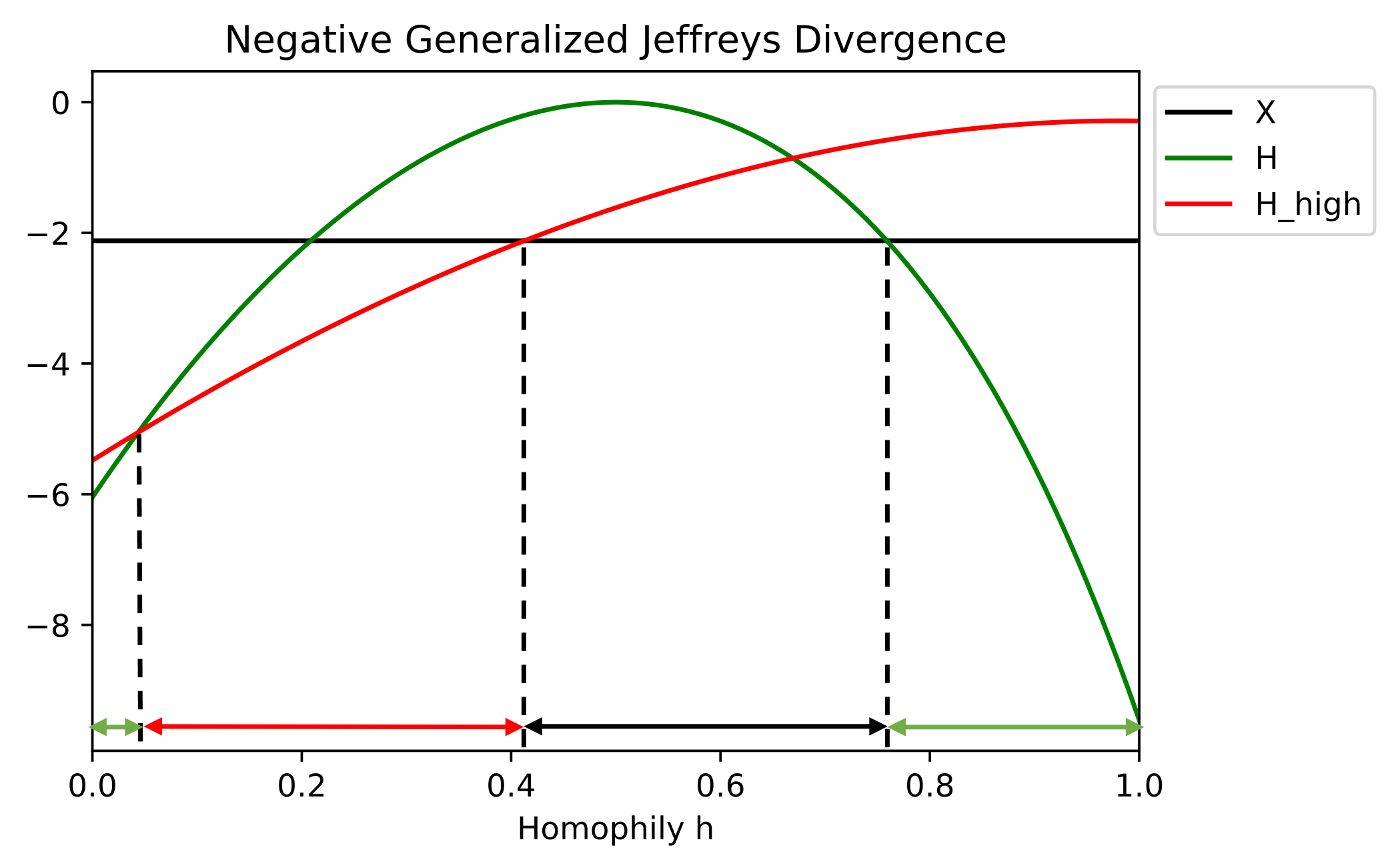}
     }\\
     \subfloat[ENND]{
     \captionsetup{justification = centering}
     \includegraphics[width=0.5\textwidth]{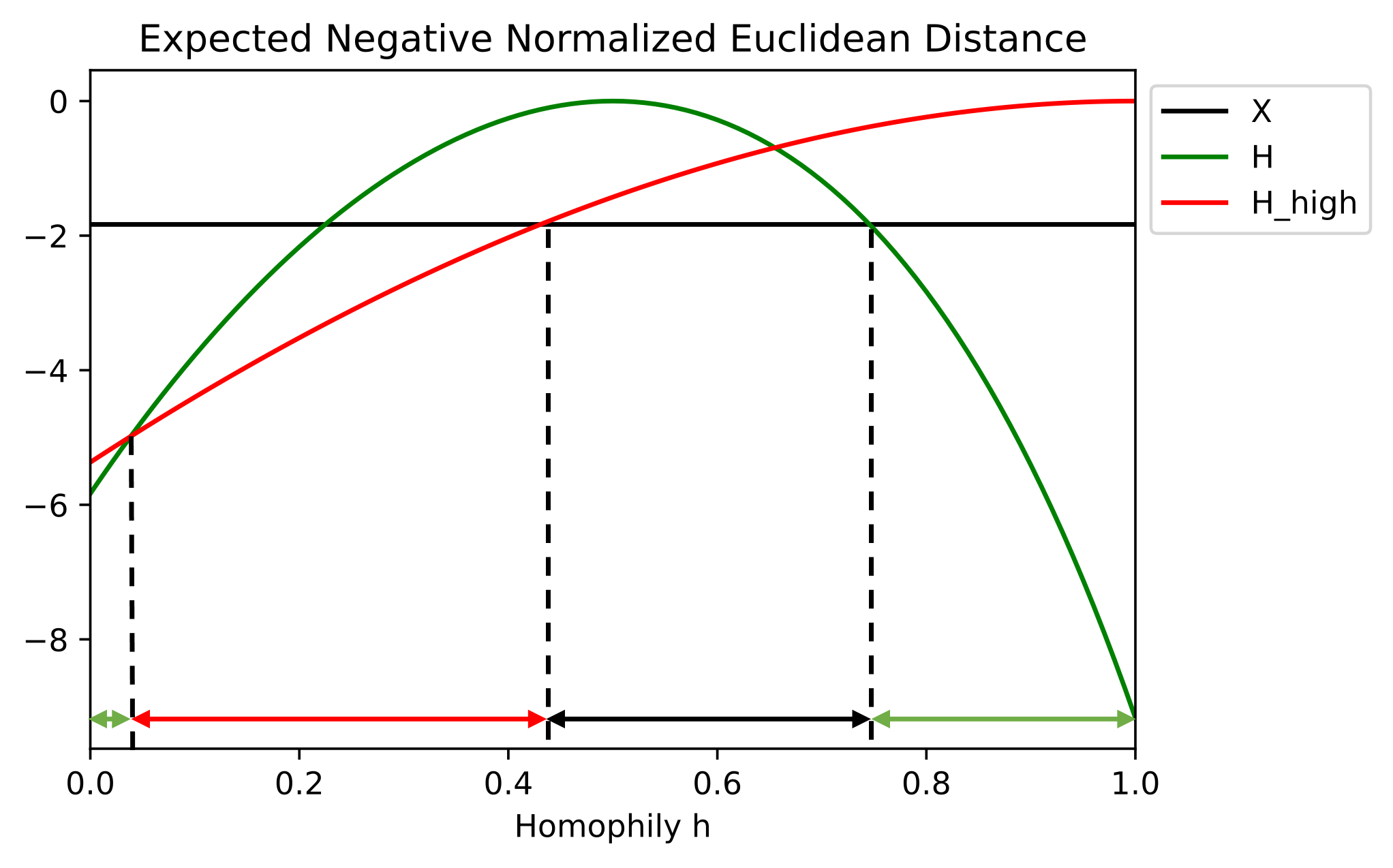}
     } \hspace*{-0.2cm}
     \subfloat[Negative Variance Ratio]{
     \captionsetup{justification = centering}
     \includegraphics[width=0.43\textwidth]{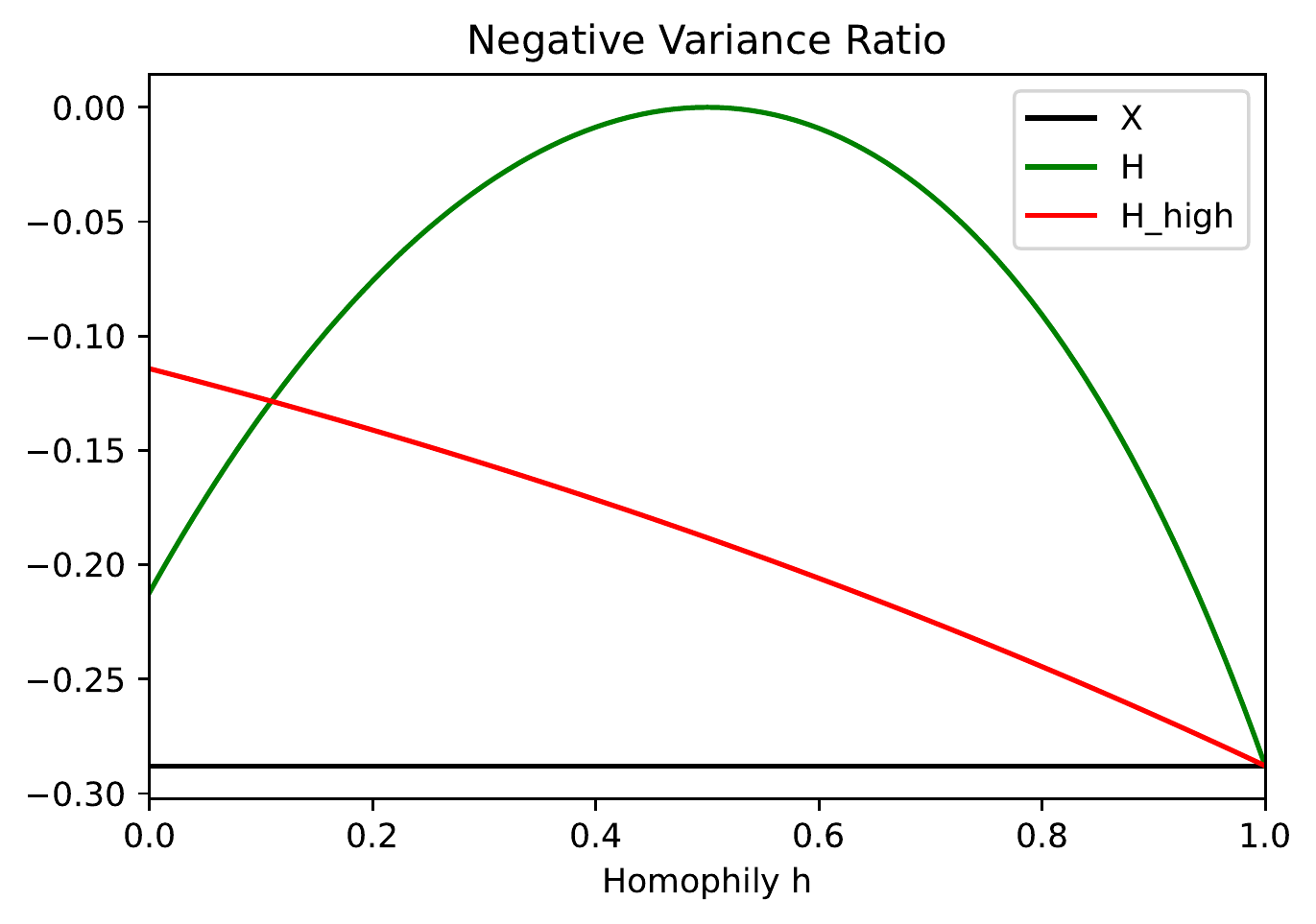}
     } 
     }
     
     \caption{Comparison of CSBM-H with $n_0=100,n_1=500$.}
     \label{fig:csbmh_n0=100_n1=500}
\end{figure}
From Figure \ref{fig:csbmh_n0=500_n1=100} (b) we can find that, as the size of the low-variation class increases, the LP regime expands and HP regime shrinks at the low homophily area in terms of $D_\text{NGJ}$. This is because in ENND, the normalization term $\frac{p_0}{2\sigma_1^2}$ gets higher weight, making the curve for LP filters move down and HP filter move up, which leads to the expansion of LP regime.

However, we can also observe that the changes of PBE and $D_\text{NGJ}$ curves show inconsistent results. As the size of the low-variation class increases, the LP regime shrinks and HP regime expands in PBE, while the LP regime expands and HP regime shrinks in $D_\text{NGJ}$. In Figure \ref{fig:csbmh_n0=100_n1=500}, we observe the similar inconsistency between PBE and $D_\text{NGJ}$ curves. This discrepancy reminds us that the performance of LP and HP filters on imbalanced datasets might be under-explored. We do not have a conclusion for this challenge in this paper and we encourage more researchers to study the connection among the prior distribution, the performance of LP and HP filters and ND.

\section{More About CSBM-H}
\label{appendix:more_about_csbmh}
\subsection{Directed or Undirected Graphs?}
\label{appendix:directed_or_undirected}
\paragraph{Question} Why not generated an undirected graph.
\paragraph{Answer} If we impose undirected assumption in CSBM-H, we have to not only discuss the node degree from intra-class edges, but also discuss degree from inter-class edges and control their relations with the corresponding homophily level. This will inevitably add more parameters to CSBM-H and make the model much more complicated. However, we find that this complication does not bring us extra benefit for understanding the effect of homophily, which deviate the main goal of our paper. And we guess this might be one of the reasons that the existing work mainly keep the discussion within the directed setting \cite{ma2021homophily}. 

Actually, when CSMB-H was firstly designed, we would like to only have one "free parameter" $h$ in it to make it simple. Because in this way, we are able to show the whole picture of the effect of homophily from 0 to 1, like the figures in Section \ref{sec:obervations_on_csbmh}.

\subsection{Extend CSBM-H to More General Settings}
\label{appendix:extend_csbmh_general_settings}
The CSBM-H can be extended to more general settings: 1. The two-normal setting can be expanded to multi-class classification problems with different sets of $(\bm{\mu},\sigma^2 \bm{I}, d)$ parameters for each class; 2. The degrees of nodes can be generalized to different degree distributions; 3. The scalar homophily parameter $h$ can be generalized to matrix $H \in \mathbb{R}^{C \times C}$, where $H_{c_1,c_2}$ represents the probability of nodes in class $\mathcal{C}_1$ connecting to nodes in class $\mathcal{C}_2$, which is the compatibility matrix used in \cite{zhu2021graph}. Furthermore, we can also define the local homophily value $H_{v,c}$ for each node, where $H_{v,c}$ indicates the proportion of neighbors of $v$ that connect to nodes from class $\mathcal{C}_c$. To demonstrate and visualize the effect of homophily intuitively and easily, we use the CSBM-H settings in this paper as stated above. Although the settings are simple, insightful results can still be obtained. The more complicated variants will be left for future work.

\subsection{Two More Metrics of ND: Negative Squared Wasserstein Distance and Hellinger Distance}
\label{appendix:two_more_nd_metrics}
\paragraph{Original Definition}
In general, the Wasserstein distance between two Gaussians is $d=W_2\left({N}\left(\bm{\mu}_0, \Sigma_0 \right) ; {N}\left(\bm{\mu_1}, \Sigma_1\right)\right)$ and we have \cite{givens1984class, knott1984optimal, olkin1982distance, dowson1982frechet}
$$
d^2=\left\|\bm{\mu}_0 - \bm{\mu}_1 \right\|_2^2 + \text{Tr}\left( \Sigma_0 + \Sigma_1 - 2\left(\Sigma_0^{1/2} \Sigma_1 \Sigma_0^{1/2}\right)^{1/2}\right)
$$

The squared Hellinger distance between two Gaussians is \cite{pardo2018statistical}
$$
\resizebox{1\hsize}{!}{$H^2({N}\left(\bm{\mu}_0, \Sigma_0 \right) ; {N}\left(\bm{\mu_1}, \Sigma_1\right)) = 1 -\frac{\text{det}\left(\Sigma_0\right)^{1/4} \text{det}\left(\Sigma_1\right)^{1/4}}{\text{det}\left(\frac{\Sigma_0+\Sigma_1}{2}\right)^{1/2}} \exp \left\{-\frac{1}{8}\left(\bm{\mu}_0 - \bm{\mu}_1 \right)^\top \left(\frac{\Sigma_0 + \Sigma_1}{2}\right)^{-1} \left(\bm{\mu}_0 - \bm{\mu}_1 \right)\right\} $}
$$
Wasserstein distance is a distance function defined between probability distributions and Hellinger distance is a type of $f$-divergence which is used to quantify the similarity between two probability distributions \cite{tanton2005encyclopedia,jeffreys1946invariant}. These two metrics can be used to study the ND of CSBM-H and we will introduce them in the following subsection.

\paragraph{Calculation for CSBM-H}
The negative squared Wasserstein distance (NSWD) for CSBM-H is

\begin{align*}
\text{NSWD} = - \left\|\bm{\mu}_0 - \bm{\mu}_1 \right\|_2^2 - F_h (\sigma_0 - \sigma_1)^2 
\end{align*}

The negative squared Hellinger distance (NSHD) for CSBM-H is
\begin{align*}
\text{NSHD} &= - 1 + \frac{\text{det}\left(\sigma_0^2 I \right)^{1/4} \text{det}\left(\sigma_1^2 I \right)^{1/4}}{\text{det}\left(\frac{\sigma_0^2 I + \sigma_1^2 I}{2}\right)^{1/2}} \exp \left\{-\frac{1}{8}\left(\bm{\mu}_0 - \bm{\mu}_1 \right)^\top \left(\frac{\sigma_0^2 + \sigma_1^2 }{2}\right)^{-1} \left(\bm{\mu}_0 - \bm{\mu}_1 \right)\right\} \\
& = - 1 + \frac{\sigma_0^{F_h/2} \sigma_1^{F_h/2}}{(\frac{\sigma_0^2  + \sigma_1^2 }{2})^{F_h/2} } \exp \left\{-\frac{1}{8}\left(\bm{\mu}_0 - \bm{\mu}_1 \right)^\top \left(\frac{\sigma_0^2 + \sigma_1^2 }{2}\right)^{-1} \left(\bm{\mu}_0 - \bm{\mu}_1 \right)\right\}\\
& = - 1 + \left(\frac{2}{\rho^2 + 1/\rho^2 } \right)^{F_h/2} \exp \left\{-\frac{d_X^2}{4 \left(\sigma_0^2 + \sigma_1^2 \right) }  \right\}
\end{align*}

Although defining the distance in different ways, both NSWD and NSHD indicate that the ND of CSBM-H depends on both intra- and inter-class ND, which is consistent with our conclusions in main paper. Besides, NSWD and NSHD provide analytic expressions for ND, which can be good tools for future research.

\section{More Figures of CSBM-H}
\label{appendix:more_figures_csbmh}

\begin{figure}[htbp!]
    \centering
     {
     \subfloat[PBE]{
     \captionsetup{justification = centering}
     \includegraphics[width=0.42\textwidth]{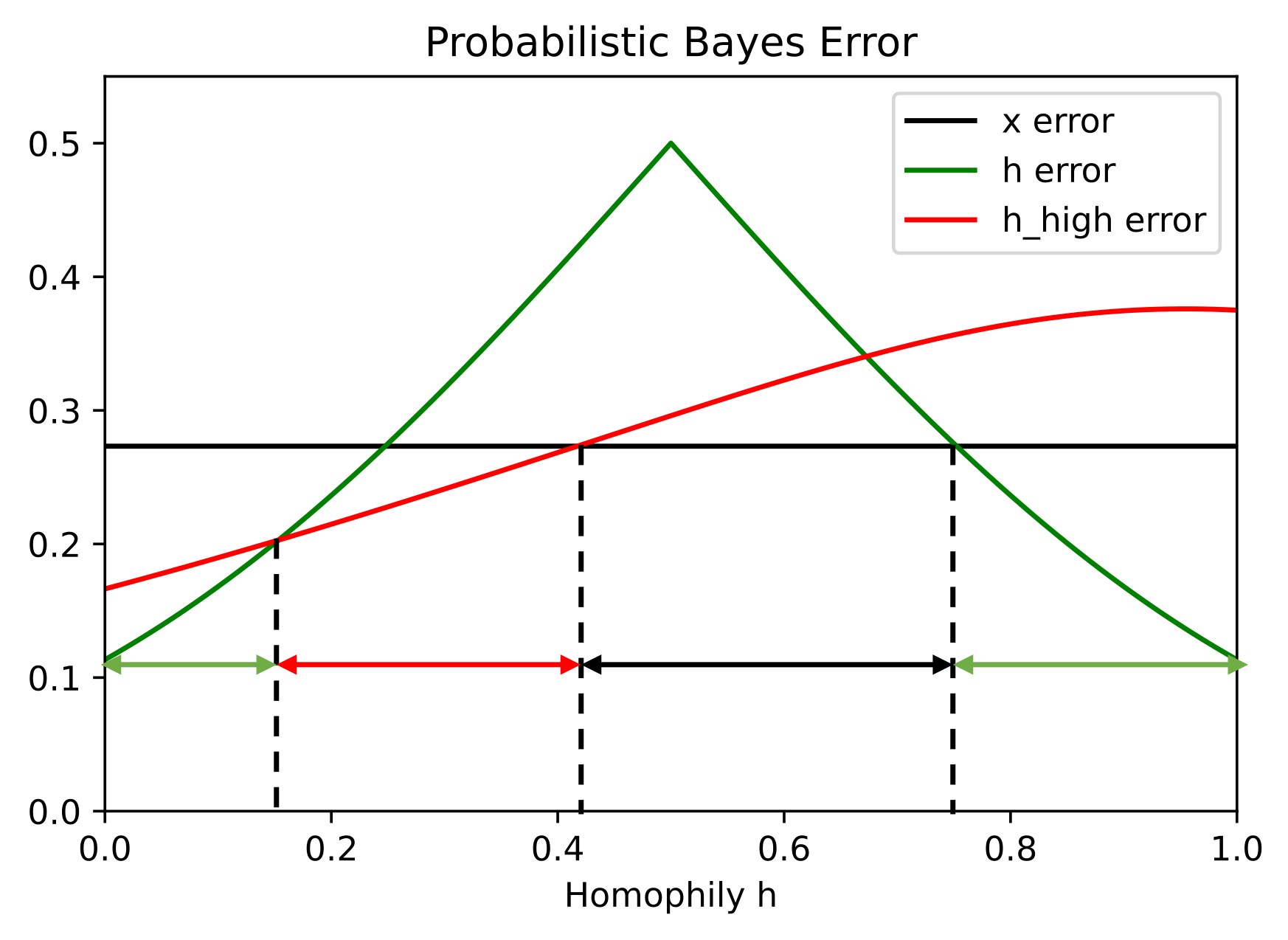}
     }
     \subfloat[$D_\text{NGJ}$]{
     \captionsetup{justification = centering}
     \includegraphics[width=0.5\textwidth]{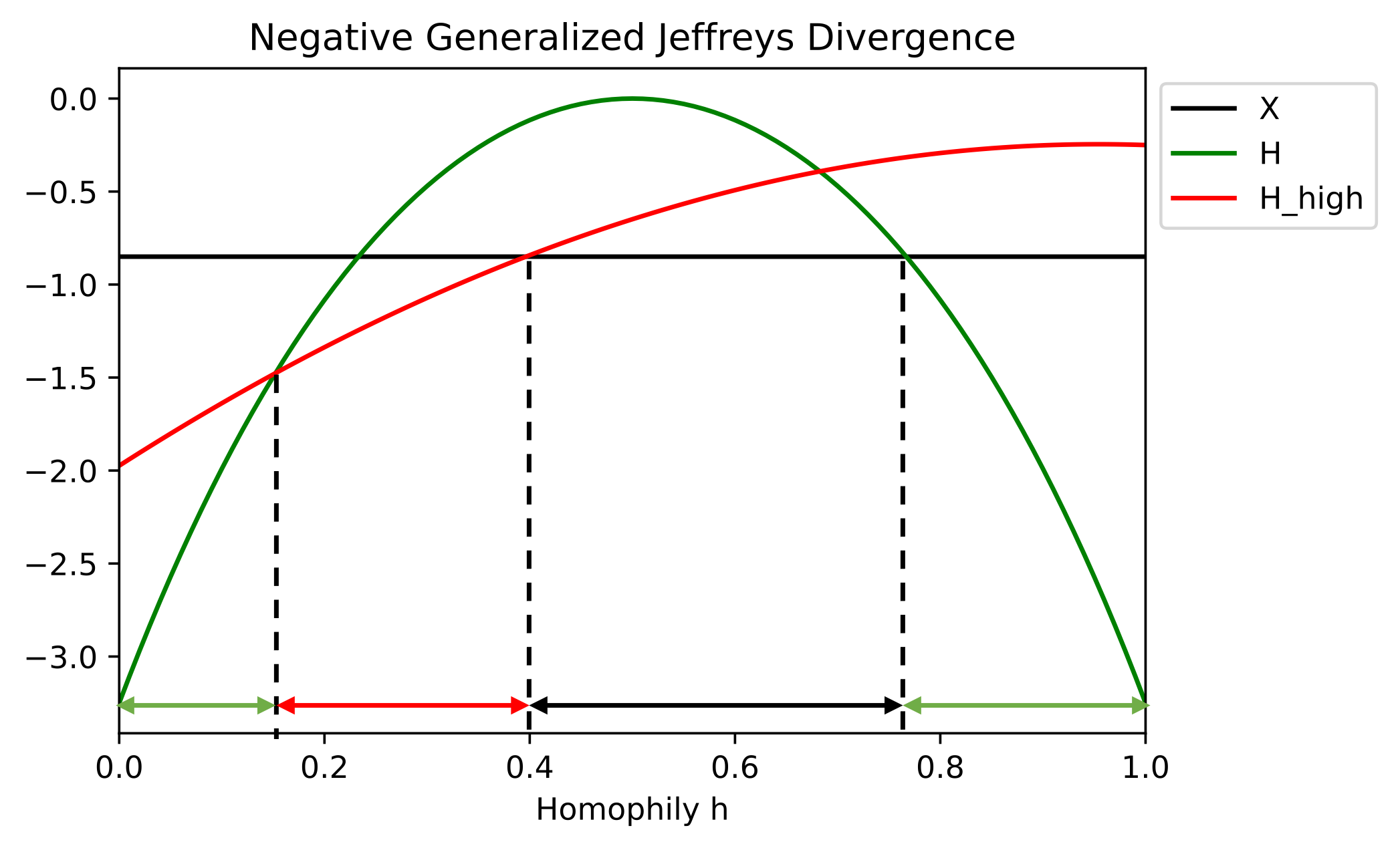}
     }\\ 
     \subfloat[ENND]{
     \captionsetup{justification = centering}
     \includegraphics[width=0.5\textwidth]{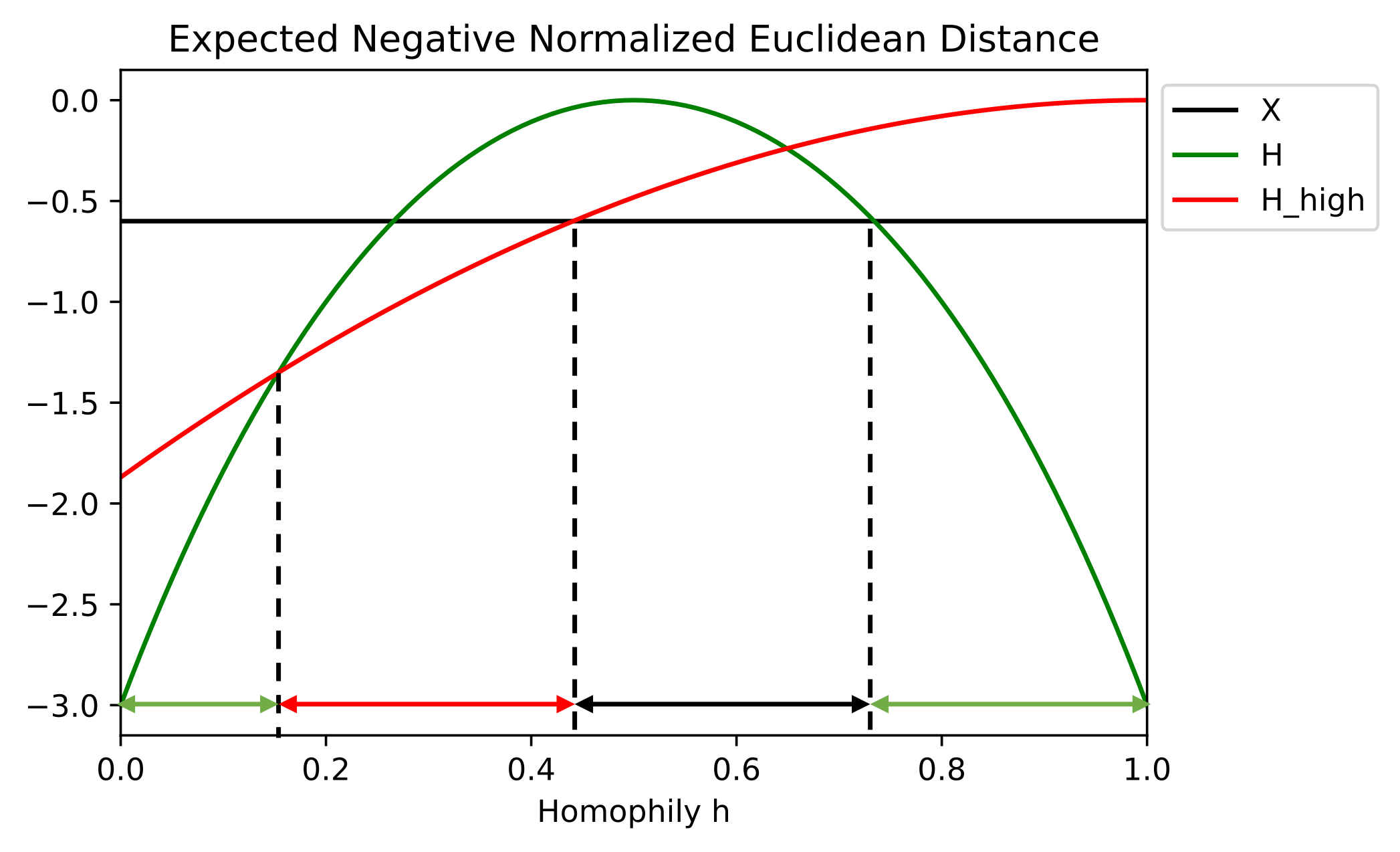}
     } \hspace*{-0.2cm}
     \subfloat[Negative Variance Ratio]{
     \captionsetup{justification = centering}
     \includegraphics[width=0.43\textwidth]{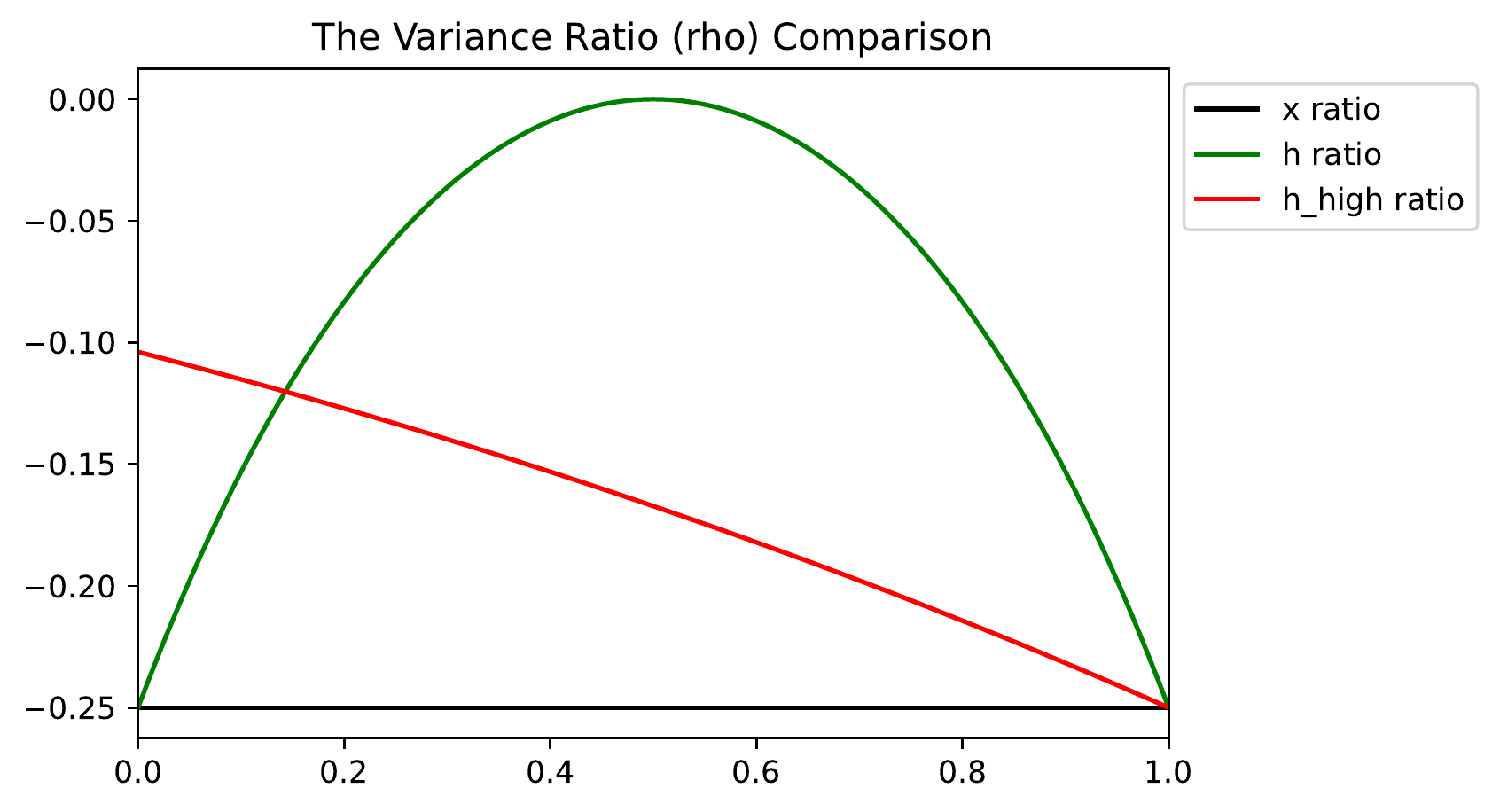}
     } 
     }
     
     \caption{Comparison of CSBM-H with $\sigma^2_0=2.5,\sigma^2_1=5$.}
     \label{fig:csbmh_sigma0=2.5_sigma1=5}
\end{figure}

\section{Proof of Theorem 2}
\label{appendix:proof_theorem2}
To prove theorem 2, we need the following two lemmas.
\begin{lemma} 1
Let $\bm{x}_i = X_{i,:}$ and suppose each dimension of $\bm{x}_i$ are independent, then for $\bm{x}_i, \bm{h}_i = (H_{i,:})^\top, \bm{h}_i^\text{HP} = (H_{i,:}^\text{HP})^\top$ we have
\begin{align*}
    \mathbb{P}\left(\left\|\mathbf{x}_{i}-\mathbf{x}_{j}\right\|_{2} \geq t\right) \leq \sum_{k=1}^{F_h} \mathbb{P}\left(\left| \mathbf{x}_{i,k}- \mathbf{x}_{j,k} \right| \geq \frac{t}{ \sqrt{F_h}} \right)\\
    \mathbb{P}\left(\left\|\mathbf{h}_{i}-\mathbf{h}_{j}\right\|_{2} \geq t\right) \leq \sum_{k=1}^{F_h} \mathbb{P}\left(\left|  \mathbf{h}_{i,k}- \mathbf{h}_{j,k} \right| \geq \frac{t}{ \sqrt{F_h}} \right)\\
    \mathbb{P}\left(\left\|\bm{h}_i^\text{HP}-\bm{h}_j^\text{HP}\right\|_{2} \geq t\right) \leq \sum_{k=1}^{F_h} \mathbb{P}\left(\left|  \mathbf{h}^\text{HP}_{i,k}- \mathbf{h}^\text{HP}_{j,k} \right| \geq \frac{t}{ \sqrt{F_h}} \right)
\end{align*}
\end{lemma}

\begin{proof}
If $\left\| \mathbf{x}_{i}-\mathbf{x}_{j} \right\|_{2} \geq t$, then at least for one $k \in\{1, \ldots, F_h\}$, the inequality $\left|\mathbf{x}_{i,k} - \mathbf{x}_{j,k}  \right| \geq \frac{t}{\sqrt{F_h}}$ holds. Therefore, we have
$$
\begin{aligned}
\mathbb{P}\left(\left\| \mathbf{x}_{i}- \mathbf{x}_{j}  \right\|_{2} \geq t \right) & \leq \mathbb{P}\left(\bigcup_{k=1}^{F_h}\left\{\left|\mathbf{x}_{i,k} - \mathbf{x}_{j,k} \right| \geq \frac{t}{\sqrt{F_h}} \right\}\right) \\
& \leq \sum_{k=1}^{F_h} \mathbb{P}\left( \left|\mathbf{x}_{i,k} - \mathbf{x}_{j,k} \right| \geq \frac{t}{\sqrt{F_h}} \right)\\
\end{aligned}
$$
The results for $\bm{h}_i, \bm{h}_i^\text{HP}$ can be proved by analogy.
\end{proof}
\begin{lemma}2(Heoffding Lemma) Let $X$ be any real-valued random variable such that $a\leq X\leq b$ almost surely, \ie{} with probability one. Then, for all $\lambda \in \mathbb {R} $,

$$ \mathbb {E} \left[e^{\lambda X}\right]\leq \exp {\Big (}\lambda \mathbb {E} [X]+{\frac {\lambda ^{2}(b-a)^{2}}{8}}{\Big )}$$

\end{lemma}

\begin{proof}
    See \cite{massart2007concentration}
\end{proof}
\begin{theorem} 2
For nodes $i,j,v \in \mathcal{V}$, suppose $z_i \neq z_j$ and $z_i = z_v$, then for constants $t_x,t_h,t_\text{HP}$ that satisfy $t_x \geq \sqrt{F_h} D_x(i,j), \; t_h \geq \sqrt{F_h} D_h(i,j), \; t_\text{HP} \geq \sqrt{F_h} D_\text{HP}(i,j)$ we have
\begin{equation}
\begin{aligned}
&   \mathbb{P}\left(\norm{\bm{x}_i - \bm{x}_j}_2  \geq \norm{\bm{x}_i - \bm{x}_v}_2 + t_x \right) \leq 2 F_h \exp{\left( -\frac{(D_x(i,j) - \frac{t_x}{\sqrt{F_h}} )^2}{V_x(i,j)} \right)},  \\
 & \mathbb{P}(\norm{\bm{h}_i-\bm{h}_j}_2  \geq \norm{\bm{h}_i-\bm{h}_v}_2 + t_h) \leq 2 F_h \exp{\left( -\frac{(D_h(i,j) - \frac{t_h}{\sqrt{F_h}} )^2}{V_h(i,j)} \right)}, \\
& \mathbb{P}(\norm{\bm{h}_i^\text{HP} - \bm{h}_j^\text{HP}}_2  \geq \norm{\bm{h}_i^\text{HP} - \bm{h}_v^\text{HP}}_2  + t_\text{HP}) \leq 2 F_h \exp{\left( -\frac{\left(D_\text{HP}(i,j) - \frac{t_\text{HP}}{\sqrt{F_h}} \right)^2}{V_\text{HP}(i,j)} \right)}, \\
\end{aligned}
\end{equation}
where 
\begin{align*}
 &D_x(i,j) = \norm{\bm{\mu}_{z_i} - \bm{\mu}_{z_j}}_2, \; V_x(i,j) = (b-a)^2, \\
 &D_h(i,j) = \norm{\tilde{\bm{\mu}}_{z_i} - \tilde{\bm{\mu}}_{z_j}}_2, V_h(i,j) = \left(\frac{1}{2d_i} + \frac{1}{2d_j}\right)(b-a)^2, \\
 &D_\text{HP}(i,j) = \norm{\bm{{\mu}}_{z_i} - \bm{\tilde{\mu}}_{z_i} - \left( \bm{{\mu}}_{z_j} - \bm{\tilde{\mu}}_{z_j} \right) }_2, \\
 & V_\text{HP}(i,j) = \left(1+\frac{1}{2d_i} + \frac{1}{2d_j}\right)(b-a)^2, \\
 & \bm{\tilde{\mu}}_{z_i} =  \sum\limits_{\substack{v \in \mathcal{N}(i)} } \mathbb{E}_{\substack{ z_{v} \sim \mathcal{D}_{z_i}, \\ \mathbf{x}_{v} \sim \mathcal{F}_{z_{v}}} } \left[ \frac{1}{d_i}  \bm{x}_{v} \right].
\end{align*}
\end{theorem}

The proof of Theorem 2 will be splitted into three parts for $\bm{x}_i, \bm{h}_i$ and  $\bm{h}_i^\text{HP}$, respectively.

\subsection{Proof for Original (Full-pass Filtered) Features}

\begin{proof}

Since we have 
$$\norm{\bm{x}_i-\bm{x}_j}_2  - \norm{\bm{x}_i-\bm{x}_v}_2 \leq \norm{\bm{x}_i-\bm{x}_j - (\bm{x}_i-\bm{x}_v)}_2 = \norm{\bm{x}_v-\bm{x}_j}_2$$
then 
\begin{align*}
\mathbb{P}&\left(\norm{\bm{x}_i-\bm{x}_j}_2  \geq \norm{\bm{x}_i-\bm{x}_v}_2 + t_x \right)  = \mathbb{P}\left(\norm{\bm{x}_i-\bm{x}_j}_2 - \norm{\bm{x}_i-\bm{x}_v}_2 \geq  t_x \right)\\
&\leq \mathbb{P}\left(\norm{\bm{x}_i-\bm{x}_j - (\bm{x}_i-\bm{x}_v)}_2 \geq t_x \right) = \mathbb{P}\left(\norm{\bm{x}_v-\bm{x}_j)}_2 \geq t_x \right)
\end{align*}
We will calculate the upper bound of $\mathbb{P}\left(\norm{\bm{x}_v - \bm{x}_j}_2 \geq t_x \right)$ in the following part. To do this, we first compute the upper bound of $\mathbb{P}\left(\bm{x}_{v,k}-\bm{x}_{j,k} \geq t \right)$.
For $t \geq \norm{\bm{{\mu}}_{z_v} -\bm{{\mu}}_{z_j} }$ and any $s \geq 0$, we have

$$
\begin{aligned}
    & \ \ \mathbb{P}\left( \bm{{x}}_{v,k} -\bm{{x}}_{j,k} \geq t\right) 
    =  \mathbb{P}\left( \exp{\left( s(\bm{{x}}_{v,k} -\bm{{x}}_{j,k} ) \right)}   \geq \exp{\left(s t \right)} \right) \\ 
     \leq  \ & \exp{\left(-s t \right)} \mathbb{E}\left[ \exp{\left( s(\bm{{x}}_{v,k} -\bm{{x}}_{j,k} ) \right)}  \right] \; (\text{Markov Inequality}) \\
    =\ &  \exp{\left(-s t\right)} \mathbb{E}\left[ \exp{\left( s\bm{x}_{v,k} \right)}  \right]   \mathbb{E}\left[ \exp{\left( - s\bm{x}_{j,k} \right)}  \right]  \; (\text{Independency})\\
    \leq \ &   \exp{\left(-s t\right)}  \exp{\left( s \mathbb{E}\left[\bm{x}_{v,k}   \right] + \frac{(b-a)^2s^2}{8 }\right)}  \times \exp{\left( -s \mathbb{E}\left[  \bm{x}_{j,k} \right] + \frac{(b-a)^2s^2}{8 } \right)}   \; (\text{Hoeffding's lemma})\\
    =\ &  \exp{\left( \frac{(b-a)^2}{4} s^2 + (\bm{{\mu}}_{z_v,k} -\bm{{\mu}}_{z_j,k} - t )s \right)} \\ 
   \leq  \ &  \exp{\left( \frac{(b-a)^2}{4} s^2 + ( \left|\bm{{\mu}}_{z_v,k} -\bm{{\mu}}_{z_j,k} \right| - t )s \right)} 
\end{aligned}
$$
Since $t \geq \norm{\bm{{\mu}}_{z_v} -\bm{{\mu}}_{z_j} } \geq \left| \bm{{\mu}}_{z_v,k} -\bm{{\mu}}_{z_j,k}  \right|$ for any $k$, so when $s=-\frac{(\left| \bm{{\mu}}_{z_v,k} -\bm{{\mu}}_{z_j,k}  \right| - t )}{\frac{(b-a)^2}{2}} \geq 0$, we get the tightest bound of the above inequality and 
$$\exp{\left( -\frac{(\left| \bm{{\mu}}_{z_v,k} - \bm{{\mu}}_{z_j,k} \right| -t )^2}{(b-a)^2} \right)} \leq \exp{\left( -\frac{(\norm{{\bm{\mu}}_{z_v}-{\bm{\mu}}_{z_j}}_2 - t )^2}{(b-a)^2} \right)} $$
With the same steps, we have 
$$
\mathbb{P}\left( \bm{{x}}_{v,k} -\bm{{x}}_{j,k}  \leq -t \right)  = \mathbb{P}\left( \bm{{x}}_{j,k} - \bm{{x}}_{v,k}  \geq t \right) \leq \exp{\left( -\frac{(\norm{ {\bm{\mu}}_{z_v}- {\bm{\mu}}_{z_j}}_2 - t )^2}{(b-a)^2} \right)}
$$
Combined together we have
$$
\mathbb{P}\left( \big|\bm{{x}}_{v,k} -\bm{{x}}_{j,k} \big| \geq t \right) \leq 2 \exp{\left( -\frac{(\norm{ {\bm{\mu}}_{z_v}- {\bm{\mu}}_{z_j}}_2 - t )^2}{(b-a)^2} \right)} 
$$
Since $\frac{t_x}{\sqrt{F_h}} \geq \norm{\bm{ {\mu}}_{z_v} -\bm{ {\mu}}_{z_j} }$, then from Lemma 1 we have
\begin{equation}
\begin{aligned}\mathbb{P}\left(\norm{ \bm{{x}}_{v} - \bm{{x}}_{j} }_2 \geq t_x \right) \leq \sum_{k=1}^{F_h} \mathbb{P}\left( \left|\mathbf{x}_{v,k} - \mathbf{x}_{j,k} \right| \geq \frac{t_x}{\sqrt{F_h}} \right) \leq 2 F_h \exp{\left( -\frac{(\norm{ {\bm{\mu}}_{z_v} -  {\bm{\mu}}_{z_j}}_2 - \frac{t_x}{\sqrt{F_h}} )^2}{(b-a)^2} \right)} \\
\end{aligned}
\end{equation}

\end{proof}

\subsection{Proof for Low-pass Filter}

\begin{proof} Part for LP filter:

Let $\bm{{h}}_{i,k}  = \frac{1}{d_i} \sum\limits_{\substack{u \in \mathcal{N}(i),\\ z_{u} \sim \mathcal{D}_{z_i}, \\ \mathbf{x}_{u,k} \sim \mathcal{F}_{z_{u},k}} } \bm{x}_{u,k}$ and $\bm{\tilde{\mu}}_{z_i,k} = \mathbb{E}\left[\bm{{h}}_{i,k} \right] = \mathbb{E} \left[\frac{1}{d_i} \sum \mathbf{x}_{u,k} \right]$. Since we have 
$$\norm{\bm{h}_i-\bm{h}_j}_2  - \norm{\bm{h}_i-\bm{h}_v}_2 \leq \norm{\bm{h}_i-\bm{h}_j - (\bm{h}_i-\bm{h}_v)}_2 = \norm{\bm{h}_v-\bm{h}_j}_2$$
then 
\begin{align*}
\mathbb{P}&\left(\norm{\bm{h}_i-\bm{h}_j}_2  \geq \norm{\bm{h}_i-\bm{h}_v}_2 + t_h \right)  = \mathbb{P}\left(\norm{\bm{h}_i-\bm{h}_j}_2 - \norm{\bm{h}_i-\bm{h}_v}_2 \geq  t_h \right)\\
&\leq \mathbb{P}\left(\norm{\bm{h}_i-\bm{h}_j - (\bm{h}_i-\bm{h}_v)}_2 \geq t_h \right) = \mathbb{P}\left(\norm{\bm{h}_v-\bm{h}_j)}_2 \geq t_h \right)
\end{align*}
We will calculate the upper bound of $\mathbb{P}\left(\norm{\bm{h}_v - \bm{h}_j}_2 \geq t_h \right)$ in the following part. To do this, we first compute the upper bound of $\mathbb{P}\left(\bm{h}_{v,k}-\bm{h}_{j,k} \geq t \right)$.
For $t \geq \norm{\bm{\tilde{\mu}}_{z_v} -\bm{\tilde{\mu}}_{z_j} }$ and any $s \geq 0$, we have

$$
\begin{aligned}
    & \ \ \mathbb{P}\left( \bm{{h}}_{v,k} -\bm{{h}}_{j,k} \geq t\right) 
    =  \mathbb{P}\left( \exp{\left( s(\bm{{h}}_{v,k} -\bm{{h}}_{j,k} ) \right)}   \geq \exp{\left(s t \right)} \right) \\ 
     \leq  \ & \exp{\left(-s t \right)} \mathbb{E}\left[ \exp{\left( s(\bm{{h}}_{v,k} -\bm{{h}}_{j,k} ) \right)}  \right] \; (\text{Markov Inequality}) \\
    =\ &  \exp{\left(-s t\right)} \mathbb{E}\left[ \exp{\left( \frac{s}{d_v} \sum\limits_{\substack{u \in \mathcal{N}(v),\\ z_{u} \sim \mathcal{D}_{z_v}, \\ \mathbf{x}_{u,k} \sim \mathcal{F}_{z_{u},k}} } \bm{x}_{u,k} \right)}  \right]   \mathbb{E}\left[ \exp{\left( \frac{-s}{d_j} \sum\limits_{\substack{u \in \mathcal{N}(j),\\ z_{u} \sim \mathcal{D}_{z_j}, \\ \mathbf{x}_{u,k} \sim \mathcal{F}_{z_{u},k}} } \bm{x}_{u,k} \right)}  \right]  \; (\text{Independency})\\
     =\ &  \exp{\left(-s t\right)} \prod\limits_{\substack{u \in \mathcal{N}(v),\\ z_{u} \sim \mathcal{D}_{z_v}, \\ \mathbf{x}_{u,k} \sim \mathcal{F}_{z_{u},k}} }  \mathbb{E}\left[ \exp{\left( \frac{s}{d_v}  \bm{x}_{u,k} \right)}  \right]   \prod\limits_{\substack{u \in \mathcal{N}(j),\\ z_{u} \sim \mathcal{D}_{z_j}, \\ \mathbf{x}_{u,k} \sim \mathcal{F}_{z_{u},k}} } \mathbb{E}\left[ \exp{\left( \frac{-s}{d_j} \bm{x}_{u,k} \right)}  \right]  \; (\text{Independency})\\
    \leq \ &   \exp{\left(-s t\right)} \prod\limits_{\substack{u \in \mathcal{N}(v),\\ z_{u} \sim \mathcal{D}_{z_v}, \\ \mathbf{x}_{u,k} \sim \mathcal{F}_{z_{u},k}} } \exp{\left(  \frac{s}{d_v} \mathbb{E}\left[\bm{x}_{u,k}   \right] + \frac{(b-a)^2s^2}{8d_v^2}\right)}  \\ 
    & \ \  \times \prod\limits_{\substack{u \in \mathcal{N}(j),\\ z_{u} \sim \mathcal{D}_{z_j}, \\ \mathbf{x}_{u,k} \sim \mathcal{F}_{z_{u,k}}} } \exp{\left(\frac{-s}{d_j}  \mathbb{E}\left[  \bm{x}_{u,k} \right] + \frac{(b-a)^2s^2}{8d_j^2} \right)}   \; (\text{Hoeffding's lemma})\\
    =\ &  \exp{\left(-s t\right)} \exp{\left(\frac{(b-a)^2s^2}{8d_v}\right)} \exp{\left(  s \mathbb{E}\left[ \frac{1}{d_v} \sum\limits_{\substack{u \in \mathcal{N}(v),\\ z_{u} \sim \mathcal{D}_{z_v}, \\ \mathbf{x}_{u,k} \sim \mathcal{F}_{z_{u},k}} } \bm{x}_{u,k}   \right] \right)}  \\ 
    & \ \  \times \exp{\left(\frac{(b-a)^2s^2}{8d_j}\right)} \exp{\left( -s  \mathbb{E}\left[  \frac{1}{d_j} \sum\limits_{\substack{u \in \mathcal{N}(j),\\ z_{u} \sim \mathcal{D}_{z_j}, \\ \mathbf{x}_{u,k} \sim \mathcal{F}_{z_{u},k}} } \bm{x}_{u,k} \right] \right)}   \; \\
    =\ &  \exp{\left( \Big(\frac{(b-a)^2}{8d_v} + \frac{(b-a)^2}{8d_j}\Big)s^2 + (\bm{\tilde{\mu}}_{z_v,k} -\bm{\tilde{\mu}}_{z_j,k} - t )s \right)} \\ 
   \leq  \ &  \exp{\left( \Big(\frac{(b-a)^2}{8d_v} + \frac{(b-a)^2}{8d_j}\Big)s^2 + ( \left|\bm{\tilde{\mu}}_{z_v,k} -\bm{\tilde{\mu}}_{z_j,k} \right| - t )s \right)} 
\end{aligned}
$$
Since $t\geq \norm{\bm{\tilde{\mu}}_{z_v} -\bm{\tilde{\mu}}_{z_j} } \geq \left| \bm{\tilde{\mu}}_{z_v,k} -\bm{\tilde{\mu}}_{z_j,k}  \right|$ for any $k$, so when $s=-\frac{(\left| \bm{\tilde{\mu}}_{z_v,k} -\bm{\tilde{\mu}}_{z_j,k}  \right| - t )}{\frac{(b-a)^2}{4 d_v} + \frac{(b-a)^2}{4 d_j}} \geq 0$, we get the tightest bound of the above inequality and 
$$\exp{\left( -\frac{(\left| \bm{\tilde{\mu}}_{z_v,k} - \bm{\tilde{\mu}}_{z_j,k} \right| -t )^2}{\frac{(b-a)^2}{2d_v} + \frac{(b-a)^2}{2d_j}} \right)} \leq \exp{\left( -\frac{(\norm{\tilde{\bm{\mu}}_{z_v}-\tilde{\bm{\mu}}_{z_j}}_2 - t )^2}{\frac{(b-a)^2}{2d_v} + \frac{(b-a)^2}{2d_j}} \right)} $$
With the same steps, we have 
$$
\mathbb{P}\left( \bm{{h}}_{v,k} -\bm{{h}}_{j,k}  \leq -t \right)  = \mathbb{P}\left( \bm{{h}}_{j,k} - \bm{{h}}_{v,k}  \geq t \right) \leq \exp{\left( -\frac{(\norm{\tilde{\bm{\mu}}_{z_v}-\tilde{\bm{\mu}}_{z_j}}_2 - t )^2}{\frac{(b-a)^2}{2d_v} + \frac{(b-a)^2}{2d_j}} \right)}
$$
Combined together we have
$$
\mathbb{P}\left( \big|\bm{{h}}_{v,k} -\bm{{h}}_{j,k} \big| \geq t \right) \leq 2 \exp{\left( -\frac{(\norm{\tilde{\bm{\mu}}_{z_v}-\tilde{\bm{\mu}}_{z_j}}_2 - t )^2}{\frac{(b-a)^2}{2d_v} + \frac{(b-a)^2}{2d_j}} \right)} 
$$
Since $\frac{t_h}{\sqrt{F_h}}\geq \norm{\bm{\tilde{\mu}}_{z_v} -\bm{\tilde{\mu}}_{z_j} }$, then from Lemma 1 we have
\begin{equation}
\begin{aligned}\mathbb{P}\left(\norm{ \bm{{h}}_{v} - \bm{{h}}_{j} }_2 \geq t_h \right) \leq \sum_{k=1}^{F_h} \mathbb{P}\left( \left|\mathbf{h}_{v,k} - \mathbf{h}_{j,k} \right| \geq \frac{t_h}{\sqrt{F_h}} \right) \leq 2 F_h \exp{\left( -\frac{(\norm{\tilde{\bm{\mu}}_{z_v} - \tilde{\bm{\mu}}_{z_j}}_2 - \frac{t_h}{\sqrt{F_h}} )^2}{\frac{(b-a)^2}{2d_v} + \frac{(b-a)^2}{2d_j}} \right)} \\
\end{aligned}
\end{equation}

\end{proof}
\subsection{Theoretical Results for High-pass Filter}

\begin{proof} 

The proof for HP filter is similar to that of LP filter.

Let $\bm{{h}}_{i}^\text{HP}  = \bm{x}_i - \bm{{h}}_i$, which is the HP filtered signal. Since we have 
$$\norm{\bm{h}_i^\text{HP} - \bm{h}_j^\text{HP}}_2  - \norm{\bm{h}_i^\text{HP} - \bm{h}_v^\text{HP}}_2 \leq \norm{\bm{h}_i^\text{HP}  - \bm{h}_j^\text{HP} - (\bm{h}_i^\text{HP} - \bm{h}_v^\text{HP})}_2 = \norm{\bm{h}_v^\text{HP}  - \bm{h}_j^\text{HP}}_2$$
then 
\begin{align*}
\mathbb{P} & \left(\norm{\bm{h}_i^\text{HP} - \bm{h}_j^\text{HP}}_2  \geq \norm{\bm{h}_i^\text{HP} - \bm{h}_v^\text{HP}}_2 + t_\text{HP} \right)  = \mathbb{P}\left(\norm{\bm{h}_i^\text{HP} - \bm{h}_j^\text{HP}}_2 - \norm{\bm{h}_i^\text{HP} - \bm{h}_v^\text{HP}}_2 \geq  t_\text{HP} \right)\\
&\leq \mathbb{P}\left(\norm{\bm{h}_i^\text{HP} - \bm{h}_j^\text{HP} - (\bm{h}_i^\text{HP} - \bm{h}_v^\text{HP})}_2 \geq t_\text{HP} \right) = \mathbb{P}\left(\norm{\bm{h}_v^\text{HP} - \bm{h}_j^\text{HP})}_2 \geq t_\text{HP} \right)
\end{align*}
We will calculate the upper bound of $\mathbb{P}\left(\norm{\bm{h}_v^\text{HP} - \bm{h}_j^\text{HP}}_2 \geq t \right)$ in the following part. To do this, we first compute the upper bound of $\mathbb{P}\left(\bm{h}_{v,k}^\text{HP} - \bm{h}_{j,k}^\text{HP} \geq t \right)$. 

For $t \geq \norm{\bm{{\mu}}_v - \bm{\tilde{\mu}}_v - \left( \bm{{\mu}}_j - \bm{\tilde{\mu}}_j \right) }_2$ and $s \geq 0$, we have

\begin{align*}
&\mathbb{P}\left( \bm{{h}}_{v,k}^\text{HP} - \bm{{h}}_{j,k}^\text{HP} \geq t\right) =  \mathbb{P}\left(\bm{x}_{v,k} - \bm{{h}}_{v,k} - \bm{x}_{j,k} + \bm{{h}}_{j,k} \geq t\right)\\
&= \mathbb{P}\left( \exp{\left( s(\bm{x}_{v,k}- \bm{{h}}_{v,k} - \bm{x}_{j,k} + \bm{{h}}_{j,k} ) \right)}   \geq \exp{\left(st\right)}\right)\\
&\leq  \exp{\left(-st\right)} \mathbb{E}\left[ \exp{\left( s(\bm{x}_{v,k} - \bm{{h}}_{v,k} - \bm{x}_{j,k} + \bm{{h}}_{j,k}) \right)}  \right] \; (\text{Markov Inequality}) \\
&=  \exp{\left(-st\right)} \times \mathbb{E} \left[ \exp{\left(s \bm{x}_{v,k} \right)} \right] \times \mathbb{E} \left[ \exp{\left(-s \bm{x}_{j,k} \right)} \right] \times  \\
& \mathbb{E}\left[ \exp{\left( -\frac{s}{d_v} \sum\limits_{\substack{u \in \mathcal{N}(v),\\ z_{u} \sim \mathcal{D}_{z_v}, \\ \mathbf{x}_{u,k} \sim \mathcal{F}_{z_{u},k}} } \bm{x}_{u,k} \right)}  \right]   \mathbb{E}\left[ \exp{\left( \frac{s}{d_j} \sum\limits_{\substack{u \in \mathcal{N}(j),\\ z_{u} \sim \mathcal{D}_{z_j}, \\ \mathbf{x}_{u,k} \sim \mathcal{F}_{z_{u,k}}} } \bm{x}_{u,k} \right)}  \right]  \; (\text{Independency})\\
&\leq  \exp{\left(-st\right)} \mathbb{E} \left[ \exp{\left(s \bm{x}_{v,k} \right)} \right] \mathbb{E} \left[ \exp{\left(-s \bm{x}_{j,k}\right)} \right]\\
&\prod\limits_{\substack{u \in \mathcal{N}(v),\\ z_{u} \sim \mathcal{D}_{z_v}, \\ \mathbf{x}_{u,k} \sim \mathcal{F}_{z_{u},k}} }  \mathbb{E}\left[ \exp{\left( \frac{-s}{d_v}  \bm{x}_{u,k} \right)}  \right]   \prod\limits_{\substack{u \in \mathcal{N}(j),\\ z_{u} \sim \mathcal{D}_{z_j}, \\ \mathbf{x}_{u,k} \sim \mathcal{F}_{z_{u,k}}} } \mathbb{E}\left[ \exp{\left( \frac{s}{d_j} \bm{x}_{u,k} \right)}  \right]  \\
&\leq  \exp{\left(-st\right)}  \exp{\left(  s \bm{{\mu}}_{v,k} + \frac{(b-a)^2s^2}{8}\right)}   \exp{\left(-s \bm{{\mu}}_{j,k} + \frac{(b-a)^2s^2}{8} \right)}\\
& \prod \exp{\left(\frac{-s}{d_v} \mathbb{E}[\bm{x}_{u,k}] + \frac{(b-a)^2s^2}{8d_v^2}\right)}  \prod \exp{\left(\frac{s}{d_j}  \mathbb{E}[\bm{x}_{u,k}] + \frac{(b-a)^2s^2}{8d_j^2} \right)} \; (\text{Hoeffding's lemma})\\
& =  \exp{\left( (\frac{(b-a)^2}{4}+\frac{(b-a)^2}{8d_v} + \frac{(b-a)^2}{8d_j})s^2 +\left(\bm{{\mu}}_{v,k} - \bm{{\mu}}_{j,k} -\left( \bm{\tilde{\mu}}_{v,k} -\bm{\tilde{\mu}}_{j,k} \right) -t \right)s \right)} \\
& \leq \exp{\left( (\frac{(b-a)^2}{4}+\frac{(b-a)^2}{8d_v} + \frac{(b-a)^2}{8d_j})s^2 + \left( \left| \bm{{\mu}}_{v,k} - \bm{{\mu}}_{j,k} -\left( \bm{\tilde{\mu}}_{v,k} -\bm{\tilde{\mu}}_{j,k}\right)  \right| - t \right)s \right)} 
\end{align*}

Since $t \geq \norm{\bm{{\mu}}_v - \bm{\tilde{\mu}}_v - \left( \bm{{\mu}}_j - \bm{\tilde{\mu}}_j \right) }_2$, then when $s=-\frac{\left| \bm{{\mu}}_{v,k} - \bm{{\mu}}_{j,k} -\left( \bm{\tilde{\mu}}_{v,k} -\bm{\tilde{\mu}}_{j,k} \right) \right| - t }{(\frac{(b-a)^2}{2}+\frac{(b-a)^2}{4d_v} + \frac{(b-a)^2}{4d_j})} > 0$, we get the tightest bound and
$$\exp{\left( -\frac{\left( \left| \bm{{\mu}}_{v,k} - \bm{{\mu}}_{j,k} -\left( \bm{\tilde{\mu}}_{v,k} -\bm{\tilde{\mu}}_{j,k}\right)  \right| - t \right)^2}{(1+\frac{1}{2d_v} + \frac{1}{2d_j})(b-a)^2} \right)} \leq \exp{\left( -\frac{\left( \norm{\bm{{\mu}}_v - \bm{\tilde{\mu}}_v - \left( \bm{{\mu}}_j - \bm{\tilde{\mu}}_j \right) }_2 - t \right)^2}{(1+\frac{1}{2d_v} + \frac{1}{2d_j})(b-a)^2} \right)} $$
Then 
$$\mathbb{P}\left( \big|\bm{{h}}_{v,k}^\text{HP} - \bm{{h}}_{j,k}^\text{HP} \big| \geq t\right) \leq 2 \exp{\left( -\frac{\left( \norm{\bm{{\mu}}_v - \bm{\tilde{\mu}}_v - \left( \bm{{\mu}}_j - \bm{\tilde{\mu}}_j \right) }_2 - t \right)^2}{(1+\frac{1}{2d_v} + \frac{1}{2d_j})(b-a)^2} \right)} $$
Since $\frac{t_\text{HP}}{\sqrt{F_h}}\geq \norm{\bm{{\mu}}_v - \bm{\tilde{\mu}}_v - \left( \bm{{\mu}}_j - \bm{\tilde{\mu}}_j \right) }_2$, then from Lemma 1, we have 
\begin{equation}
\begin{aligned}
\mathbb{P}\left(\norm{ \bm{{h}}_{v}^\text{HP} - \bm{{h}}_{j}^\text{HP} }_2 \geq t_\text{HP} \right) & \leq \sum_{k=1}^{F_h} \mathbb{P}\left( \left|\mathbf{h}_{v,k}^\text{HP} - \mathbf{h}_{j,k}^\text{HP} \right| \geq \frac{t_\text{HP}}{\sqrt{F_h}} \right) \\ 
& \leq 2 F_h \exp{\Bigg( -\frac{\left( \norm{\bm{{\mu}}_v - \bm{\tilde{\mu}}_v - \left( \bm{{\mu}}_j - \bm{\tilde{\mu}}_j \right) }_2 - \frac{t_\text{HP}}{\sqrt{F_h}} \right)^2}{(1+\frac{1}{2d_v} + \frac{1}{2d_j})(b-a)^2} \Bigg)} \\
\end{aligned}
\end{equation}
$\norm{\bm{{\mu}}_v - \bm{\tilde{\mu}}_v - \left( \bm{{\mu}}_j - \bm{\tilde{\mu}}_j \right) }_2$ is essentially the relative center movement.

\end{proof}

\section{Detailed Discussion of Performance Metrics and More Experimental Results}
\label{appendix:detailed_discussion_performance_metrics_more_experimental_results}

\subsection{Hypothesis Testing of ACM-GNNs vs. GNNs}
To more comprehensively validate if "intra-class embedding distance is smaller than the inter-class embedding distance" closely correlates to the superiority of a given model versus another model, we choose the SOTA model ACM-GNNs and conduct the following hypothesis testing of ACM-GNNs \cite{luan2022revisiting} versus GNNs and ACM-GNNs versus MLPs. From the results in table \ref{tab:tab_acmgnns_vs_gnns} we can see that the above statements hold except in ACM-SGC-1 vs. SGC-1 on Squirrel and ACM-GCN vs. GCN on CiteSeer. This again verifies that the relationship between intra- and inter-class embedding distance strongly relates to the model performance.
\begin{table}[htbp]
  \centering
  \caption{Hypothesis testing results of ACM-GNNs v.s. GNNs: The cells marked by orange are the cases that the p-values significantly indicate the opposite direction as the trained results (ground truth).}
  \resizebox{1\hsize}{!}{
    \begin{tabular}{c|c|ccccccccc}
    \toprule
    \toprule
          &       & Cornell & Wisconsin & Texas & Film  & Chameleon & Squirrel & Cora  & CiteSeer & PubMed \\
    \midrule
          &  p-value & 1.00  & 1.00  & 1.00  & 1.00  & 0.19  & \cellcolor[rgb]{ .929,  .49,  .192}1.00 & 1.00  & 1.00  & 1.00 \\
    {ACM-SGC-1 } & ACC ACM-SGC-1 & 93.77 $\pm$ 1.91 & 93.25 $\pm$ 2.92 & 93.61 $\pm$ 1.55 & 39.33 $\pm$ 1.25 & 63.68 $\pm$ 1.62 & 46.4 $\pm$ 1.13 & 86.63 $\pm$ 1.13 & 80.96 $\pm$ 0.93 & 87.75 $\pm$ 0.88 \\
    {v.s. SGC-1} & ACC SGC-1 & 70.98 $\pm$ 8.39 & 70.38 $\pm$ 2.85 & 83.28 $\pm$ 5.43 & 25.26 $\pm$ 1.18 & 64.86 $\pm$ 1.81 & 47.62 $\pm$ 1.27 & 85.12 $\pm$ 1.64 & 79.66 $\pm$ 0.75 & 85.5 $\pm$ 0.76 \\
          & \textbf{Diff Acc} & 22.79 & 22.87 & 10.33 & 14.07 & -1.18 & -1.22 & 1.51  & 1.30  & 2.25 \\
    \midrule
          & p-value & 1.00  & 1.00  & 1.00  & 1.00  & 1.00  & 1.00  & 0.41  & \cellcolor[rgb]{ .929,  .49,  .192}0.00 & 1.00 \\
    {ACM-GCN } & ACC ACM-GCN & 94.75 $\pm$ 3.8 & 95.75 $\pm$ 2.03 & 94.92 $\pm$ 2.88 & 41.62 $\pm$ 1.15 & 69.04 $\pm$ 1.74 & 58.02 $\pm$ 1.86 & 88.62 $\pm$ 1.22 & 81.68 $\pm$ 0.97 & 90.66 $\pm$ 0.47 \\
    {v.s. GCN} & ACC GCN & 82.46 $\pm$ 3.11 & 75.5 $\pm$ 2.92 & 83.11 $\pm$ 3.2 & 35.51 $\pm$ 0.99 & 64.18 $\pm$ 2.62 & 44.76 $\pm$ 1.39 & 87.78 $\pm$ 0.96 & 81.39 $\pm$ 1.23 & 88.9 $\pm$ 0.32 \\
          & \textbf{Diff Acc} & 12.29 & 20.25 & 11.81 & 6.11  & 4.86  & 13.26 & 0.84  & 0.29  & 1.76 \\
    \midrule
          & p-value & 0.10  & 0.00  & 0.50  & 1.00  & 1.00  & 1.00  & 1.00  & 1.00  & 0.42 \\
    {ACM-SGC-1 } & ACC ACM-SGC-1 & 93.77 $\pm$ 1.91 & 93.25 $\pm$ 2.92 & 93.61 $\pm$ 1.55 & 39.33 $\pm$ 1.25 & 63.68 $\pm$ 1.62 & 46.4 $\pm$ 1.13 & 86.63 $\pm$ 1.13 & 80.96 $\pm$ 0.93 & 87.75 $\pm$ 0.88 \\
    {v.s. MLP-1} & ACC MLP-1 & 93.77 $\pm$ 3.34 & 93.87 $\pm$ 3.33 & 93.77 $\pm$ 3.34 & 34.53 $\pm$ 1.48 & 45.01 $\pm$ 1.58 & 29.17 $\pm$ 1.46 & 74.3 $\pm$ 1.27 & 75.51 $\pm$ 1.35 & 86.23 $\pm$ 0.54 \\
          & \textbf{Diff Acc} & 0.00  & -0.62 & -0.16 & 4.80  & 18.67 & 17.23 & 12.33 & 5.45  & 1.52 \\
    \midrule
          & p-value & 0.94  & 1.00  & 1.00  & 1.00  & 1.00  & 1.00  & 1.00  & 1.00  & 1.00 \\
    {ACM-GCN} & ACC ACM-GCN & 94.75 $\pm$ 3.8 & 95.75 $\pm$ 2.03 & 94.92 $\pm$ 2.88 & 41.62 $\pm$ 1.15 & 69.04 $\pm$ 1.74 & 58.02 $\pm$ 1.86 & 88.62 $\pm$ 1.22 & 81.68 $\pm$ 0.97 & 90.66 $\pm$ 0.47 \\
    { v.s. MLP-2} & ACC MLP-2 & 91.30 $\pm$ 0.70 & 93.87 $\pm$ 3.33 & 92.26 $\pm$ 0.71 & 38.58 $\pm$ 0.25 & 46.72 $\pm$ 0.46 & 31.28 $\pm$ 0.27 & 76.44 $\pm$ 0.30 & 76.25 $\pm$ 0.28 & 86.43 $\pm$ 0.13 \\
          & \textbf{Diff Acc} & 3.45  & 1.88  & 2.66  & 3.04  & 22.32 & 26.74 & 12.18 & 5.43  & 4.23 \\
    \bottomrule
    \bottomrule
    \end{tabular}%
    }
  \label{tab:tab_acmgnns_vs_gnns}%
\end{table}%

\subsection{Implementation Details of KR and GNB}

Classifier-based performance metrics: we measure the quality of aggregated features based on the performance of a "training-free" classifier. In this paper, we take use of kernel regression and naive Bayes classifiers.

\paragraph{Kernel Regression} Kernel method utilizes a pairwise similarity function $K(\bm{x}_i, \bm{x}_j)$ to measure how closely related two node embeddings are, without the need for any training process \cite{koutroumbas2008pattern, hofmann2008kernel, mohri2018foundations}. A higher value of $K(\bm{x}_i, \bm{x}_j)$ indicates a smaller distance between the embeddings of nodes $\bm{x}_i$ and $\bm{x}_j$ and vice versa. 
\begin{algorithm}
\caption{Pseudo code for kernel regression}\label{alg:kernel_regression}
\begin{algorithmic}
\Require $X, \hat{A}, Z, N, N_S, N_\text{epochs}$    \Comment{$N$ is the number of nodes, $N_S$ is the number of samples}
\For{$i$ in $N_\text{epochs}$}
\State $S \gets \text{sample}(N,N_S)$
\State Get $K_S^X, K_S^H, Z_S$                \Comment{$K_S^X, K_S^H$ are the kernels for $X$ and $H$ for the sampled nodes}
\State $S_\text{train}, S_\text{test} \gets \text{sample}(S, 0.6N_S, 0.4N_S)$
\State $f_K^X \gets \left( K_S^X[S_\text{test},:][:,S_\text{train}] \right) \left( K_S^X[S_\text{train},:][:,S_\text{train}] \right)^{-1} Z_S[S_\text{train},:]$
\State $f_K^H \gets \left( K_S^H[S_\text{test},:][:,S_\text{train}] \right) \left( K_S^H[S_\text{train},:][:,S_\text{train}] \right)^{-1} Z_S[S_\text{train},:]$
\State Compute $\text{ACC}^X_i, \text{ACC}^H_i \gets \text{Accuracy}(f_K^X, Z_S[S_\text{test},:]),\ \text{Accuracy}(f_K^H, Z_S[S_\text{test},:]$)
\EndFor
\State $\text{p-value} \gets \text{ttest}( \text{ACC}^X, \text{ACC}^H)$
\end{algorithmic}
\end{algorithm}

To capture the \textbf{feature-based non-linear node similarity}, we use Neural Network Gaussian Process (NNGP) \cite{lee2017deep,arora2019exact,garriga2018deep,matthews2018gaussian}. Specifically, we consider the activation function $\phi(x) = \text{ReLU}(x)$ and have

\begin{align*}
&K_\text{NL}(\bm{x}_i, \bm{x}_j) = \frac{1}{2\pi}\left(\bm{x}_i^\top \bm{x}_j\left(\pi-\tilde{\phi}\left(\frac{\bm{x}_i^\top \bm{x}_j}{\|\bm{x}_i\|_2\|\bm{x}_j\|_2}\right)\right) +\sqrt{\|\bm{x}_i\|_2^2\|\bm{x}_j\|_2^2-\left(\bm{x}_i^\top \bm{x}_j\right)^2}\right)
\end{align*}

where $\tilde{\phi}({x}) = \arccos({x})$ is the dual activation function of $\text{ReLU}$. \footnote{Note that when $\phi({x}) =  \exp{(i{x})}$, we have $K(\bm{x}_i, \bm{x}_j) = \exp{(-\frac{1}{2}\norm{\bm{x}_i- \bm{x}_j}_2^2)}$, which is closely related to the Euclidean distance of node embeddings tested in section \ref{sec:hypothesis_testing}, further emphasizing the strong relationship between embedding distances and kernel similarities.}

Furthermore, we observe that there exist some datasets where linear G-aware models do not have the same performance disparities compared to their coupled G-agnostic models as non-linear G-aware models, \eg{} as the results on PubMed shown in table \ref{tab:hypothesis_testing_kernel_homo_performance_comparison}, SGC-1 underperforms MLPs while GCN outperforms MLP-2. This implies that relying on a single non-linear metric to assess whether G-aware models will surpass their coupled G-agnostic models is not enough, we need a linear metric as well. Thus, we choose the following linear kernel (inner product) for regression
\begin{align*}
K_\text{L}(\bm{x}_i, \bm{x}_j) = \frac{\bm{x}_i^\top \bm{x}_j}{\|\bm{x}_i\|_2\|\bm{x}_j\|_2} 
\end{align*}.

For Gaussian Naïve Bayes (GNB), we just use the features and aggregated features of the sampled training nodes to fit two separate classifiers and get the predicted accuracy for the test nodes. Note that Gaussian Naïve Bayes is just a linear classifier.

\paragraph{Threshold Values} Typically, the threshold for homophily and heterophily graphs is set at 0.5 \cite{zhu2020beyond,zhu2020graph,yan2021two,luan2022revisiting} . For classifier-based performance metrics, we establish two benchmark thresholds as below,
\begin{itemize}
    \item Normal Threshold 0.5 (NT0.5): Although not indicating statistical significance, we are still comfortable to set 0.5 as a loose threshold. A value exceeding 0.5 suggests that the G-aware model is not very likely to underperform their coupled G-agnostic model on the tested graph and vice versa.
    \item Statistical Significant Threshold 0.05 (SST0.05): Instead of offering an ambiguous statistical interpretation, SST0.05 will offer a clear statistical meaning. A value smaller than 0.05 implies that the G-aware model significantly underperforms their coupled G-agnostic model and a value greater than 0.95 suggests a high likelihood of G-aware model outperforming their coupled G-agnostic model. Besides, a value ranging from 0.05 to 0.95 indicates no significant performance distinction between G-aware model and its G-agnostic model.
\end{itemize}

We show the results of $\text{KR}_\text{L},\text{KR}_\text{NL}$ and GNB in section \ref{appendix:results_small_large_datasets} . Cells marked by grey are errors according to NT0.5 and results marked by \red{red} are incorrect according to SST 0.05. The comparisons with the existing homophily metrics are shown in section \ref{appendix:statistics_comparisons}. We can see that, no matter on small- (table \ref{tab:detailed_comparison_small_datasets} \ref{tab:statistics_small_scale_datasets}) or large-scale (table \ref{tab:detaied_comparison_large_scale_datasets} \ref{tab:statistics_large_scale_datasets}) datasets, the classifier-based performance metrics (CPMs) are significantly better than the existing homophily metrics on revealing the advantages and disadvantages of GNNs, decreasing the overall error rate from at least 0.34 to 0.13 (table \ref{tab:statistics_overall}). The running time of CPM is short (table \ref{tab:running_time}), only taking several minutes \footnote{1 NVIDIA V100 GPU with 16G memory, 8-core CPU with 16G memory} even on large-scale datasets such as pokec and snap-patents, which contains millons of nodes and tens of millions of edges \cite{lim2021new}.

\subsection{Results on Small-scale and Large-scale Datasets}
\label{appendix:results_small_large_datasets}
\begin{table}[htbp]
  \centering
  \caption{Comparison on small datasets}
  \resizebox{1\hsize}{!}{
    \begin{tabular}{c|c|ccccccccc}
    \toprule
    \toprule
          &       & Cornell & Wisconsin & Texas & Film  & Chameleon & Squirrel & Cora  & CiteSeer & PubMed \\
    \midrule
          & $\text{H}_\text{edge}$ & \cellcolor[rgb]{ .647,  .647,  .647}0.5669 & 0.4480 & 0.4106 & 0.3750 & \cellcolor[rgb]{ .647,  .647,  .647}0.2795 & \cellcolor[rgb]{ .647,  .647,  .647}0.2416 & 0.8100 & 0.7362 & \cellcolor[rgb]{ .647,  .647,  .647}0.8024 \\
     & $\text{H}_\text{node}$ & 0.3855 & 0.1498 & 0.0968 & 0.2210 & \cellcolor[rgb]{ .647,  .647,  .647}0.2470 & \cellcolor[rgb]{ .647,  .647,  .647}0.2156 & 0.8252 & 0.7175 & \cellcolor[rgb]{ .647,  .647,  .647}0.7924 \\
    {Baseline} & $\text{H}_\text{class}$ & 0.0468 & 0.0941 & 0.0013 & 0.0110 & \cellcolor[rgb]{ .647,  .647,  .647}0.0620 & \cellcolor[rgb]{ .647,  .647,  .647}0.0254 & 0.7657 & 0.6270 & \cellcolor[rgb]{ .647,  .647,  .647}0.6641 \\
    {Homophily} & $\text{H}_\text{agg}$ & \cellcolor[rgb]{ .647,  .647,  .647}0.8032 & \cellcolor[rgb]{ .647,  .647,  .647}0.7768 & \cellcolor[rgb]{ .647,  .647,  .647}0.6940 & \cellcolor[rgb]{ .647,  .647,  .647}0.6822 & 0.61  & \cellcolor[rgb]{ .647,  .647,  .647}0.3566 & 0.9904 & 0.9826 & \cellcolor[rgb]{ .647,  .647,  .647}0.9432 \\
    {  Metrics} & $\text{H}_\text{GE}$ & 0.31  & 0.34  & 0.35  & 0.16  & \cellcolor[rgb]{ .647,  .647,  .647}0.0152 & \cellcolor[rgb]{ .647,  .647,  .647}0.0157 & \cellcolor[rgb]{ .647,  .647,  .647}0.1700 & \cellcolor[rgb]{ .647,  .647,  .647}0.1900 & 0.27 \\
     & $\text{H}_\text{adj}$ & 0.1889 & 0.0826 & 0.0258 & 0.1272 & \cellcolor[rgb]{ .647,  .647,  .647}0.0663 & \cellcolor[rgb]{ .647,  .647,  .647}0.0196 & 0.8178 & 0.7588 & \cellcolor[rgb]{ .647,  .647,  .647}0.7431 \\
     & $\text{LI}$    & 0.0169 & 0.1311 & 0.1923 & 0.0002 & \cellcolor[rgb]{ .647,  .647,  .647}0.048 & \cellcolor[rgb]{ .647,  .647,  .647}0.0015 & 0.5904 & \cellcolor[rgb]{ .647,  .647,  .647}0.4508 & 0.4093 \\
           \midrule
{Classifier-based } & $\text{KR}_\text{L}$ & 1.39  & 0.00  & 0.00  & \cellcolor[rgb]{ .647,  .647,  .647}\textcolor[rgb]{ 1,  0,  0}{0.7834} & 1.00  & 1.00  & 1.00  & 1.00  & \cellcolor[rgb]{ .647,  .647,  .647}0.8026 \\
    {Performance Metrics} & GNB & 0.00  & 0.00  & 0.00  & 0.00  & 1.00  & 1.00  & 1.00  & 1.00  & \cellcolor[rgb]{ .647,  .647,  .647}1.0000 \\
    \midrule
          & ACC SGC & 70.98 $\pm$ 8.39 & 70.38 $\pm$ 2.85 & 83.28 $\pm$ 5.43 & 25.26 $\pm$ 1.18 & 64.86 $\pm$ 1.81 & 47.62 $\pm$ 1.27 & 85.12 $\pm$ 1.64 & 79.66 $\pm$ 0.75 & 85.5 $\pm$ 0.76 \\
    {SGC v.s. MLP1} & ACC MLP-1 & 93.77 $\pm$ 3.34 & 93.87 $\pm$ 3.33 & 93.77 $\pm$ 3.34 & 34.53 $\pm$ 1.48 & 45.01 $\pm$ 1.58 & 29.17 $\pm$ 1.46 & 74.3 $\pm$ 1.27 & 75.51 $\pm$ 1.35 & 86.23 $\pm$ 0.54 \\
          & \textbf{Diff Acc} & -22.79 & -23.49 & -10.49 & -9.27 & 19.85 & 18.45 & 10.82 & 4.15  & -0.73 \\
    \midrule
    \multicolumn{1}{p{12.57em}|}{ } & $\text{H}_\text{edge}$ & \cellcolor[rgb]{ .647,  .647,  .647}0.5669 & 0.4480 & 0.4106 & 0.3750 & \cellcolor[rgb]{ .647,  .647,  .647}0.2795 & \cellcolor[rgb]{ .647,  .647,  .647}0.2416 & 0.8100 & 0.7362 & 0.8024 \\
     & $\text{H}_\text{node}$ & 0.3855 & 0.1498 & 0.0968 & 0.2210 & \cellcolor[rgb]{ .647,  .647,  .647}0.2470 & \cellcolor[rgb]{ .647,  .647,  .647}0.2156 & 0.8252 & 0.7175 & 0.7924 \\
    {Baseline} & $\text{H}_\text{class}$ & 0.0468 & 0.0941 & 0.0013 & 0.0110 & \cellcolor[rgb]{ .647,  .647,  .647}0.0620 & \cellcolor[rgb]{ .647,  .647,  .647}0.0254 & 0.7657 & 0.6270 & 0.6641 \\
    {Homophily} & $\text{H}_\text{agg}$ & \cellcolor[rgb]{ .647,  .647,  .647}0.8032 & \cellcolor[rgb]{ .647,  .647,  .647}0.7768 & \cellcolor[rgb]{ .647,  .647,  .647}0.6940 & \cellcolor[rgb]{ .647,  .647,  .647}0.6822 & 0.61  & \cellcolor[rgb]{ .647,  .647,  .647}0.3566 & 0.9904 & 0.9826 & 0.9432 \\
    { Metrics} & $\text{H}_\text{GE}$ & 0.31  & 0.34  & 0.35  & 0.16  & \cellcolor[rgb]{ .647,  .647,  .647}0.0152 & \cellcolor[rgb]{ .647,  .647,  .647}0.0157 & \cellcolor[rgb]{ .647,  .647,  .647}0.1700 & \cellcolor[rgb]{ .647,  .647,  .647}0.1900 & \cellcolor[rgb]{ .647,  .647,  .647}0.2700 \\
     & $\text{H}_\text{adj}$ & 0.1889 & 0.0826 & 0.0258 & 0.1272 & \cellcolor[rgb]{ .647,  .647,  .647}0.0663 & \cellcolor[rgb]{ .647,  .647,  .647}0.0196 & 0.8178 & 0.7588 & 0.7431 \\
    & $\text{LI}$    & 0.0169 & 0.1311 & 0.1923 & 0.0002 & \cellcolor[rgb]{ .647,  .647,  .647}0.048 & \cellcolor[rgb]{ .647,  .647,  .647}0.0015 & 0.5904 & \cellcolor[rgb]{ .647,  .647,  .647}0.4508 & \cellcolor[rgb]{ .647,  .647,  .647}0.4093 \\
    \midrule
    {Classifier-based } & $\text{KR}_\text{NL}$ & 0.00  & 0.00  & 0.00  & 0.00  & 1.00  & 1.00  & 1.00  & 1.00  & 1.00 \\
    {Performance Metrics} & GNB & 0.00  & 0.00  & 0.00  & 0.00  & 1.00  & 1.00  & 1.00  & 1.00  & 1.00 \\
    \midrule
          & ACC GCN & 82.46 $\pm$ 3.11 & 75.5 $\pm$ 2.92 & 83.11 $\pm$ 3.2 & 35.51 $\pm$ 0.99 & 64.18 $\pm$ 2.62 & 44.76 $\pm$ 1.39 & 87.78 $\pm$ 0.96 & 81.39 $\pm$ 1.23 & 88.9 $\pm$ 0.32 \\
    {GCN v.s. MLP2} & ACC MLP-2 & 91.30 $\pm$ 0.70 & 93.87 $\pm$ 3.33 & 92.26 $\pm$ 0.71 & 38.58 $\pm$ 0.25 & 46.72 $\pm$ 0.46 & 31.28 $\pm$ 0.27 & 76.44 $\pm$ 0.30 & 76.25 $\pm$ 0.28 & 86.43 $\pm$ 0.13 \\
          & \textbf{Diff Acc} & -8.84 & -18.37 & -9.15 & -3.07 & 17.46 & 13.48 & 11.34 & 5.14  & 2.47 \\
    \bottomrule
    \bottomrule
    \end{tabular}%
    }
  \label{tab:detailed_comparison_small_datasets}%
\end{table}%

\begin{table}[htbp]
  \centering
  \caption{Comparison on large-scale datasets}
  \resizebox{1\hsize}{!}{
    \begin{tabular}{c|c|ccccccc}
    \toprule
    \toprule
          &       & Penn94 &  pokec &  arXiv-year &  snap-patents &  genius &  twitch-gamers & Deezer-Europe \\
    \midrule
          & $\text{H}_\text{edge}$ & 0.4700 & 0.4450 & \cellcolor[rgb]{ .647,  .647,  .647}0.2220 & 0.0730 & \cellcolor[rgb]{ .647,  .647,  .647}0.6180 & \cellcolor[rgb]{ .647,  .647,  .647}0.5450 & \cellcolor[rgb]{ .647,  .647,  .647}0.5250 \\
     & $\text{H}_\text{node}$ & 0.4828 & 0.4283 & \cellcolor[rgb]{ .647,  .647,  .647}0.2893 & 0.2206 & \cellcolor[rgb]{ .647,  .647,  .647}0.5087 & \cellcolor[rgb]{ .647,  .647,  .647}0.5564 & \cellcolor[rgb]{ .647,  .647,  .647}0.5299 \\
    {Baseline} & $\text{H}_\text{class}$ & 0.0460 & 0.0000 & \cellcolor[rgb]{ .647,  .647,  .647}0.2720 & 0.1000 & 0.0800 & 0.0900 & 0.0300 \\
    {Homophily} & $\text{H}_\text{agg}$ & 0.2712 & 0.0807 & 0.7066 & \cellcolor[rgb]{ .647,  .647,  .647}0.6170 & \cellcolor[rgb]{ .647,  .647,  .647}0.7823 & 0.4172 & \cellcolor[rgb]{ .647,  .647,  .647}0.5580 \\
    {  Metrics} & $\text{H}_\text{GE}$ & 0.3734 & \cellcolor[rgb]{ .647,  .647,  .647}0.9222 & 0.8388 & \cellcolor[rgb]{ .647,  .647,  .647}0.6064 & \cellcolor[rgb]{ .647,  .647,  .647}0.6655 & 0.2865 & 0.0378 \\
    & $\text{H}_\text{adj}$ & 0.0366 & -0.1132 & \cellcolor[rgb]{ .647,  .647,  .647}0.0729 & 0.0907 & 0.1432 & 0.1010 & 0.1586 \\
    & $\text{LI}$    & 0.0851 & 0.0172 & \cellcolor[rgb]{ .647,  .647,  .647}0.0407 & 0.0243 & 0.0025 & 0.0058 & 0.0007 \\
    \midrule
    {Classifier-based } & $\text{KR}_\text{L}$ & 0.00  & 0.03  & 0.98  & \textcolor[rgb]{ 1,  0,  0}{0.19} & 0.00  & \textcolor[rgb]{ 1,  0,  0}{0.25} & 0.00 \\
    {Performance Metrics} & GBN & 0.00  & 0.00  & 1.00  & \cellcolor[rgb]{ .647,  .647,  .647}\textcolor[rgb]{ 1,  0,  0}{1.00} & 0.00  & \cellcolor[rgb]{ .647,  .647,  .647}\textcolor[rgb]{ 1,  0,  0}{1.00} & 0.00 \\
    \midrule
          & ACC SGC & 67.06 $\pm$ 0.19 & 52.88 $\pm$ 0.64 & 35.58 $\pm$ 0.22  & 29.65 $\pm$ 0.04 & 82.31 $\pm$ 0.45 & 57.9 $\pm$ 0.18 & 61.63 $\pm$ 0.25 \\
    {SGC vs MLP1} & ACC MLP-1 & 73.72 $\pm$ 0.5 & 59.89 $\pm$ 0.11 & 34.11 $\pm$ 0.17 & 30.59 $\pm$ 0.02 & 86.48 $\pm$ 0.11 & 59.45 $\pm$ 0.16 & 63.14 $\pm$ 0.41 \\
          & \textbf{Diff Acc} & -6.66 & -7.01 & 1.47  & -0.94 & -4.17 & -1.55 & -1.51 \\
    \midrule
    { } & $\text{H}_\text{edge}$ & \cellcolor[rgb]{ .647,  .647,  .647}0.4700 & \cellcolor[rgb]{ .647,  .647,  .647}0.4450 & \cellcolor[rgb]{ .647,  .647,  .647}0.2220 & \cellcolor[rgb]{ .647,  .647,  .647}0.0730 & \cellcolor[rgb]{ .647,  .647,  .647}0.6180 & 0.5450 & \cellcolor[rgb]{ .647,  .647,  .647}0.5250 \\
     & $\text{H}_\text{node}$ & \cellcolor[rgb]{ .647,  .647,  .647}0.4828 & \cellcolor[rgb]{ .647,  .647,  .647}0.4283 & \cellcolor[rgb]{ .647,  .647,  .647}0.2893 & \cellcolor[rgb]{ .647,  .647,  .647}0.2206 & \cellcolor[rgb]{ .647,  .647,  .647}0.5087 & 0.5564 & \cellcolor[rgb]{ .647,  .647,  .647}0.5299 \\
    {Baseline} & $\text{H}_\text{class}$ & \cellcolor[rgb]{ .647,  .647,  .647}0.0460 & \cellcolor[rgb]{ .647,  .647,  .647}0.0000 & \cellcolor[rgb]{ .647,  .647,  .647}0.2720 & \cellcolor[rgb]{ .647,  .647,  .647}0.1000 & 0.0800 & \cellcolor[rgb]{ .647,  .647,  .647}0.0900 & 0.0300 \\
    {Homophily} & $\text{H}_\text{agg}$ & \cellcolor[rgb]{ .647,  .647,  .647}0.2712 & \cellcolor[rgb]{ .647,  .647,  .647}0.0807 & 0.7066 & 0.6170 & \cellcolor[rgb]{ .647,  .647,  .647}0.7823 & \cellcolor[rgb]{ .647,  .647,  .647}0.4172 & \cellcolor[rgb]{ .647,  .647,  .647}0.5580 \\
    { Metrics} & $\text{H}_\text{GE}$ & \cellcolor[rgb]{ .647,  .647,  .647}0.3734 & 0.9222 & 0.8388 & 0.6064 & \cellcolor[rgb]{ .647,  .647,  .647}0.6655 & \cellcolor[rgb]{ .647,  .647,  .647}0.2865 & 0.0378 \\
    & $\text{H}_\text{adj}$ & \cellcolor[rgb]{ .647,  .647,  .647}0.0366 & \cellcolor[rgb]{ .647,  .647,  .647}-0.1132 & \cellcolor[rgb]{ .647,  .647,  .647}0.0729 & \cellcolor[rgb]{ .647,  .647,  .647}0.0907 & 0.1432 & \cellcolor[rgb]{ .647,  .647,  .647}0.1010 & 0.1586 \\
    & $\text{LI}$    & \cellcolor[rgb]{ .647,  .647,  .647}0.0851 & \cellcolor[rgb]{ .647,  .647,  .647}0.0172 & \cellcolor[rgb]{ .647,  .647,  .647}0.0407 & \cellcolor[rgb]{ .647,  .647,  .647}0.0243 & 0.0025 & \cellcolor[rgb]{ .647,  .647,  .647}0.0058 & 0.0007 \\
    \midrule
    {Classifier-based } & $\text{KR}_\text{NL}$ & \cellcolor[rgb]{ .647,  .647,  .647}\textcolor[rgb]{ 1,  0,  0}{0.00} & \textcolor[rgb]{ 1,  0,  0}{0.57} & 1.00  & \cellcolor[rgb]{ .647,  .647,  .647}\textcolor[rgb]{ 1,  0,  0}{0.4083} & 0.00  & 1.00  & 0.00 \\
    {Performance Metrics} & GNB & \cellcolor[rgb]{ .647,  .647,  .647}\textcolor[rgb]{ 1,  0,  0}{0.00} & \cellcolor[rgb]{ .647,  .647,  .647}\textcolor[rgb]{ 1,  0,  0}{0.00} & 1.00  & 1.00  & 0.00  & 1.00  & 0.00 \\
    \midrule
          & ACC GCN & 82.08 $\pm$ 0.31 & 70.3 $\pm$ 0.1 & 40 $\pm$ 0.26 & 35.8 $\pm$ 0.05 & 83.26 $\pm$ 0.14 & 62.33 $\pm$ 0.23 & 60.16 $\pm$ 0.51 \\
    {GCN vs MLP2} & ACC MLP-2 & 74.68 $\pm$ 0.28 & 62.13 $\pm$ 0.1 & 36.36 $\pm$ 0.23 & 31.43 $\pm$ 0.04 & 86.62 $\pm$ 0.08 & 60.9 $\pm$ 0.11 & 64.25 $\pm$ 0.41 \\
          & \textbf{Diff Acc} & 7.40  & 8.17  & 3.64  & 4.37  & -3.36 & 1.43  & -4.09 \\
    \bottomrule
    \bottomrule
    \end{tabular}%
    }
  \label{tab:detaied_comparison_large_scale_datasets}%
\end{table}%

\newpage

\subsection{Statistics and Comparisons}
\label{appendix:statistics_comparisons}
\paragraph{Discrepancy Between Linear and Non-linear Models} From the experimental results on large-scale datasets reported in Table \ref{tab:detaied_comparison_large_scale_datasets}, we observe that, for linear and non-linear G-aware models, there exists inconsistency between their comparison with their coupled G-agnostics models. For example, on $\textit{Penn94, pokec, snap-patents}$ and $\textit{twitch-gamers}$, SGC-1 underperforms MLP-1 but GCN outperforms MLP-2. In fact, $\text{PubMed}$ in Table \ref{tab:detailed_comparison_small_datasets} also belongs to this family of datasets. We do not have a proved theory to explain this phenomenon for now. But there is obviously a synergy between homophily/heterophily and non-linearity that cause this discrepancy together. And we think, on this special subset of heterophilic graphs, we should develop theoretical analysis to discuss the interplay between graph structure and feature non-linearity, and how they affect node distinguishability together.

The current homophily values (including the proposed metrics) are not able to explain the phenomenon associated with this group of datasets. We keep it as an open question and encourage people from the GNN community to study it in the future.

\begin{table}[htbp]
  \centering
  \caption{Statistics on small-scale datasets}
    \begin{tabular}{c|cc}
    \midrule
    \midrule
          & Total Error & Error Rate \\
    \midrule
    $\text{H}_\text{edge}$ & 7     & 0.39 \\
    $\text{H}_\text{node}$ & 5     & 0.28 \\
    $\text{H}_\text{class}$ & 5     & 0.28 \\
    $\text{H}_\text{agg}$ & 11    & 0.61 \\
    $\text{H}_\text{GE}$ & 9     & 0.50 \\
    $\text{H}_\text{adj}$ & 5     & 0.28 \\
    $\text{LI}$    & 7     & 0.39 \\
    \midrule
    $\text{KR}$ (NT0.5)  & 2     & 0.11 \\
    $\text{KR}$ (SST0.05) & 1     & 0.06 \\
    GNB (NT0.5) & 1     & 0.06 \\
    GNB (SST0.05) & 1     & 0.06 \\
    \bottomrule
    \bottomrule
    \end{tabular}%
  \label{tab:statistics_small_scale_datasets}%
\end{table}%

\begin{table}[htbp]
  \centering
  \caption{Statistics on large-scale datasets}
    \begin{tabular}{c|cc}
    \midrule
    \midrule
          & Total Error & Error Rate \\
    \midrule
    $\text{H}_\text{edge}$ & 9     & 0.64 \\
    $\text{H}_\text{node}$ & 9     & 0.64 \\
    $\text{H}_\text{class}$ & 6     & 0.43 \\
    $\text{H}_\text{agg}$ & 8     & 0.57 \\
    $\text{H}_\text{GE}$ & 6     & 0.43 \\
    $\text{H}_\text{adj}$ & 6     & 0.43 \\
    $\text{LI}$    & 6     & 0.43 \\
    \midrule
    $\text{KR}$ (NT0.5)  & 2     & 0.14 \\
    $\text{KR}$ (SST0.05) & 5     & 0.36 \\
    GNB (NT0.5) & 4     & 0.29 \\
    GNB (SST0.05) & 4     & 0.29 \\
    \bottomrule
    \bottomrule
    \end{tabular}%
  \label{tab:statistics_large_scale_datasets}%
\end{table}%

\begin{table}[htbp]
  \centering
  \caption{Overall statistics on small- and large-scale datasets}
    \begin{tabular}{c|cc}
    \midrule
    \midrule
          & Total Error & Error Rate \\
    \midrule
    $\text{H}_\text{edge}$ & 16    & 0.50 \\
    $\text{H}_\text{node}$ & 14    & 0.44 \\
    $\text{H}_\text{class}$ & 11    & 0.34 \\
    $\text{H}_\text{agg}$ & 19    & 0.59 \\
    $\text{H}_\text{GE}$ & 15    & 0.47 \\
    $\text{H}_\text{adj}$ & 11    & 0.34 \\
    $\text{LI}$    & 13    & 0.41 \\
    \midrule
    $\text{KR} $ (NT0.5)  & 4     & 0.13 \\
    $\text{KR} $ (SST0.05) & 6     & 0.19 \\
    GNB (NT0.5) & 5     & 0.16 \\
    GNB (SST0.05) & 5     & 0.16 \\
    \bottomrule
    \bottomrule
    \end{tabular}%
  \label{tab:statistics_overall}%
\end{table}%

\begin{table}[htbp]
  \centering
  \caption{Total running time (seconds/100 samples) of $\text{KR}_\text{L}$, $\text{KR}_\text{NL}$ and GNB}
    \begin{tabular}{c|ccc}
    \toprule
    \toprule
          & $\text{KR}_\text{L} $ & $\text{KR}_\text{NL}$ & GNB \\
    \midrule
    Cornell & 0.58  & 0.67  & 1.39 \\
    Wisconsin & 0.78  & 0.87  & 1.72 \\
    Texas & 0.59  & 0.67  & 1.41 \\
    Film  & 5.29  & 5.41  & 2.72 \\
    Chameleon & 3.97  & 3.95  & 3.81 \\
    Squirrel & 5.39  & 5.36  & 4.15 \\
    Cora  & 3.94  & 4.10  & 3.08 \\
    CiteSeer & 4.85  & 5.05  & 6.55 \\
    PubMed & 9.35  & 9.41  & 5.27 \\
    Penn94 & 18.57 & 18.68 & 12.43 \\
     pokec & 84.47 & 86.08 & 50.03 \\
     arXiv-year & 7.77  & 7.82  & 4.56 \\
     snap-patents & 304.06 & 296.21 & 163.84 \\
     genius & 8.20  & 8.12  & 5.30 \\
     twitch-gamers & 9.34  & 9.24  & 4.17 \\
    Deezer-Europe & 37.41 & 39.49 & 59.84 \\
    \bottomrule
    \bottomrule
    \end{tabular}%
  \label{tab:running_time}%
\end{table}%

\newpage
\subsection{Results for Symmetric Renormalized Affinity Matrix}
\label{appendix:results_symmetric_renormalized_affinity_matrix}

To evaluate if the benefits of classifier-based performance metrics can be maintained for different aggregation operators, we replace the random walk renormalized affinity matrix with synmmetric renormalized affinity matrix in SGC-1, GCN, $\text{KR}_\text{L},\text{KR}_\text{NL}$ and GNB and report the results and comparisons as belows.

It is observed that the superiority holds on both small- (table \ref{tab:detailed_comparison_symmetric_affinity_small_scale}, \ref{tab:statistics_symmetric_affinity_small_scale}) and large-scale datasets (table \ref{tab:detailed_comparison_symmetric_affinity_large_scale}, \ref{tab:statistics_symmetric_affinity_large_scale}), reducing the overall error rate from at least 0.31 to 0.13 (table \ref{tab:statistics_symmetric_affinity_overall}).

\begin{table}[htbp]
  \centering
  \caption{Results for symmetric renormalized affinity matrix on small-scale datasets}
   \resizebox{1\hsize}{!}{
    \begin{tabular}{c|c|ccccccccc}
    \toprule
    \toprule
          &       & Cornell & Wisconsin & Texas & Film  & Chameleon & Squirrel & Cora  & CiteSeer & PubMed \\
    \midrule
          & $\text{H}_\text{edge}$ & \cellcolor[rgb]{ .647,  .647,  .647}0.5669 & 0.4480 & 0.4106 & 0.3750 & \cellcolor[rgb]{ .647,  .647,  .647}0.2795 & \cellcolor[rgb]{ .647,  .647,  .647}0.2416 & 0.8100 & 0.7362 & 0.8024 \\
     & $\text{H}_\text{node}$ & 0.3855 & 0.1498 & 0.0968 & 0.2210 & \cellcolor[rgb]{ .647,  .647,  .647}0.2470 & \cellcolor[rgb]{ .647,  .647,  .647}0.2156 & 0.8252 & 0.7175 & 0.7924 \\
    {Baseline} & $\text{H}_\text{class}$ & 0.0468 & 0.0941 & 0.0013 & 0.0110 & \cellcolor[rgb]{ .647,  .647,  .647}0.0620 & \cellcolor[rgb]{ .647,  .647,  .647}0.0254 & 0.7657 & 0.6270 & 0.6641 \\
    {Homophily} & $\text{H}_\text{agg}$ & \cellcolor[rgb]{ .647,  .647,  .647}0.8032 & \cellcolor[rgb]{ .647,  .647,  .647}0.7768 & \cellcolor[rgb]{ .647,  .647,  .647}0.6940 & \cellcolor[rgb]{ .647,  .647,  .647}0.6822 & 0.61  & \cellcolor[rgb]{ .647,  .647,  .647}0.3566 & 0.9904 & 0.9826 & 0.9432 \\
    {  Metrics} & $\text{H}_\text{GE}$ & 0.31  & 0.34  & 0.35  & 0.16  & \cellcolor[rgb]{ .647,  .647,  .647}0.0152 & \cellcolor[rgb]{ .647,  .647,  .647}0.0157 & \cellcolor[rgb]{ .647,  .647,  .647}0.1700 & \cellcolor[rgb]{ .647,  .647,  .647}0.1900 & \cellcolor[rgb]{ .647,  .647,  .647}0.2700 \\
    & $\text{H}_\text{adj}$  & 0.1889 & 0.0826 & 0.0258 & 0.1272 & \cellcolor[rgb]{ .647,  .647,  .647}0.0663 & \cellcolor[rgb]{ .647,  .647,  .647}0.0196 & 0.8178 & 0.7588 & 0.7431 \\
    & $\text{LI}$    & 0.0169 & 0.1311 & 0.1923 & 0.0002 & \cellcolor[rgb]{ .647,  .647,  .647}0.048 & \cellcolor[rgb]{ .647,  .647,  .647}0.0015 & 0.5904 & \cellcolor[rgb]{ .647,  .647,  .647}0.4508 & \cellcolor[rgb]{ .647,  .647,  .647}0.4093 \\
    \midrule
    {Classifier-based } & $\text{KR}_\text{L}$ & 0.00  & 0.00  & 0.00  & \cellcolor[rgb]{ .647,  .647,  .647}\textcolor[rgb]{ 1,  0,  0}{0.9304} & 1.00  & 1.00  & 1.00  & 1.00  & \cellcolor[rgb]{ .647,  .647,  .647}\textcolor[rgb]{ 1,  0,  0}{0.0003} \\
    {Performance Metrics} & GNB & 0.00  & 0.00  & 0.00  & 0.00  & 1.00  & 1.00  & 1.00  & 1.00  & 1.00 \\
    \midrule
          & ACC SGC & 51.64 $\pm$ 12.27 & 39.63 $\pm$ 5.39 & 30.82 $\pm$ 4.96 & 27.02 $\pm$ 1 & 63.26 $\pm$ 1.98 & 46.03 $\pm$ 1.74 & 84.38 $\pm$ 1.5 & 79.51 $\pm$ 1.04 & 87.24 $\pm$ 0.44 \\
    {SGC vs MLP1} & ACC MLP-1 & 93.77 $\pm$ 3.34 & 93.87 $\pm$ 3.33 & 93.77 $\pm$ 3.34 & 34.53 $\pm$ 1.48 & 45.01 $\pm$ 1.58 & 29.17 $\pm$ 1.46 & 74.3 $\pm$ 1.27 & 75.51 $\pm$ 1.35 & 86.23 $\pm$ 0.54 \\
          & \textbf{Diff Acc} & -42.13 & -54.24 & -62.95 & -7.51 & 18.25 & 16.86 & 10.08 & 4.00  & 1.01 \\
    \midrule
    { } & $\text{H}_\text{edge}$ & \cellcolor[rgb]{ .647,  .647,  .647}0.5669 & 0.4480 & 0.4106 & 0.3750 & \cellcolor[rgb]{ .647,  .647,  .647}0.2795 & \cellcolor[rgb]{ .647,  .647,  .647}0.2416 & 0.8100 & 0.7362 & 0.8024 \\
     & $\text{H}_\text{node}$ & 0.3855 & 0.1498 & 0.0968 & 0.2210 & \cellcolor[rgb]{ .647,  .647,  .647}0.2470 & \cellcolor[rgb]{ .647,  .647,  .647}0.2156 & 0.8252 & 0.7175 & 0.7924 \\
    {Baseline} & $\text{H}_\text{class}$ & 0.0468 & 0.0941 & 0.0013 & 0.0110 & \cellcolor[rgb]{ .647,  .647,  .647}0.0620 & \cellcolor[rgb]{ .647,  .647,  .647}0.0254 & 0.7657 & 0.6270 & 0.6641 \\
    {Homophily} & $\text{H}_\text{agg}$ & \cellcolor[rgb]{ .647,  .647,  .647}0.8032 & \cellcolor[rgb]{ .647,  .647,  .647}0.7768 & \cellcolor[rgb]{ .647,  .647,  .647}0.6940 & \cellcolor[rgb]{ .647,  .647,  .647}0.6822 & 0.61  & \cellcolor[rgb]{ .647,  .647,  .647}0.3566 & 0.9904 & 0.9826 & 0.9432 \\
    { Metrics}& $\text{H}_\text{GE}$ & 0.31  & 0.34  & 0.35  & 0.16  & \cellcolor[rgb]{ .647,  .647,  .647}0.0152 & \cellcolor[rgb]{ .647,  .647,  .647}0.0157 & \cellcolor[rgb]{ .647,  .647,  .647}0.1700 & \cellcolor[rgb]{ .647,  .647,  .647}0.1900 & \cellcolor[rgb]{ .647,  .647,  .647}0.2700 \\
     & $\text{H}_\text{adj}$  & 0.1889 & 0.0826 & 0.0258 & 0.1272 & \cellcolor[rgb]{ .647,  .647,  .647}0.0663 & \cellcolor[rgb]{ .647,  .647,  .647}0.0196 & 0.8178 & 0.7588 & 0.7431 \\
    & $\text{LI}$    & 0.0169 & 0.1311 & 0.1923 & 0.0002 & \cellcolor[rgb]{ .647,  .647,  .647}0.048 & \cellcolor[rgb]{ .647,  .647,  .647}0.0015 & 0.5904 & \cellcolor[rgb]{ .647,  .647,  .647}0.4508 & \cellcolor[rgb]{ .647,  .647,  .647}0.4093 \\
    \midrule
    {Classifier-based } & $\text{KR}_\text{NL}$ & 0.00  & 0.00  & 0.00  & 0.00  & 1.00  & 1.00  & 1.00  & 1.00  & 0.98 \\
    {Performance Metrics} & GNB & 0.00  & 0.00  & 0.00  & 0.00  & 1.00  & 1.00  & 1.00  & 1.00  & 1.00 \\
    \midrule
          & ACC GCN & 82.62 $\pm$ 3.04 & 70.38 $\pm$ 3.16 & 82.46 $\pm$ 2.94 & 35.79 $\pm$ 1.09 & 68.95 $\pm$ 1.09 & 52.98 $\pm$ 0.85 & 87.87 $\pm$ 0.99 & 81.79 $\pm$ 1.09 & 89.47 $\pm$ 0.27 \\
    {GCN vs MLP2} & ACC MLP-2 & 91.30 $\pm$ 0.70 & 93.87 $\pm$ 3.33 & 92.26 $\pm$ 0.71 & 38.58 $\pm$ 0.25 & 46.72 $\pm$ 0.46 & 31.28 $\pm$ 0.27 & 76.44 $\pm$ 0.30 & 76.25 $\pm$ 0.28 & 86.43 $\pm$ 0.13 \\
          & \textbf{Diff Acc} & -8.68 & -23.49 & -9.80 & -2.79 & 22.23 & 21.70 & 11.43 & 5.54  & 3.04 \\
    \bottomrule
    \bottomrule
    \end{tabular}%
    }
  \label{tab:detailed_comparison_symmetric_affinity_small_scale}%
\end{table}%

\begin{table}[htbp]
  \centering
  \caption{Results for symmetric renormalized affinity matrix on large-scale datasets}
    \resizebox{1\hsize}{!}{
    \begin{tabular}{c|c|ccccccc}
    \toprule
    \toprule
          &       & Penn94 &  pokec &  arXiv-year &  snap-patents &  genius &  twitch-gamers & Deezer-Europe \\
    \midrule
          & $\text{H}_\text{edge}$ & 0.4700 & 0.4450 & \cellcolor[rgb]{ .647,  .647,  .647}0.2220 & 0.0730 & \cellcolor[rgb]{ .647,  .647,  .647}0.6180 & \cellcolor[rgb]{ .647,  .647,  .647}0.5450 & \cellcolor[rgb]{ .647,  .647,  .647}0.5250 \\
     & $\text{H}_\text{node}$ & 0.4828 & 0.4283 & \cellcolor[rgb]{ .647,  .647,  .647}0.2893 & 0.2206 & \cellcolor[rgb]{ .647,  .647,  .647}0.5087 & \cellcolor[rgb]{ .647,  .647,  .647}0.5564 & \cellcolor[rgb]{ .647,  .647,  .647}0.5299 \\
    {Baseline} & $\text{H}_\text{class}$ & 0.0460 & 0.0000 & \cellcolor[rgb]{ .647,  .647,  .647}0.2720 & 0.1000 & 0.0800 & 0.0900 & 0.0300 \\
    {Homophily} & $\text{H}_\text{agg}$ & 0.2712 & 0.0807 & 0.7066 & \cellcolor[rgb]{ .647,  .647,  .647}0.6170 & \cellcolor[rgb]{ .647,  .647,  .647}0.7823 & 0.4172 & \cellcolor[rgb]{ .647,  .647,  .647}0.5580 \\
    {  Metrics}& $\text{H}_\text{GE}$ & 0.3734 & \cellcolor[rgb]{ .647,  .647,  .647}0.9222 & 0.8388 & \cellcolor[rgb]{ .647,  .647,  .647}0.6064 & \cellcolor[rgb]{ .647,  .647,  .647}0.6655 & 0.2865 & 0.0378 \\
    & $\text{H}_\text{adj}$ & 0.0366 & -0.1132 & \cellcolor[rgb]{ .647,  .647,  .647}0.0729 & 0.0907 & 0.1432 & 0.1010 & 0.1586 \\
    & $\text{LI}$    & 0.0851 & 0.0172 & \cellcolor[rgb]{ .647,  .647,  .647}0.0407 & 0.0243 & 0.0025 & 0.0058 & 0.0007 \\
    \midrule
    {Classifier-based } & $\text{KR}_\text{L}$ & 0.00  & 0.02  & \cellcolor[rgb]{ .647,  .647,  .647}\textcolor[rgb]{ 1,  0,  0}{0.32} & \textcolor[rgb]{ 1,  0,  0}{0.46} & 0.00  & 0.01  & 0.00 \\
    {Performance Metrics} & GNB & 0.00  & 0.03  & 1.00  & \cellcolor[rgb]{ .647,  .647,  .647}\textcolor[rgb]{ 1,  0,  0}{1.00} & 0.00  & \cellcolor[rgb]{ .647,  .647,  .647}\textcolor[rgb]{ 1,  0,  0}{0.97} & 0.00 \\
          & ACC SGC & 64.63 $\pm$ 0.15 & 51.97 $\pm$ 0.38 & 35.24 $\pm$ 0.14 & 30.32 $\pm$ 0.05 & 81.66 $\pm$ 0.58 & 58.77 $\pm$ 0.18 & 60.2 $\pm$ 0.47 \\
    {SGC vs MLP1} & ACC MLP-1 & 73.72 $\pm$ 0.5 & 59.89 $\pm$ 0.11 & 34.11 $\pm$ 0.17 & 30.59 $\pm$ 0.02 & 86.48 $\pm$ 0.11 & 59.45 $\pm$ 0.16 & 63.14 $\pm$ 0.41 \\
          & \textbf{Diff Acc} & -9.09 & -7.92 & 1.13  & -0.27 & -4.82 & -0.68 & -2.94 \\
    \midrule
    { } & $\text{H}_\text{edge}$ & \cellcolor[rgb]{ .647,  .647,  .647}0.4700 & \cellcolor[rgb]{ .647,  .647,  .647}0.4450 & \cellcolor[rgb]{ .647,  .647,  .647}0.2220 & \cellcolor[rgb]{ .647,  .647,  .647}0.0730 & \cellcolor[rgb]{ .647,  .647,  .647}0.6180 & 0.5450 & \cellcolor[rgb]{ .647,  .647,  .647}0.5250 \\
     & $\text{H}_\text{node}$ & \cellcolor[rgb]{ .647,  .647,  .647}0.4828 & \cellcolor[rgb]{ .647,  .647,  .647}0.4283 & \cellcolor[rgb]{ .647,  .647,  .647}0.2893 & \cellcolor[rgb]{ .647,  .647,  .647}0.2206 & \cellcolor[rgb]{ .647,  .647,  .647}0.5087 & 0.5564 & \cellcolor[rgb]{ .647,  .647,  .647}0.5299 \\
    {Baseline} & $\text{H}_\text{class}$ & \cellcolor[rgb]{ .647,  .647,  .647}0.0460 & \cellcolor[rgb]{ .647,  .647,  .647}0.0000 & \cellcolor[rgb]{ .647,  .647,  .647}0.2720 & \cellcolor[rgb]{ .647,  .647,  .647}0.1000 & 0.0800 & \cellcolor[rgb]{ .647,  .647,  .647}0.0900 & 0.0300 \\
    {Homophily} & $\text{H}_\text{agg}$ & \cellcolor[rgb]{ .647,  .647,  .647}0.2712 & \cellcolor[rgb]{ .647,  .647,  .647}0.0807 & 0.7066 & 0.6170 & \cellcolor[rgb]{ .647,  .647,  .647}0.7823 & \cellcolor[rgb]{ .647,  .647,  .647}0.4172 & \cellcolor[rgb]{ .647,  .647,  .647}0.5580 \\
    { Metrics} & $\text{H}_\text{GE}$ & \cellcolor[rgb]{ .647,  .647,  .647}0.3734 & 0.9222 & 0.8388 & 0.6064 & \cellcolor[rgb]{ .647,  .647,  .647}0.6655 & \cellcolor[rgb]{ .647,  .647,  .647}0.2865 & 0.0378 \\
    & $\text{H}_\text{adj}$ & \cellcolor[rgb]{ .647,  .647,  .647}0.0366 & \cellcolor[rgb]{ .647,  .647,  .647}-0.1132 & \cellcolor[rgb]{ .647,  .647,  .647}0.0729 & \cellcolor[rgb]{ .647,  .647,  .647}0.0907 & 0.1432 & \cellcolor[rgb]{ .647,  .647,  .647}0.1010 & 0.1586 \\
    & $\text{LI}$    & \cellcolor[rgb]{ .647,  .647,  .647}0.0851 & \cellcolor[rgb]{ .647,  .647,  .647}0.0172 & \cellcolor[rgb]{ .647,  .647,  .647}0.0407 & \cellcolor[rgb]{ .647,  .647,  .647}0.0243 & 0.0025 & \cellcolor[rgb]{ .647,  .647,  .647}0.0058 & 0.0007 \\
    \midrule
    {Classifier-based } & $\text{KR}_\text{NL}$ & \cellcolor[rgb]{ .647,  .647,  .647}\textcolor[rgb]{ 1,  0,  0}{0.00} & \cellcolor[rgb]{ .647,  .647,  .647}\textcolor[rgb]{ 1,  0,  0}{0.32} & 0.99  & 0.99  & 0.00  & 1.00  & 0.00 \\
    {Performance Metrics} & GNB & \cellcolor[rgb]{ .647,  .647,  .647}\textcolor[rgb]{ 1,  0,  0}{0.00} & \cellcolor[rgb]{ .647,  .647,  .647}\textcolor[rgb]{ 1,  0,  0}{0.29} & 1.00  & 1.00  & 0.00  & 0.97  & 0.00 \\
    \midrule
          & ACC GCN & 81.45 $\pm$ 0.29 & 69.55 $\pm$ 0.1 & 40.02 $\pm$ 0.19 & 35.4 $\pm$ 0.04 & 83.02 $\pm$ 0.14 & 62.59 $\pm$ 0.14 & 62.32 $\pm$ 0.44 \\
    {GCN vs MLP2} & ACC MLP-2 & 74.68 $\pm$ 0.28 & 62.13 $\pm$ 0.1 & 36.36 $\pm$ 0.23 & 31.43 $\pm$ 0.04 & 86.62 $\pm$ 0.08 & 60.9 $\pm$ 0.11 & 64.25 $\pm$ 0.41 \\
          & \textbf{Diff Acc} & 6.77  & 7.42  & 3.66  & 3.97  & -3.60 & 1.69  & -1.93 \\
    \bottomrule
    \bottomrule
    \end{tabular}%
    }
  \label{tab:detailed_comparison_symmetric_affinity_large_scale}%
\end{table}%

\begin{table}[htbp]
  \centering
  \caption{Statistics for symmetric renormalized affinity matrix on small-scale datasets}
    \begin{tabular}{c|cc}
    \midrule
    \midrule
          & Total Error & Error Rate \\
    \midrule
    $\text{H}_\text{edge}$ & 6     & 0.33 \\
    $\text{H}_\text{node}$ & 4     & 0.22 \\
    $\text{H}_\text{class}$ & 4     & 0.22 \\
    $\text{H}_\text{agg}$ & 10    & 0.56 \\
    $\text{H}_\text{GE}$ & 10    & 0.56 \\
    $\text{H}_\text{adj}$ & 4     & 0.22 \\
    $\text{LI}$ & 8     & 0.44 \\
    \midrule
    $\text{KR} $ (NT0.5)  & 2     & 0.11 \\
    $\text{KR} $ (SST0.05) & 2     & 0.11 \\
    GNB (NT0.5) & 0     & 0.00 \\
    GNB (SST0.05) & 0     & 0.00 \\
    \bottomrule
    \bottomrule
    \end{tabular}%
  \label{tab:statistics_symmetric_affinity_small_scale}%
\end{table}%

\begin{table}[htbp]
  \centering
  \caption{Statistics for symmetric renormalized affinity matrix on large-scale datasets}
    \begin{tabular}{c|cc}
    \midrule
    \midrule
          & Total Error & Error Rate \\
    \midrule
    $\text{H}_\text{edge}$ & 10    & 0.71 \\
    $\text{H}_\text{node}$ & 10    & 0.71 \\
    $\text{H}_\text{class}$ & 6     & 0.43 \\
    $\text{H}_\text{agg}$ & 8     & 0.57 \\
    $\text{H}_\text{GE}$ & 6     & 0.43 \\
    $\text{H}_\text{adj}$ & 6     & 0.43 \\
    $\text{LI}$    & 6     & 0.43 \\
    \midrule
    $\text{KR} $ (NT0.5)  & 3     & 0.21 \\
    $\text{KR} $ (SST0.05) & 4     & 0.29 \\
    GNB (NT0.5) & 4     & 0.29 \\
    GNB (SST0.05) & 4     & 0.29 \\
    \bottomrule
    \bottomrule
    \end{tabular}%
  \label{tab:statistics_symmetric_affinity_large_scale}%
\end{table}%

\begin{table}[htbp]
  \centering
  \caption{Overall statistics for symmetric renormalized affinity matrix on small- and large-scale datasets}
    \begin{tabular}{c|cc}
    \midrule
    \midrule
          & Total Error & Error Rate \\
    \midrule
    $\text{H}_\text{edge}$ & 16    & 0.50 \\
    $\text{H}_\text{node}$ & 14    & 0.44 \\
    $\text{H}_\text{class}$ & 10    & 0.31 \\
    $\text{H}_\text{agg}$ & 18    & 0.56 \\
    $\text{H}_\text{GE}$ & 16    & 0.50 \\
    $\text{H}_\text{adj}$ & 10    & 0.31 \\
    $\text{LI}$    & 14    & 0.44 \\
    \midrule
    $\text{KR}_\text{NNGP}$ (NT0.5)  & 5     & 0.16 \\
    $\text{KR}_\text{NNGP}$ (SST0.05) & 6     & 0.19 \\
    GNB (NT0.5) & 4     & 0.13 \\
    GNB (SST0.05) & 4     & 0.13 \\
    \bottomrule
    \bottomrule
    \end{tabular}%
  \label{tab:statistics_symmetric_affinity_overall}%
\end{table}%

\newpage
\subsection{Experiments on Synthetic Graphs}
\label{appendix:synthetic_graphs}

To comprehensively investigate and corroborate the correlation between CPMs and the performance of G-aware models versus their corresponding G-agnostic models across the entire spectrum of homophily levels, we conduct experiments with the synthetic graphs. The data generation process is similar to \cite{luan2022revisiting}.

\paragraph{Data Generation \& Experimental Setup}
We generated a total of $280$ graphs with 28 different levels of edge homophily, ranging from 0.005 to 0.95, and generated 10 graphs for each homophily level. Each graph consisted of 5 classes, with 400 nodes in each class. For nodes in each class, we randomly generated 4000 intra-class edges and [$\frac{4000}{H_\text{edge}(\mathcal{G})} -4000$] inter-class edges, and assigned features to the nodes using the \textit{Cora, CiteSeer, PubMed, Chameleon, Squirrel, Film} datasets. We then randomly split the nodes into train/validation/test sets in a 60\%/20\%/20\% ratio. We trained GCN, SGC-1, MLP-2, and MLP-1 models on the synthetic graphs with fine-tuned hyperparameters as \cite{luan2022revisiting}. For each edge homophily level $H_\text{edge}(\mathcal{G})$, we computed the average test accuracy of the 4 models, as well as  $\text{KR}_\text{L},\text{KR}_\text{NL}$ and other homophily metrics. The comparisons of $\text{KR}_\text{L}$, $\text{KR}_\text{NL}$ and the performance of GCN vs. MLP-2, SGC-1 vs. MLP-1 according to edge homophily were shown in Figure \ref{fig:synthetic_graphs_cpms}.

It can be observed in Figure \ref{fig:synthetic_graphs_cpms} that the points where $\text{KR}_\text{L}$ intersects NT0.5 or SST0.05 (green) and the intersections of SGC-1 and MLP-1 performance (red) are perfectly matched and the curve of  $\text{KR}_\text{L}$ share the similar U-shape as SGC-1, so do $\text{KR}_\text{NL}$ curve (blue) and GCN and MLP-2 performance curves (black) \footnote{We only draw the vertical dot lines for the intersection of $\text{KR}_\text{L}$ and NT0.5 in order to keep the figures clear. The corresponding x-values for other intersections can be observed from the figures.}. This indicates that the advantages and disadvantages of G-aware models over G-agnostic models can be better revealed by CPMs at different homophily levels than the baseline homophily metrics shown in Figure \ref{fig:synthetic_graphs_baseline_homo}.

In Figure \ref{fig:synthetic_graphs_baseline_homo}, we can see that the curves of node homophily (orange), class homophily (pink), generalized edge homophily (yellow) and adjusted homophily (blue) are almost linear increasing, which does not reflect the U-shaped performance curve of GNNs' performance. Although the curves for aggregation homophily (purple) and label informativeness (grey) have a rebound in low homophily area, they are unable to provide a suitable threshold value and fails to capture the intersection points.

Since the values of CPMs are either (very close to) 0 or (very close to) 1 and there do not exist enough intermediate values between 0 and 1, we do not plot the relationship between GNNs performance and CPMs as Figure 2 in \cite{luan2022revisiting}. 

\begin{figure}[htbp!]
    \centering
     {
     \subfloat[Cora]{
     \captionsetup{justification = centering}
     \includegraphics[width=0.5\textwidth]{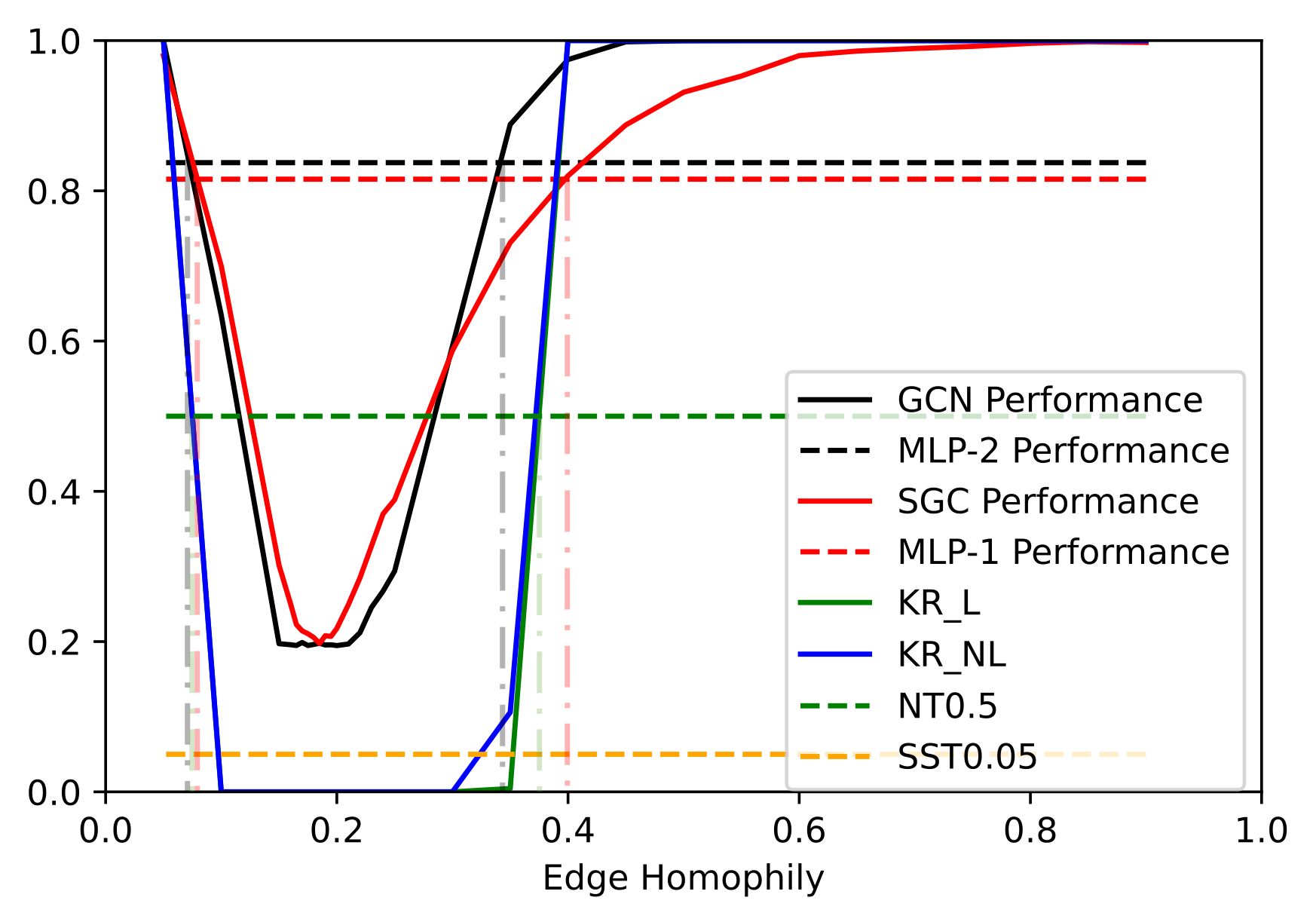}
     }
     \subfloat[CiteSeer]{
     \captionsetup{justification = centering}
     \includegraphics[width=0.5\textwidth]{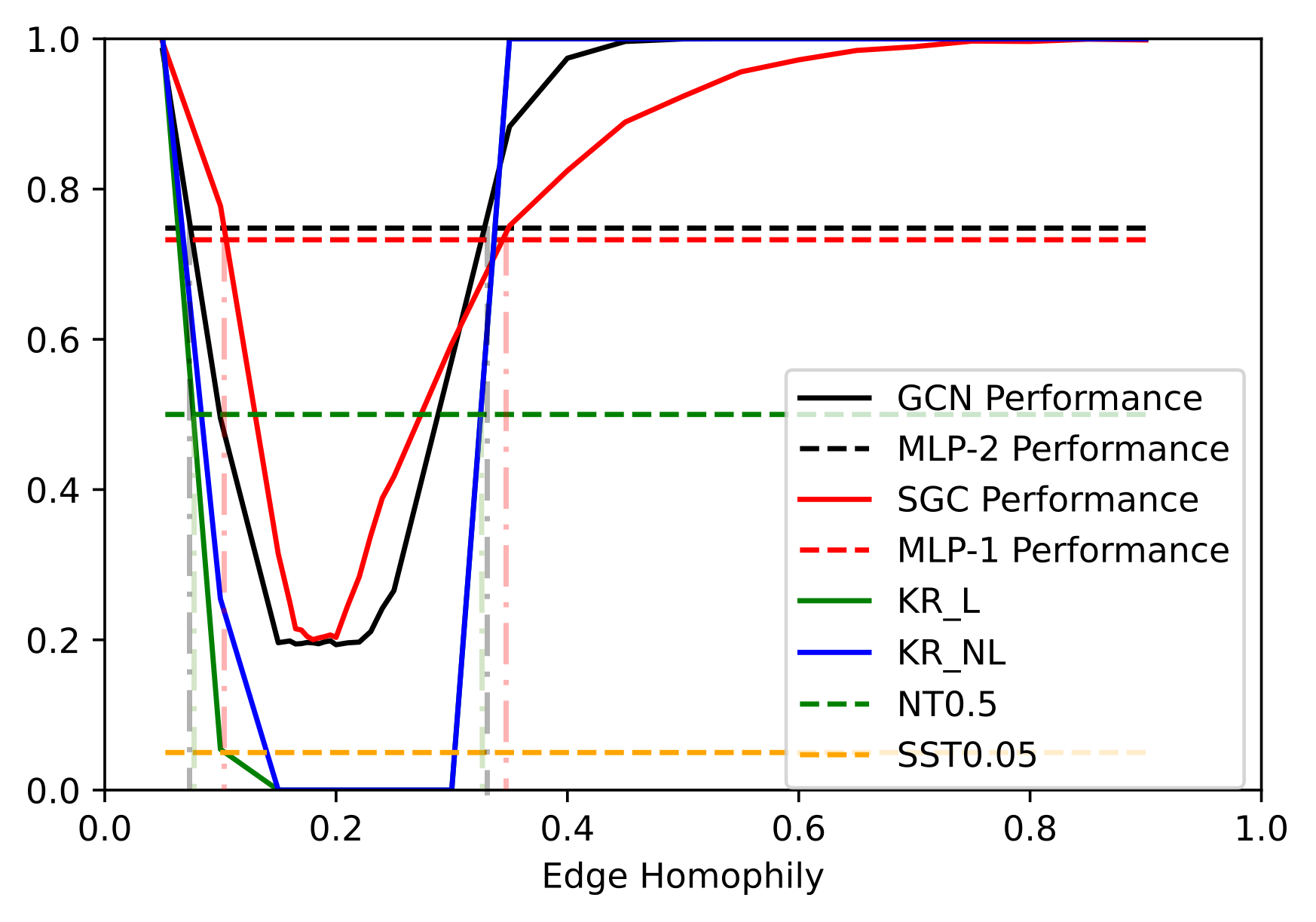}
     }\\
     \subfloat[PubMed]{
     \captionsetup{justification = centering}
     \includegraphics[width=0.5\textwidth]{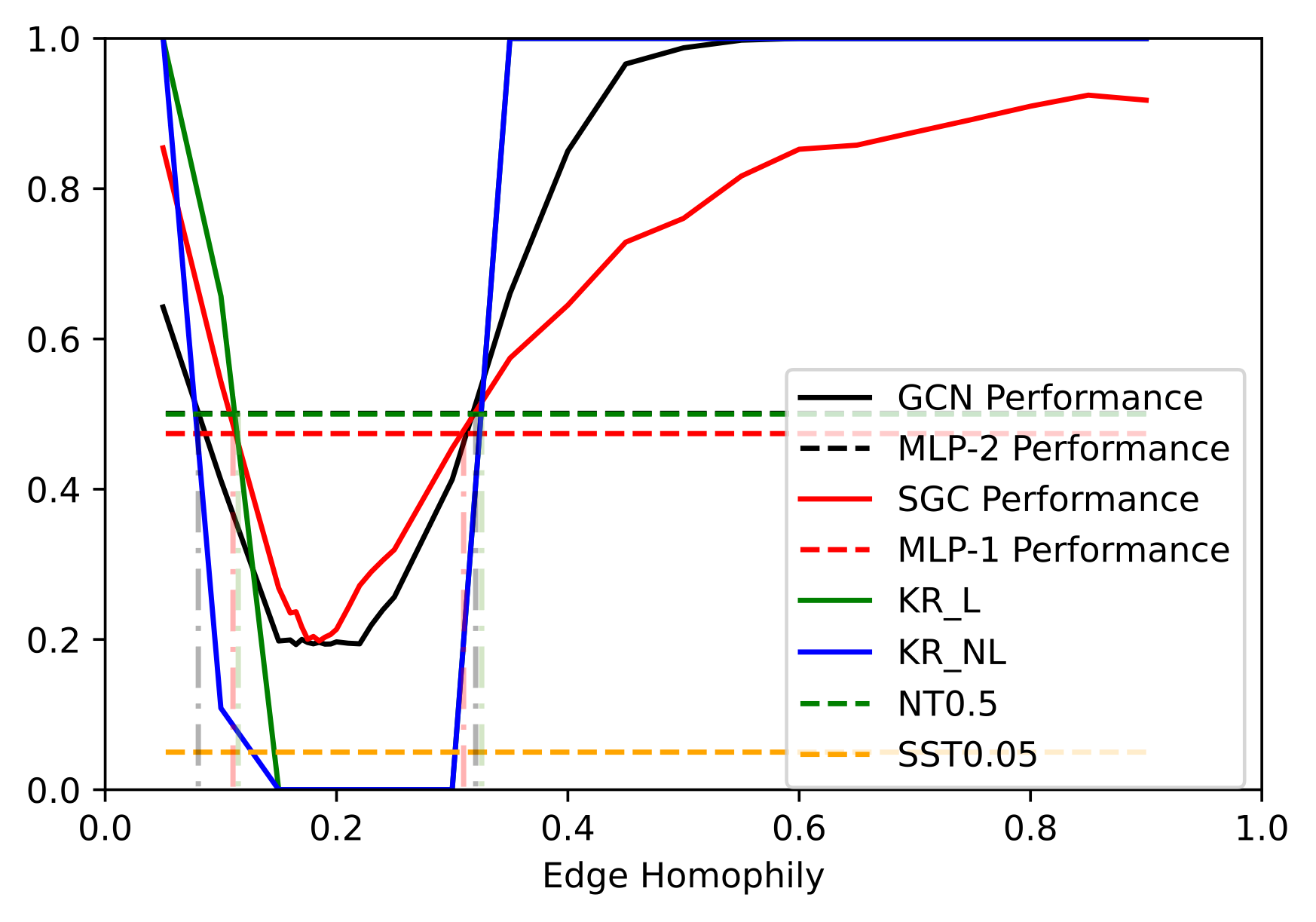}
     } 
     \subfloat[Chameleon]{
     \captionsetup{justification = centering}
     \includegraphics[width=0.5\textwidth]{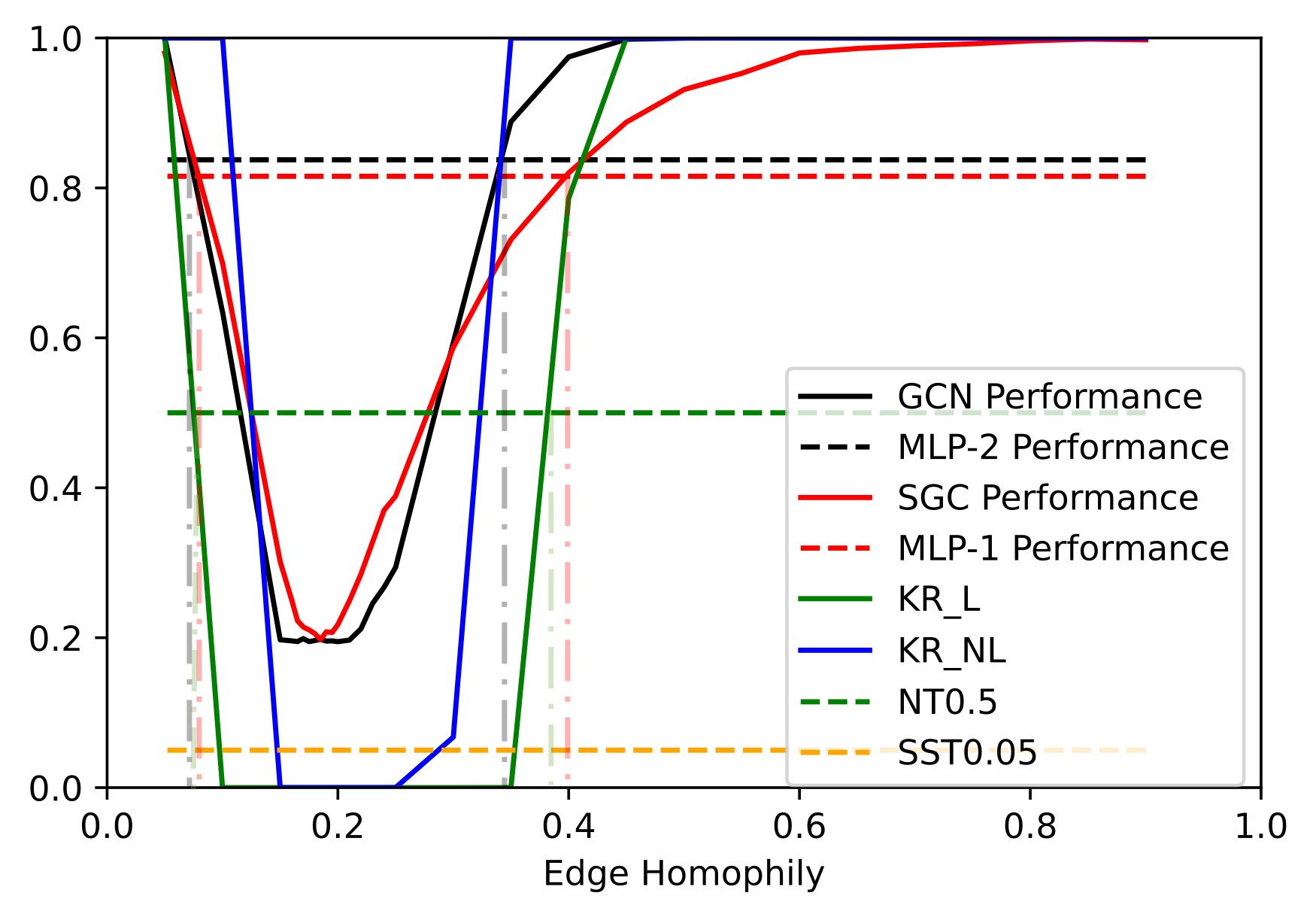}
     } 
     \\
     \subfloat[Squirrel]{
     \captionsetup{justification = centering}
     \includegraphics[width=0.5\textwidth]{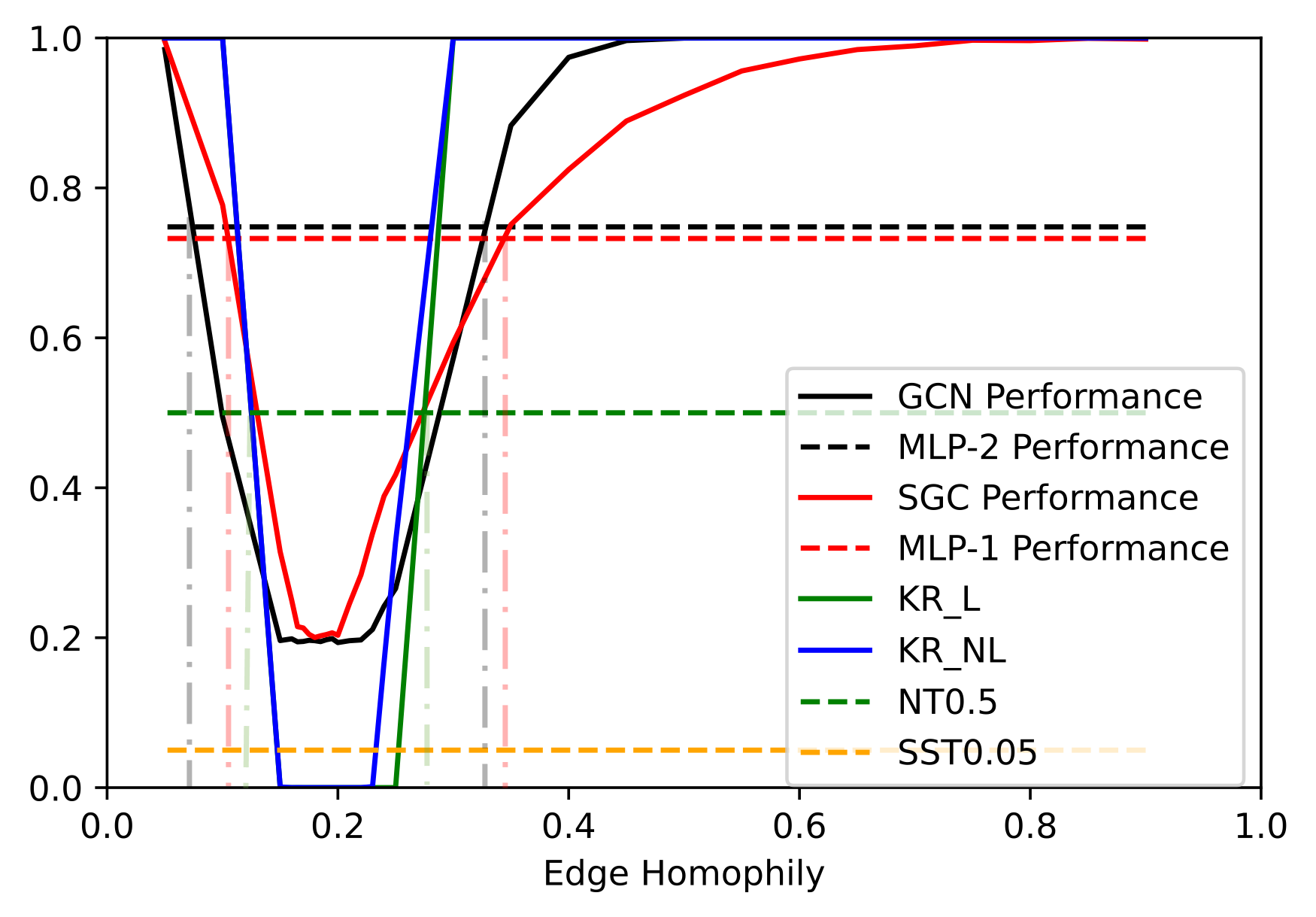}
     } 
     \subfloat[Film]{
     \captionsetup{justification = centering}
     \includegraphics[width=0.5\textwidth]{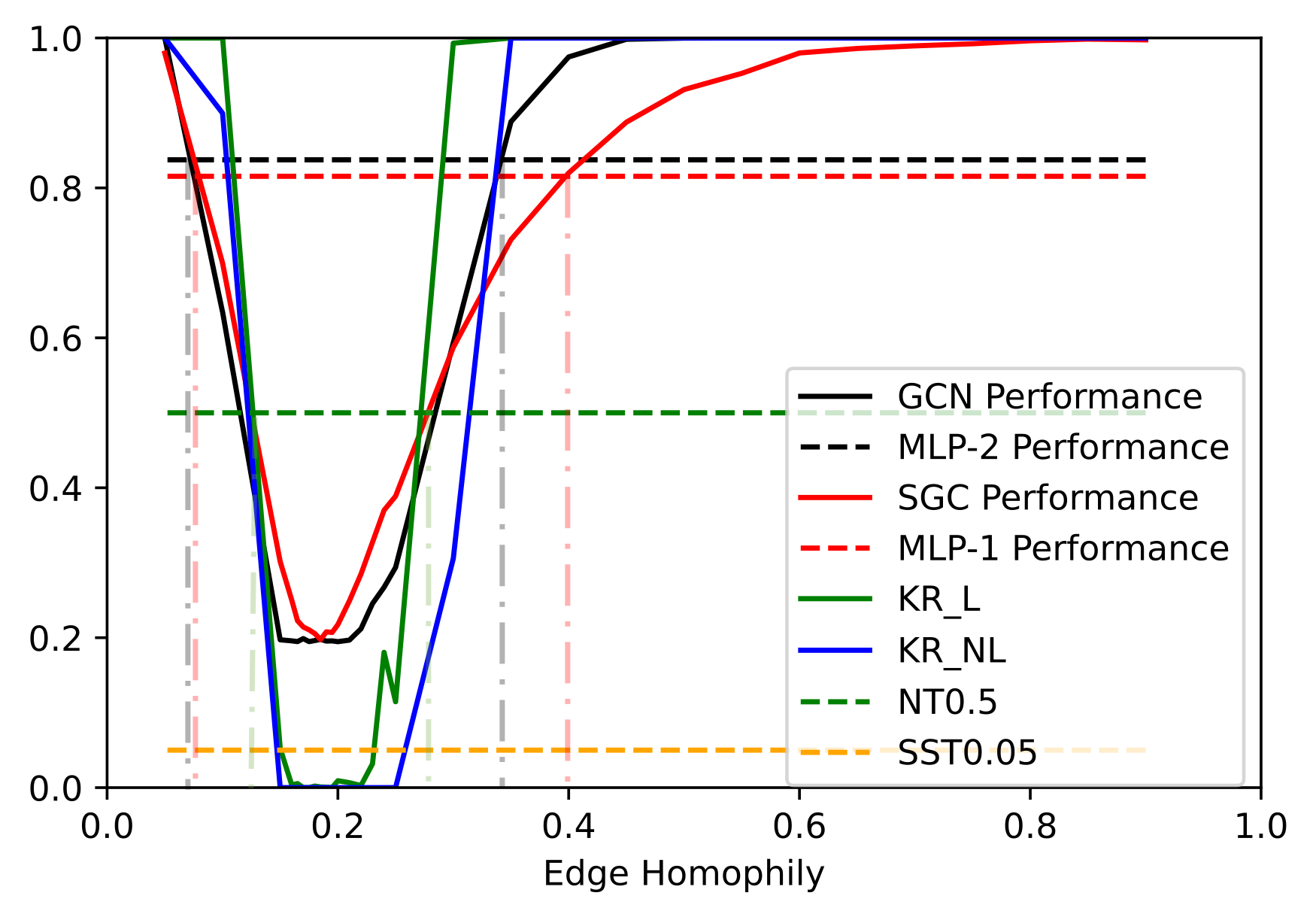}
     } 
     }
     \caption{Results and comparisons of $\text{KR}_\text{L}$ and $\text{KR}_\text{NL}$ on synthetic graphs}
     \label{fig:synthetic_graphs_cpms}
\end{figure}

\begin{figure}[htbp!]
    \centering
     {
     \subfloat[Cora]{
     \captionsetup{justification = centering}
     \includegraphics[width=0.5\textwidth]{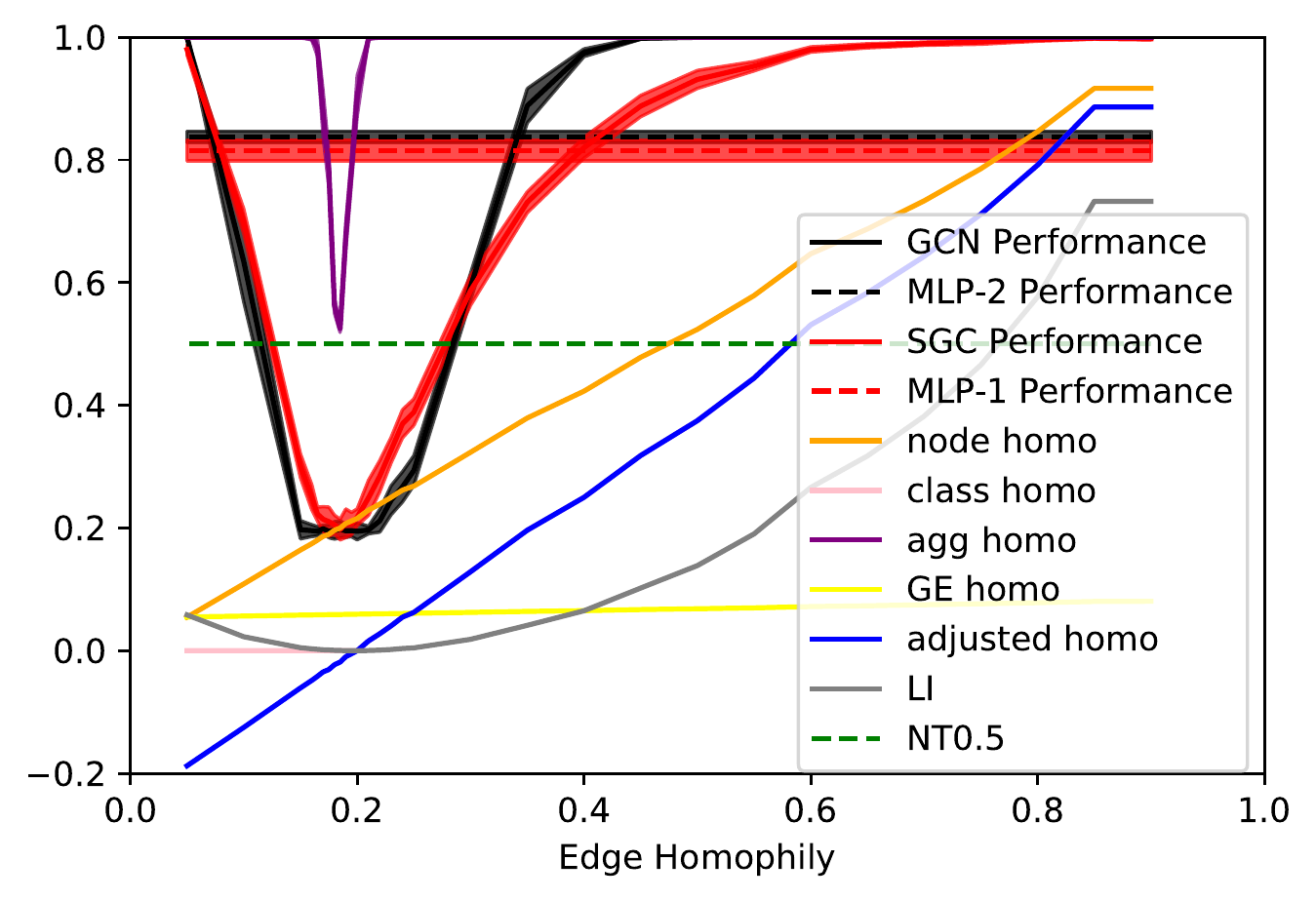}
     }
     \subfloat[CiteSeer]{
     \captionsetup{justification = centering}
     \includegraphics[width=0.5\textwidth]{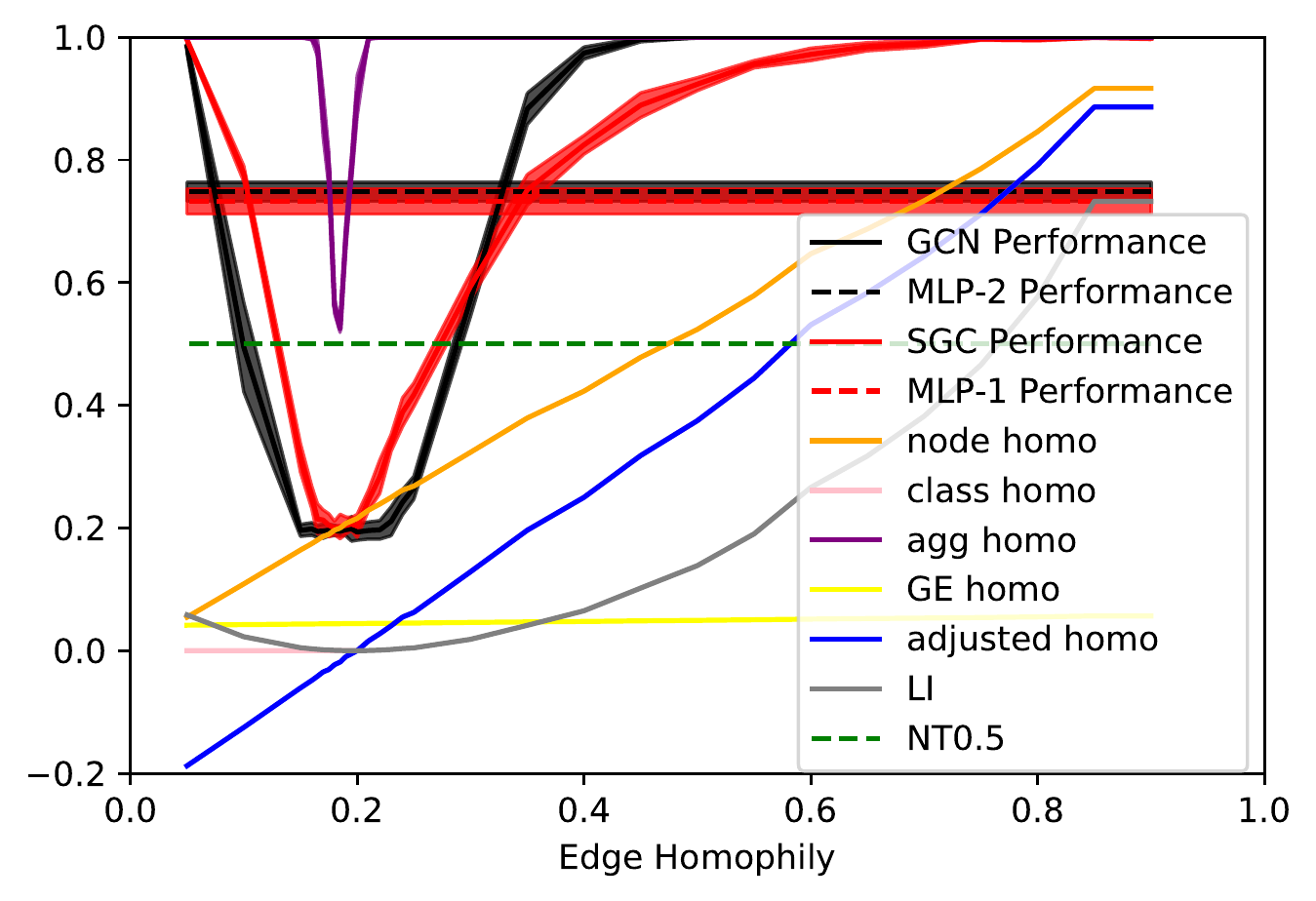}
     }\\
     \subfloat[PubMed]{
     \captionsetup{justification = centering}
     \includegraphics[width=0.5\textwidth]{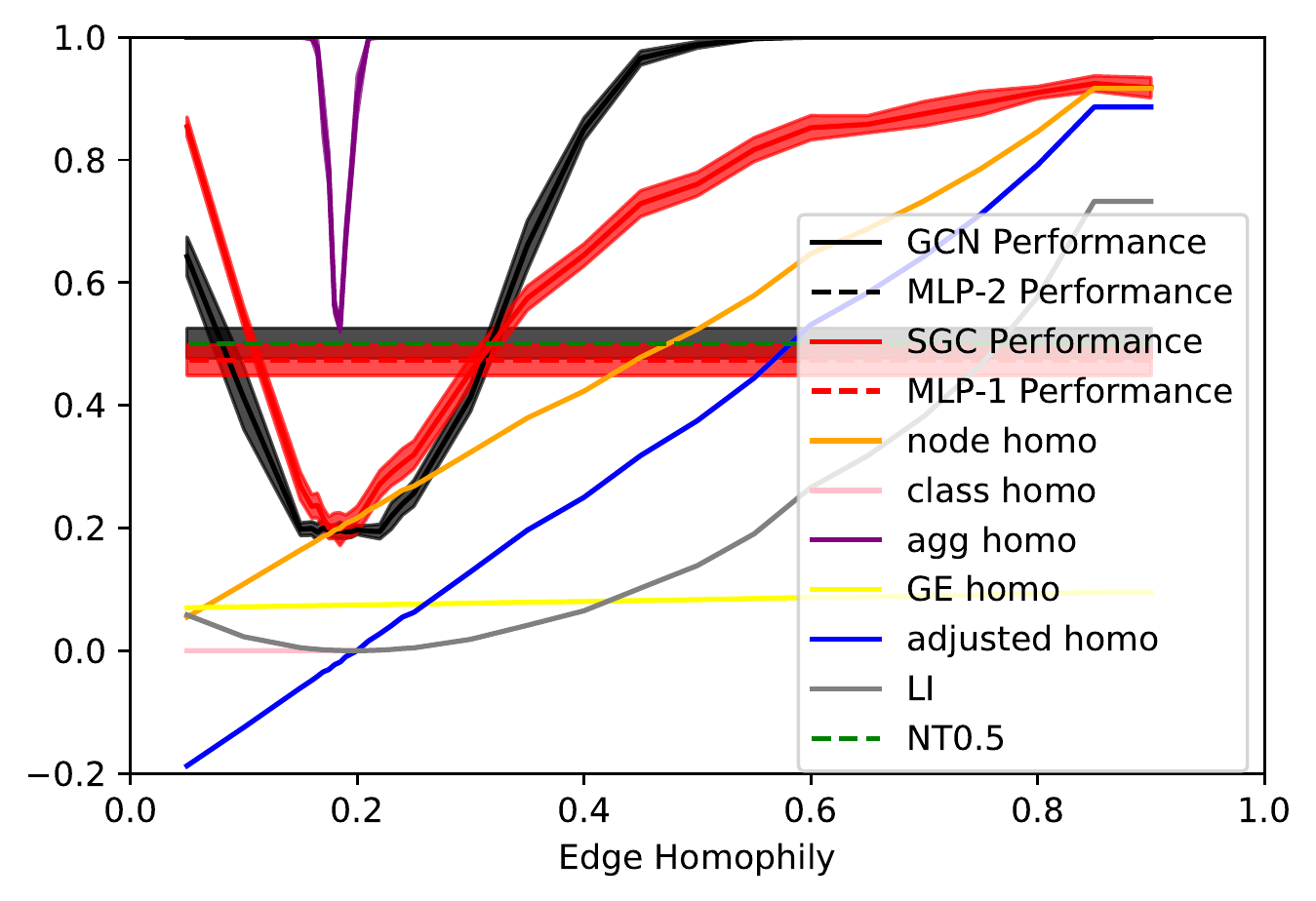}
     } 
     \subfloat[Chameleon]{
     \captionsetup{justification = centering}
     \includegraphics[width=0.5\textwidth]{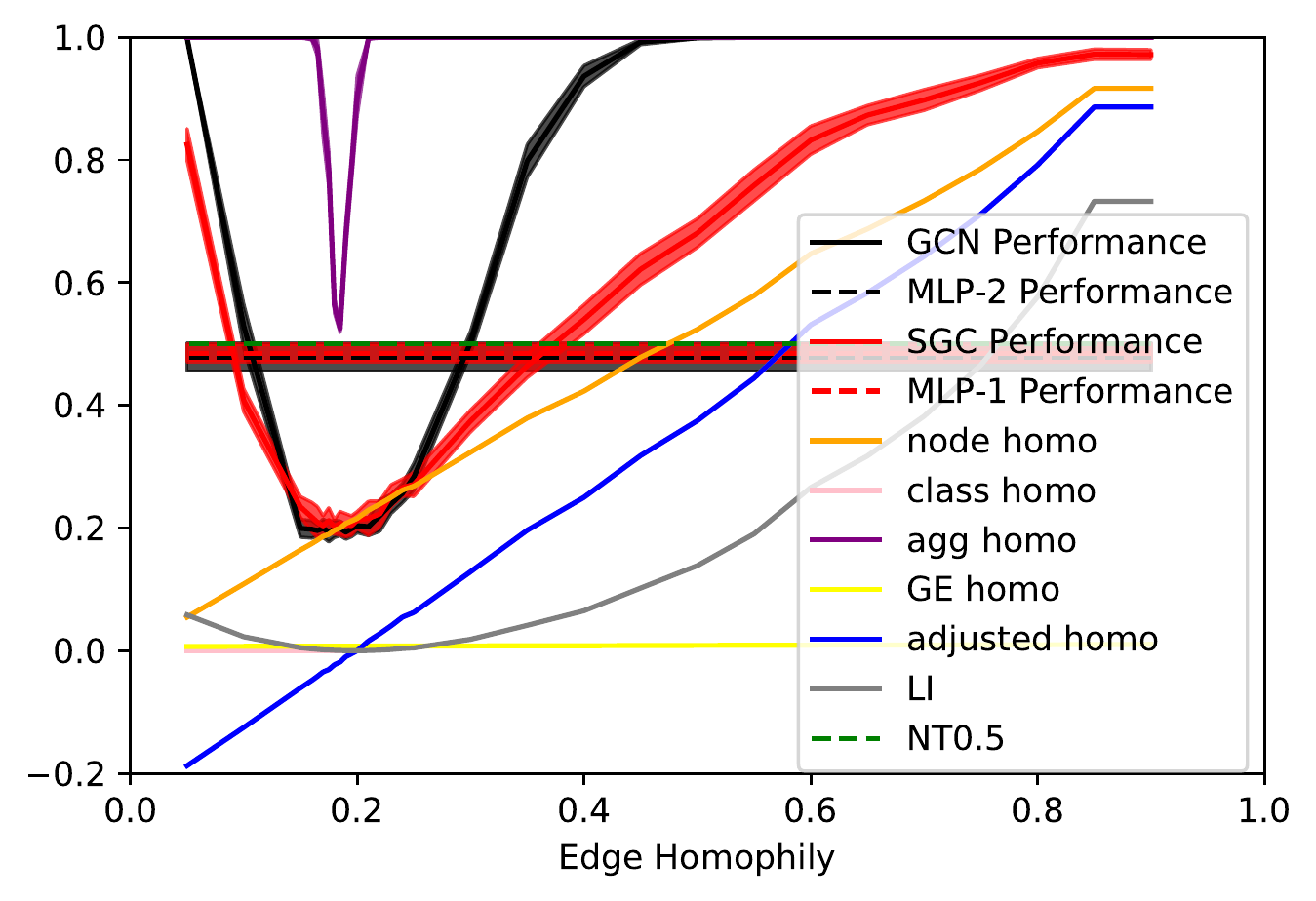}
     } 
     \\
     \subfloat[Squirrel]{
     \captionsetup{justification = centering}
     \includegraphics[width=0.5\textwidth]{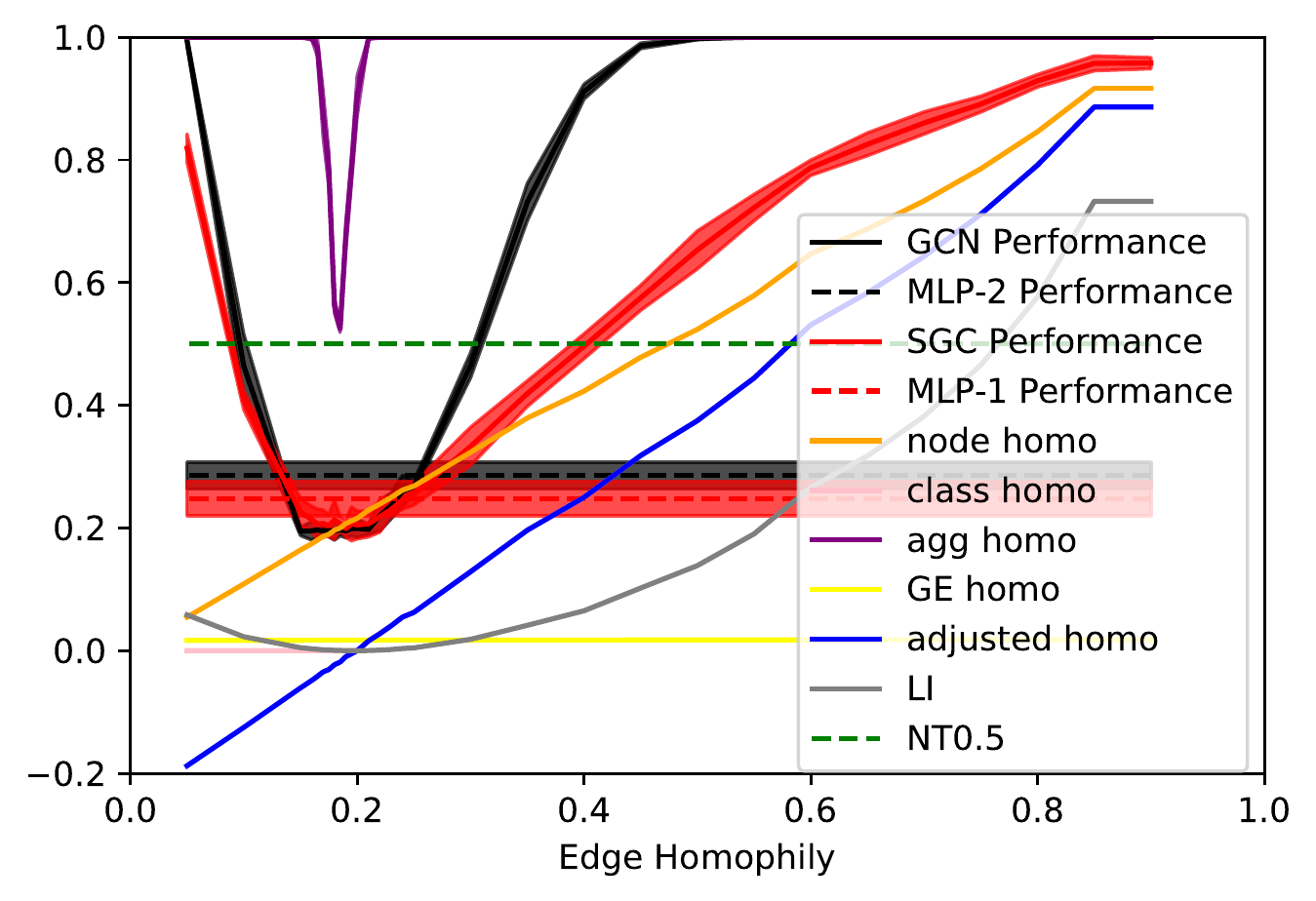}
     } 
     \subfloat[Film]{
     \captionsetup{justification = centering}
     \includegraphics[width=0.5\textwidth]{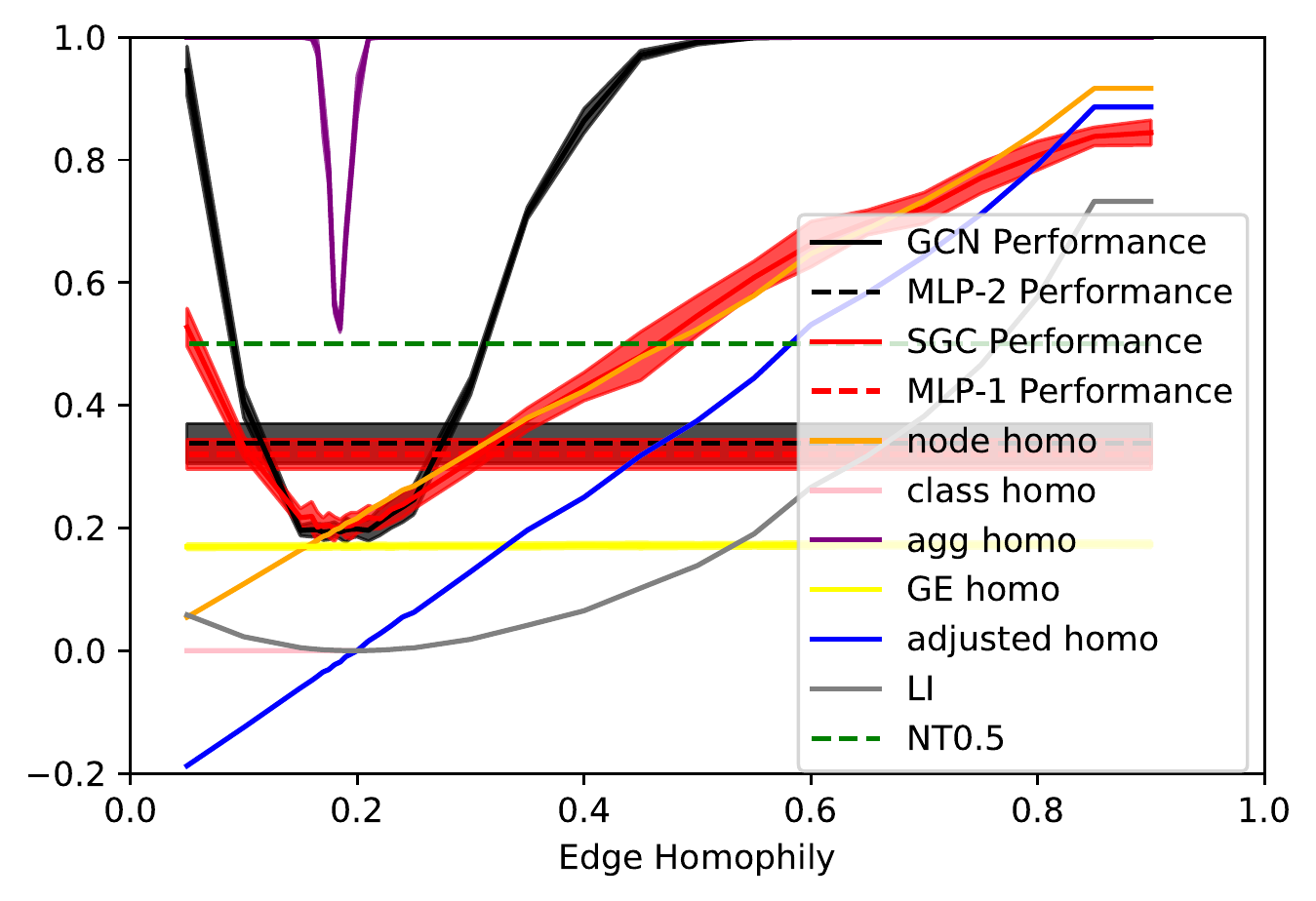}
     } 
     }
     \caption{Results on Synthetic Graphs}
     \label{fig:synthetic_graphs_baseline_homo}
\end{figure}

\end{document}